\documentclass[12pt,oneside]{amsart}
\usepackage[margin=1in]{geometry}
\usepackage{amsmath,amsthm,amssymb,amsfonts}
\usepackage{graphicx}
\usepackage{latexsym}
\usepackage[T2A]{fontenc}
\usepackage[utf8]{inputenc}
\usepackage[russian,english]{babel}
\usepackage{natbib}
\usepackage{ragged2e}
\usepackage{setspace}
\usepackage[table]{xcolor}
\usepackage{csquotes}
\usepackage{ulem}
\usepackage{lineno}
\usepackage{thmtools}

\usepackage[hyphens]{url}

\usepackage{hyperref}
\usepackage{diagbox}
\usepackage{multirow}
\usepackage{cleveref}
\usepackage{booktabs} 
\usepackage{enotez}

\let\footnote=\endnote 

\theoremstyle{definition}
 \newtheorem{theorem}{Theorem}
 \newtheorem{example}[theorem]{Example}
 \newtheorem{proposition}[theorem]{Proposition}
 \newtheorem{definition}[theorem]{Definition}
 \newtheorem{claim}[theorem]{Claim}
 
\newcommand{\floor}[1]{\lfloor #1\rfloor}

\newcommand{\ria}{\rightarrow}

\newcommand{\bs}[1]{\boldsymbol{#1}}
\newcommand{\EE}{\mathbb{E}}
\DeclareMathOperator*{\argmin}{arg\,min}
 
\begin{document}
\normalem

\title[]{How should we score athletes and candidates: geometric scoring rules}

\author{Aleksei Y. Kondratev}

\author{Egor Ianovski}

\author{Alexander S. Nesterov}

\thanks{akondratev@hse.ru, corresponding -- HSE University, Russia -- \url{http://orcid.org/0000-0002-8424-8198}.
\\
eianovski@hse.ru --- HSE University, Russia --- \url{https://orcid.org/0000-0001-9411-2529}.
\\
asnesterov@hse.ru --- HSE University, Russia --- \url{http://orcid.org/0000-0002-9143-2938}
\\
\phantom{ } We thank Herv{\'e} Moulin, Elena Yanovskaya, Artem Baklanov, Constantine Sorokin and other our colleagues from the International Laboratory of Game Theory and Decision Making for their suggestions and support. We thank Dominik Peters for providing us with the reference to \citet{Morkeliunas82}. We also are grateful to Josep Freixas for providing us with the MotoGP 1999 example during the 15th European Meeting on Game Theory (SING15). 
}


\begin{abstract}


Scoring rules are widely used to rank athletes in sports and candidates in elections. Each position in each individual ranking is worth a certain number of points; the total sum of points determines the aggregate ranking. The question is how to choose a scoring rule for a specific application. First, we derive a one-parameter family with geometric scores which satisfies two principles of independence: once an extremely strong or weak candidate is removed, the aggregate ranking ought to remain intact. This family includes Borda count, generalised plurality (medal count), and generalised antiplurality (threshold rule) as edge cases, and we find which additional axioms characterise these rules. Second, we introduce a one-parameter family with optimal scores: the athletes should be ranked according to their expected overall quality. Finally, using historical data from biathlon, golf, and athletics we demonstrate how the geometric and optimal scores can simplify the selection of suitable scoring rules, show that these scores closely resemble the actual scores used by the organisers, and provide an explanation for empirical phenomena observed in golf tournaments. We see that geometric scores approximate the optimal scores well in events where the distribution of athletes’ performances is roughly uniform.

\medskip
\bigskip
\bigskip


\medskip
\medskip







\end{abstract}

\maketitle

\thispagestyle{empty}

\newpage

\thispagestyle{empty}


\section{Introduction}


Many sports competitions consist of a series of races during a season. At the end of each race the athletes are ranked in the order they finished, and assigned a number of points based on their position in that order. The scores the athletes received are summed across the races, and the one with the highest total score is declared the winner -- a procedure that is known as a scoring rule.

Scoring rules are ubiquitous in sporting events, contests to evince merit, elections in professional societies, and occasionally at the national level (Nauru, Kiribati). Besides being intuitive and easy to use, they have desirable axiomatic features which, as we shall argue, make them uniquely suitable for this purpose.

However, in order to use a scoring rule, one must first choose a vector of scores; and the choice is anything but simple. The International Biathlon Union (IBU) World Cup uses the scoring vector in \autoref{table:Glazyrina}, containing 40 non-zero scores. Scaling the scores or changing the zero point will not change the order produced by the scoring rule, which means the organiser has 39 degrees of freedom in selecting a vector like this one. The prospect of picking an optimal vector with respect to whatever criteria is daunting, but the choice is important -- the vector may play as great a role in determining the winner as the performance of the athletes themselves.

In this paper we propose two ways to reduce the problem to the choice of a single parameter. In Sections \ref{sec:gsr} and \ref{sec:newcharacterisations} we take an axiomatic approach, motivated by reducing the risk of the final ranking changing by the addition or removal of spoilers. This leads to the geometric family of scoring rules governed by parameter $p$, where the $j$th position is worth $p$ times as much as the $(j+1)$th position. In Sections \ref{sec:optimal} and \ref{sec:empirical} we introduce optimal scoring rules, which maximise the expected quality of the winning athlete based on empirical data. The two approaches, in general, yield different families of rules, but optimal scores converge to the geometric in the case of a uniform distribution of athlete performance. Interestingly, we find that the scores actually used in biathlon, golf, and athletics closely resemble the optimal scoring rules.

\section{A ranking paradox}

The Women's Pursuit category of the 2014/2015 IBU Biathlon World Cup consisted of seven races (\autoref{table:Glazyrina}). Kaisa M\"ak\"ar\"ainen came first with two first place finishes, two second, a third, a fourth, and a twelfth, for a total score of 348 points. Second was Darya Domracheva, with four first place finishes, one fourth, a seventh, and a thirteenth, for a total score of 347. In tenth place was Ekaterina Glazyrina, well out of the running with 190 points.

\begin{table}[ht!]
\centering
\caption{2014/15 Biathlon -- Women's Pursuit: scoring system and event results}
\begin{tabular}{lccccccccccccccc}
\toprule
Position&1&2&3&4&5&6&7&8&9&10&11&12&13&$\cdots$&40\\
\hline
Points&60&54&48&43&40&38&36&34&32&31&30&29&28&$\cdots$&1\\
\bottomrule
\end{tabular}
\begin{tabular}{lcccccccr}
\\
\toprule
Athlete&\multicolumn{7}{c}{Event number: points}&Total\\
&1&2&3&4&5&6&7&score\\
\hline
M\"ak\"ar\"ainen&60&60&54&48&54&29&43&348\\
Domracheva&43&28&60&60&60&36&60&347\\
$\cdots$&&&&&&&&\\
Glazyrina&32&54&10&26&38&20&10&190\\
\bottomrule
\end{tabular}
\label{table:Glazyrina}
\end{table}

Four years later, Glazyrina was disqualified for doping violations, and all her results from 2013 onwards were annulled. This bumped Domracheva's thirteenth place finish in race two into a twelfth, and her total score to 348. The number of first place finishes is used as a tie breaker, and in March 2019 the official results implied that M\"ak\"ar\"ainen will be stripped of the trophy in favour of Domracheva. Because the tenth place competitor was disqualified for doping four years after the fact.\footnote{After the disqualification of Glazyrina, the IBU eventually decided to award the 2014/15 Pursuit Globe to both Domracheva and M\"ak\"ar\"ainen (\url{https://biathlonresults.com}). While this keeps both athletes happy, it is clearly an ad-hoc solution. In November 2021, another biathlete (Olga Zaitseva) was disqualified. M\"ak\"ar\"ainen is again the victim, and after the score recount would have lost the 2013/14 Big Crystal Globe to Tora Berger. Again, the IBU decided to recognise both Berger and M\"ak\"ar\"ainen as the total score winners (\url{https://web.archive.org/web/20220114020227/https://www.biathlonworld.com/news/berger-makarainen-total-score/7JONQpUehHSZGX6ucrwn0x}). Perhaps this approach will become the new normal. We refer to \citet{Wright14} and \citet{KendallLenten17}, for surveys of where other sporting rules led to unintended consequences.}

Clearly there is something unsatisfying about this. We would hope that the relative ranking of M\"ak\"ar\"ainen and Domracheva depends solely on the relative performance of the two athletes, and not on whether or not a third party was convicted of doping, especially if said third party was not a serious contender for the title.

Classical results from social choice theory mean this goal is impossible: the ranking of athletes will never be fully independent of the addition or deletion of third parties. However, the degree of this vulnerability depends on the vector of scores used to rank the athletes. We shall see that by introducing two extremely weak independence axioms, we narrow the range of scoring rules to a one-parameter family -- the geometric scoring rules. If the event organiser finds these axioms convincing, then this reduces the problem of choosing a scoring rule to the choice of a single parameter.

\subsection*{The science of impossibility}

Suppose we have a set of $m$ athletes and a profile of $n$ races, $R_1,\dots,R_n$, with $R_i$ being the order in which the athletes finished race~$i$. What we are after is a procedure which will map the $n$ rankings into a single ranking for the entire competition, $R$: a mapping $(R_1,\dots,R_n)\mapsto R$. The result of this ranking rule must reflect the results of the individual races in some way. A minimal condition is \textbf{unanimity} – if an athlete finishes first in every race, we should expect this athlete to rank first in $R$. Motivated by the scenario above, we also want the relative ranking of athletes $a$ and $b$ in the end result to depend only on the relative ranking of $a$ and $b$ in the individual races, a condition known as the \textbf{independence of irrelevant alternatives}. Here we hit the most famous result in social choice theory -- Arrow's result that the only ordinal procedure that meets our criteria is dictatorship \citep[p.~52]{Arrow50,CampbellKelly02}.\footnote{If we want not to determine an aggregate ranking but only to select a single winner, then unanimity and independence of irrelevant alternatives still lead to dictatorship -- it follows from a more general result of \citet[theorem~1]{DuttaJacksonLeBreton01}.}

The characteristic feature of dictatorship is that all decisions stem from a single individual.
In the case of a sporting competition, this could be the case where the first $n-1$ races are treated as warm-ups or friendly races, and only the finals $R_n$ contributes to the final ranking $R$. This is not necessarily an absurd ranking system, but it will not do if we want to keep viewers interested over the course of the championship, rather than just the finals. We need to relax independence.\footnote{The approach in this paper is axiomatic. We want a ranking rule that always satisfies a certain notion of independence, and mathematical impossibilities force us to relax the notion until it is weak enough to be compatible with other desirable properties. This is not the only way to approach the problem. For example, we might accept that we cannot have independence all of the time, and instead look for a rule that will satisfy independence most of the time. \citet[theorem~1]{Gehrlein1982}, in that framework, show that under the impartial culture assumption (i.e. when all rankings are equally probable), the scoring rule most likely to select the same winner before and after a random candidate is deleted is Borda.}

A weaker independence condition we may consider is \textbf{independence of winners/losers}.\footnote{Independence of winners/losers is also known as local stability \citep{Young88} and local independence of irrelevant alternatives. Independence of losers is also known as independence of bottom alternatives \citep{FreemanBrillConitzer14}.} As the name suggests, this is the condition that if we disqualify the top or the bottom athlete in the final ranking $R$, the remainder of the ranking remains unchanged. This could be a pressing issue if the winner is accused of doping, and a rule that satisfies this condition will guarantee that the cup is given to the runner up without requiring a retallying of the total scores. In the case of the loser, there is the additional concern that it is a lot easier to add a loser to a race than a winner, and if the authors were to take their skis off the shelf and lose ingloriously in the next Biathlon, one would hope that the standing of the real competitors would remain unaffected.

It turns out that, given some standard assumptions, there is a unique rule that is independent of winners and losers -- the Kemeny rule \citep[footnote~18]{Young88}. The procedure amounts to choosing a ranking $R$ that minimises the sum of the Kendall tau distance from $R$ to the individual races. This is one of its disadvantages -- it is a stretch to expect a sports enthusiast to plot race results in the space of linear orders and compute the central point. For viewers, the results might as well come from a black box. What is worse, it is a difficult procedure computationally \citep[theorem~2]{BartholdiToveyTrick89}, so even working out the winner may not be feasible. But perhaps most damning of all is that it violates a property known as \textbf{electoral consistency} (formal definition can be found in \autoref{app:complete}). In Biathlon, every race falls into one of four categories (sprint, pursuit, individual, mass start). At the end of the championship a winner is selected for each category, as well as an overall winner. It would be strange if a biathlete were to win in every category but lose the overall title -- but that is a possibility under Kemeny.

In fact, the only ordinal procedure that guarantees electoral consistency is a generalised scoring rule \citep{Smith73,Young74,Young75} -- every athlete is awarded a number of points based on their position in a race, and the athletes are ranked based on total points; in the case of ties, another scoring rule can be used to break them. It seems there is no alternative to the rules actually used in biathlon (IBU World Cup), auto racing (Formula One World Championship), cycling (Tour de France green jersey), golf (Professional Golfers' Association Tour, stylised in all capital letters as PGA TOUR by its officials), skiing (International Ski and Snowboard Federation World Cup), athletics (International Association of Athletics Federations Diamond League), and other events in this format -- but that is not necessarily a bad thing. Scoring rules are easy to compute and understand, and every additional result contributes to the overall ranking in a predictable way, all of which is very desirable for a sporting event.

We have seen that neither of the independence notions we have defined so far can apply here, but how bad can the situation get? The answer is, as bad as possible. A result of \citet[theorem~2]{Fishburn81DAM} shows that if the scores awarded for positions are monotone and decreasing, it is possible to construct a sequence of race results such that if one athlete is removed, the remaining order is not only changed, but inverted. So while M\"ak\"ar\"ainen may not be pleased with the current turn of events, there is a possible biathlon where after the disqualification of Glazyrina, M\"ak\"ar\"ainen finished last, and Domracheva second to last. It is interesting to speculate whether the competition authorities would have had the resolve to carry through such a reordering if it had taken place.

\section{Geometric scoring rules}\label{sec:gsr}

In order to motivate our final notion of independence, let us first consider why the results of the biathlon may not be as paradoxical as they appear at first glance. Note that removing Glazyrina from the ranking in Table~\ref{table:Glazyrina} changed the total score of Domracheva but not of M\"ak\"ar\"ainen. This is because M\"ak\"ar\"ainen was unambiguously better than Glazyrina, finishing ahead of her in every race, while Domracheva was beaten by Glazyrina in race~2. As such, the athletes' performance vis-\`a-vis Glazyrina served as a measuring stick, allowing us to conclude that M\"ak\"ar\"ainen was just that little bit better. Once Glazyrina is removed, however, the edge M\"ak\"ar\"ainen had is lost.

So suppose then that the removed athlete is symmetric in her performance with respect to all the others. In other words, she either came last in every race, and is thus a unanimous loser, or came first, and is a unanimous winner. Surely disqualifying such an athlete cannot change the final outcome? Why, yes it can.

The results for the Women's Individual category of the 2013/14 IBU Biathlon World Cup are given in the left panel of Table~\ref{table:Soukalova}. The category consists of two races, and Gabriela Soukalov\'a came first in both, and is thus a unanimous winner, followed by Darya Domracheva, Anastasiya Kuzmina, Nadezhda Skardino, and Franziska Hildebrand. However, in the hypothetical event of Soukalov\'a being disqualified the result is different: the recalculated total scores are in the right panel of Table~\ref{table:Soukalova}. Domracheva takes gold and Kuzmina silver as expected, but Hildebrand passes Skardino to take the bronze.

\begin{table}
\centering
\caption{2013/14 IBU Biathlon World Cup -- Women's Individual}
\begin{tabular}{lccr}
\toprule
Athlete & \multicolumn{2}{c}{Event} & Total \\
 & $1$ & $2$ & score \\
\hline
Soukalov\'a & 60\small{/1} & 60\small{/1} & 120 \\
Domracheva & 38\small{/6} & 54\small{/2} & 92 \\
Kuzmina & 54\small{/2} & 30\small{/11} & 84\\
Skardino & 36\small{/7} & 36\small{/7} & 72\\
Hildebrand & 28\small{/13} & 43\small{/4} & 71\\
\bottomrule
\end{tabular}
\quad
\begin{tabular}{lccr}
\toprule
Athlete & \multicolumn{2}{c}{Event} & Total \\
 & $1$ & $2$ & score \\
\hline
\sout{Soukalov\'a} & \sout{60\small{/1}} & \sout{60\small{/1}} & \sout{120} \\
Domracheva & 40\small{/5} & 60\small{/1} & 100 \\
Kuzmina & 60\small{/1} & 31\small{/10} & 91\\
Hildebrand & 29\small{/12} & 48\small{/3} & 77\\
Skardino & 38\small{/6} & 38\small{/6} & 76\\
\bottomrule
\end{tabular}
\label{table:Soukalova}
\vspace{0.2cm}
\justify 
\footnotesize{\textit{Notes}: The left panel presents the official points/position; the right panel presents the points/position after a hypothetical disqualification of Soukalov\'a. The total scores given in the table are before the disqualification of another athlete, Iourieva, that occurred a few months after the race. With the most recent total scores we would still observe Hildebrand overtaking Skardino, but we would have to resort to tie-breaking to do it.}
\end{table}

In contrast to the previous paradoxes, this is one we can do something about. By picking the right set of scores we can ensure that the unanimous loser will come last, the unanimous winner first, and dropping either will leave the remaining order unchanged.

Let us formalise our key notions. 

\begin{definition}
Let $M$ be the number of potential athletes (either finite or $M=\infty)$. For every number $m$ of athletes, $m\leq M$, a \textbf{positional scoring rule}, or a \textbf{scoring rule} for short, is defined by a sequence of $m$ real numbers $s_1^m,\ldots,s_m^m$. For a profile, an athlete receives a score $s_j^m$ for position~$j$ in an individual ranking. The sum of scores across all rankings gives the athlete's total score. The total scores determine the overall ranking: athletes with higher total scores are ranked higher, athletes with equal total scores are ranked equally.

For example, \textbf{plurality} is the scoring rule with scores $(1,0,\ldots,0)$ for each $m$, while \textbf{antiplurality} corresponds to scores $(1,\ldots,1,0)$, and \textbf{Borda} to $(m-1,m-2,\ldots,1,0)$.
\end{definition}

Of course, it is possible that two athletes attain the same total score, and are thus tied in the final ranking. In general, this problem is unavoidable -- if two athletes perform completely symmetrically vis-\`a-vis each other, no reasonable procedure can distinguish between them. However, if things are not quite so extreme, ties can be broken via a secondary procedure -- for example, in the case of the IBU we have seen that ties are broken with the number of first place finishes. This gives rise to the notion of a generalised scoring rule, where ties in the initial ranking are broken with a secondary sequence of scores, any remaining ties with a third, and so on.

\clearpage

\begin{definition} Let $M$ be the number of potential athletes (either finite or $M=\infty)$. For every number $m$ of athletes, $m\leq M$, a \textbf{generalised scoring rule} is defined by $\overline{r}(m)$ sequences of $m$ real numbers $s_1^{m,r},\ldots,s_m^{m,r}$ -- one sequence for each tie-breaking round $r=1,\ldots,\overline{r}(m)$. For a profile, in  round $r$, an athlete $a$ receives score $s_j^{m,r}$ for position~$j$ in an individual ranking. The total sum of scores gives a total score $S^r_a$ of athlete~$a$. The total scores determine the overall ranking lexicographically: $a$ is ranked higher than $b$ if $S^r_a>S^r_b$ for some round $r$ and $S^l_a=S^l_b$ for all $l<r$. Athletes $a$ and $b$ are equally ranked if $S^r_a=S^r_b$ for all rounds $r\leq \overline{r}(m)$.

For example, for each $m$ athletes, \textbf{generalised plurality} has $m-1$ rounds with scores $(\overbrace{1,\ldots,1}^r, \overbrace{0, \ldots, 0}^{m-r})$ in round $r$. \textbf{Generalised antiplurality} has $(\overbrace{1,\ldots,1}^{m-r}, \overbrace{0,\ldots, 0}^r)$.\footnote{Generalised plurality is also known as lexicographic ranking or the Olympic medal count \citep{ChurilovFlitman06}; generalised antiplurality -- as the threshold rule \citep{AleskerovChistyakovKalyagin10}.}
\end{definition}

Note that by definition a scoring rule is a generalised scoring rule with only one tie-breaking round.

\begin{definition}
An athlete is a \textbf{unanimous loser} if the athlete is ranked last in every race. A generalised scoring rule satisfies \textbf{independence of unanimous losers} if it ranks the unanimous loser last in the overall ranking, and removing the unanimous loser from every race leaves the overall ranking of the other athletes unchanged.

Symmetrically, an athlete is a \textbf{unanimous winner} if the athlete is ranked first in every race. A generalised scoring rule satisfies \textbf{independence of unanimous winners} if it ranks the unanimous winner is ranked first in the overall ranking, and removing the unanimous winner from every race leaves the overall ranking of the other athletes unchanged.
\end{definition}

Observe that the order produced by a scoring rule is invariant under scaling and translation, e.g. the scores 4, 3, 2, 1 produce the same order as 8, 6, 4, 2 or 5, 4, 3, 2. We will thus say that scores $s_1,\dots,s_m$ and $t_1,\dots,t_m$ are \textbf{affinely equivalent} if there exists an $\alpha>0$ and a $\beta$ such that $s_j = \alpha t_j + \beta$.

The intuition behind the following result is clear: $t_1,\dots,t_k$ produces the same ranking of the first/last $k$ athletes, if and only if it is affinely equivalent to the first/last $k$ scores in the original ranking system. The proof of the theorem, and all subsequent theorems, can be found in \autoref{app:proof}, and full characterisations as ranking rules in \autoref{app:complete}.

\begin{restatable}{proposition}{affinelyequivalent}\label{prop:affinelyequivalent}
A scoring rule satisfies independence of unanimous losers if and only if $s^m_1>\ldots>s^m_m$, and the scores for $k$ athletes, $s^k_1,\dots,s^k_k$, are affinely equivalent to the first $k$ scores for $m$ athletes, $s^m_1,\dots,s^m_k$, for all $k<m\leq M$.

A scoring rule satisfies independence of unanimous winners if and only if $s^m_1>\ldots>s^m_m$ and the scores for $k$ athletes, $s^k_1,\dots,s^k_k$, are affinely equivalent to the last $k$ scores for $m$ athletes, $s^m_{m-k+1},\dots,s^m_m$, for all $k<m\leq M$.
\end{restatable}

Now we see why the biathlon scores are vulnerable to dropping unanimous winners but not unanimous losers -- since the scores for a smaller number of athletes are obtained by trimming the full list, every subsequence of the list of scores is indeed equivalent to itself. However if we drop the winner, then the subsequence 60, 54, 48, $\dots$, is certainly not affinely equivalent to 54, 48, 43, $\dots$.

What happens when we combine the two conditions? The property of affine equivalence is clearly an equivalence relation, so if the scores for $k$ athletes are affinely equivalent to the first $k$ scores for $m$ candidates and the last $k$ scores for $m$ candidates, then the first and last $k$ scores for $m$ candidates must be affinely equivalent to each other, and in particular the scores $s_1^m,\dots,s_{m-1}^m$ must be affinely equivalent to $s_2^m,\dots,s_{m}^m$. Given that we need not distinguish scores up to scaling and translation we can assume that the score for the last place, $s_m^m$, is zero, and $s_{m-1}^m=1$. The third-to-last athlete must then get a larger number of points than 1, say $s_{m-2}^m=1+p$. Now we have a sequence $0,1,1+p$, and we know it must be affinely equivalent to the sequence $1,1+p, s_{m-3}^m$. Since the only way to obtain the second sequence from the first is to scale by $p$ and add 1, it follows that $s_{m-3}^m=1+p+p^2$ and so on. Clearly if $p=1$ this sequence is just Borda, and some algebraic manipulation gives us a formula of $s_j^m=(p^{m-j}-1)/(p-1)$ for the $j$th position in the general case. This gives us the following family of scoring rules consisting of the geometric, arithmetic, and inverse geometric sequences.\footnote{Despite their natural formulation, geometric scoring rules have received very little attention to date. Recently we have discovered the characterisation of \citet[theorem~4.3]{FineFine74b}, which is almost identical to our \autoref{thm:geometricrules}, but it does not appear that anyone noticed this in the subsequent literature. The independence axioms of \citet{FineFine74b} do not require the ranking of a unanimous winner first and a unanimous loser last, as our axioms do, but instead their definition of scoring rules directly requires non-decreasing scores. 

\citet{Phillips2014} independently used geometric sequences to approximate Formula One World Championship scores, and Laplace suggested scoring with the sequence $2^{m-1},2^{m-2},\dots,1$ \citep[p. 261--263]{Daunou1995}, i.e.\ a geometric scoring rule with $p=2$. Laplace's motivation was that if we suppose that a voter's degree of support for a candidate ranked $j$th can be quantified as $x$, then we cannot say whether the voter's support for the $(j+1)$th candidate is $x-1,x-2$, or any other smaller value. As such, he proposed to take the average, or $x/2$. The recent work of \citet{Csato21scoring}, answering a question posed by an earlier version of this paper, investigates geometric scoring rules in the context of the threat of early clinch in Formula One racing.}

\begin{definition}
A \textbf{geometric scoring rule} is a generalised scoring rule that is defined with respect to a parameter $p$. The score of the $j$th position is affinely equivalent to:
\[s_j^m= \begin{cases} 
      p^{m-j}& 1<p<\infty, \\
      m-j & p=1, \\
      1-p^{m-j} & 0<p<1. 
   \end{cases}
\]
We include generalised plurality $(p\ria\infty)$ and generalised antiplurality $(p\ria~0)$ as edge cases.\footnote{Observe that for any fixed number of races $n$, choosing a $p\geq n$ will guarantee that no amount of $(j+1)$th places will compensate for the loss of a single $j$th place -- this rule will be precisely generalised plurality. Likewise, choosing $p\leq 1/n$ will give us generalised antiplurality. By including these rules in the family of geometric scoring rules, all we are asserting is that the organiser is allowed to choose a different value of $p$ if the length of the tournament, $n$, changes.}
\end{definition}

\bigskip

\begin{theorem}\label{thm:geometricrules}
A scoring rule satisfies independence of unanimous winners and independence of unanimous losers if and only if it is a geometric scoring rule.
\end{theorem}

\bigskip

Observe that the axioms we used are extremely weak individually. If $s_1^m,\dots,s_m^m$ is any monotone decreasing sequence of scores whatsoever, and we obtain $s_j^{m-1}$ by dropping $s_m^m$, we will satisfy independence of unanimous losers. Likewise, if we obtain $s_j^{m-1}$ by dropping $s_1^m$, we will satisfy independence of unanimous winners. In short our only restriction is that more points are awarded for the $j$th place than for the $(j+1)$th place, which in the context of a sporting event is hardly a restriction at all. If we want to satisfy both axioms, however, we are suddenly restricted to a class with just one degree of freedom.

It is easy to see that geometric scoring rules also satisfy two stronger properties one might label independence of unanimous winning/losing cliques. Suppose there is a clique of $k$ athletes that always come in the first (last) $k$ positions, but possibly in any order; adding or removing such a clique will not change the order of the other athletes. Such a property is relevant in sporting events such as Formula One racing, where the top spots are consistently taken by a small number of strong teams.

In the following sections we will explore this class. We shall see how our axioms allow new axiomatisations of well-known (generalised) scoring rules, and how geometric scoring rules compare to optimal rules for a given organiser's objective.

\section{New characterisations}\label{sec:newcharacterisations}

\subsection*{$p>1$: Convex rules and winning in every race}

The FIM motorcycle Grand Prix is another championship that uses a scoring system to select a winner. The 125cc category of the 1999 season had a curious outcome: the winner was Emilio Alzamora, who accumulated the largest amount of points, yet did not win a single race (\Cref{table:Alzamora}). This does not detract in any way from Alzamora's achievement -- he outperformed his competitors by virtue of his consistently high performance (compare with Melandri who performed well in the second half, and Azuma in the first), and if he did not take any unnecessary risks to clinch the first spot then he was justified in not doing so. However, racing is a spectator sport. If a fan attends a particular event then they want to see the athletes give their best performance on the day, rather than play it safe for the championship. 

\begin{table}
\centering
\caption{1999 Motorcycle Grand Prix -- 125cc: scoring system and event results}
\begin{tabular}{lccccccccccccccc}
\toprule
Position&1&2&3&4&5&6&7&8&9&10&11&12&13&14&15\\
\hline
Points&25&20&16&13&11&10&9&8&7&6&5&4&3&2&1\\
\bottomrule
\end{tabular}
\begin{tabular}{lccccccccccccccccr}
\\
\toprule
Rider&\multicolumn{16}{c}{Event number: points}&Total\\
&1&2&3&4&5&6&7&8&9&10&11&12&13&14&15&16&score\\
\hline
Alzamora&20&16&16&16&10&20&13&16&20&10&13&20&1&-&16&20&227\\
Melandri&-&-&-&10&20&16&8&11&\textbf{25}&\textbf{25}&\textbf{25}&-&\textbf{25}&16&20&\textbf{25}&226\\
Azuma&\textbf{25}&\textbf{25}&\textbf{25}&13&9&-&\textbf{25}&\textbf{25}&10&4&6&-&11&2&10&-&190\\
\bottomrule
\end{tabular}
\label{table:Alzamora}
\vspace{0.2cm}
\justify 
\footnotesize{\textit{Notes}: The scores for first place finishes are in bold. Observe that Azuma performed well in the first half of the tournament and Melandri in the second, while Alzamora performed consistently in both halves, yet never came first.}
\end{table}

Bernie Ecclestone, the former chief executive of the Formula One Group, was outspoken about similar issues in Formula One racing -- ``It's just not on that someone can win the world championship without winning a race.'' Instead of the scores then used, Ecclestone proposed a medal system. The driver who finished first in a race would be given a gold medal, the runner-up the silver, the third the bronze. The winner of the championship would be the driver with the most gold medals; in case of a tie, silver medals would be added, then the bronze, then fourth-place finishes, and so on.\footnote{
Bernie Ecclestone justified the medal system as follows:
\begin{quote}
The whole reason for this was that I was fed up with people talking about no overtaking. The reason there's no overtaking is nothing to do with the circuit or the people involved, it's to do with the drivers not needing to overtake.

If you are in the lead and I'm second, I'm not going to take a chance and risk falling off the road or doing something silly to get two more points.

If I need to do it to win a gold medal, because the most medals win the world championship, I'm going to do that. I will overtake you.
\end{quote}
From: \url{https://web.archive.org/web/20191116144110/https://www.rte.ie/sport/motorsport/2008/1126/241550-ecclestone/}} In other words, he proposed the generalised plurality system with $p\ria\infty$. And indeed, for every other geometric scoring rule, it is possible to construct a profile where the overall winner did not win a single race.

\begin{restatable}{proposition}{alzamoraparadox}\label{prop:alzamora}
For any $p<\infty$, there exist $n,m$, and a profile with $n$ races and $m$ athletes where the overall winner does not come first in any race.
\end{restatable}

There is a natural dual concept to Ecclestone's criterion: rather than asking how many races an athlete must win to have a chance of winning the championship, we could ask after how many victories is the championship guaranteed.\footnote{To win the championship for sure, an athlete must come first in more than $n(s_1^m-s_m^m)/(2s_1^m-s_2^m-s_m^m)$ races; see \citet[][theorem~10]{KondratevNesterov20} and \citet[][theorem~1]{BaharadNitzan02}.} This leads us to the \textbf{majority criterion}, which requires that any athlete that won more than half the races (the \textbf{majority winner}) should also win the championship. The majority criterion together with independence of unanimous winners allows us to characterise generalised plurality.

\bigskip

\begin{restatable}{theorem}{generalisedplurality}\label{thm:generalisedplurality}
Generalised plurality is the only generalised scoring rule that satisfies independence of unanimous winners and always ranks the majority winner first.
\end{restatable}

\bigskip

The presence of the majority criterion in the above theorem is not surprising. We could expect as much given the axiomatisation of plurality by \citet[theorem~2]{Lepelley92} and \citet[theorem~4.1]{Sanver02}. But the fact that adding independence of unanimous winners allows us to pin down generalised plurality is interesting because generalised scoring rules are notoriously hard to characterise \citep{BossertSuzumura20}, but in this case two intuitive axioms suffice.

Here we run into a conundrum. Ecclestone's criterion is desirable in the case of a sporting event for the reasons we have mentioned -- the stakes are high in every race, and this encourages the athletes to fight for the top spot rather than settling for second place. The majority criterion, on the other hand, tells a different story. If a driver were to win the first $\floor{n/2 + 1}$ races, the championship is over. The remaining races will take place before an empty stadium.

It seems there is a trade-off between keeping the tension high in an individual race and over the course of the entire tournament, which could explain why the organisers of Formula One World Championship went through so many ranking systems over the years.\footnote{The organisers of Formula One World Championship have produced a study comparing historical results to the hypothetical outcome had Ecclestone's medal system been used (\url{https://web.archive.org/web/20100106134601/http://www.fia.com/en-GB/mediacentre/pressreleases/f1releases/2009/Pages/f1_medals.aspx}). While such comparisons should be taken with a grain of salt, since they do not take into account the fact that athlete's strategies would have been different under a different scoring system, it is nevertheless interesting that they find that 14 championships would have been shorter, while 8 would have been longer.

In a recent work, \citet{Csato21scoring} examines the trade-off between reducing the likelihood of an athlete winning the championship without coming first in any race, and delaying the point at which the championship is decided. The author compares the historic Formula One scoring schemes and geometric scoring rules on a synthetic dataset, and finds that the current Formula One system is indeed on the Pareto frontier.}

\subsection*{$p=1$: The Borda rule and top-winner reversal bias}

\citet{SaariBarney03} recount an amusing anecdote to motivate an axiom known as top-winner reversal bias. In a departmental election the voters were asked to rank three candidates. Mathematically, this is the same problem as ours -- to aggregate $n$ rankings into one final result. All voters ranked the candidates from best to worst, but the chair expected the votes to be ordered from worst to best. The ``winner'' was thus the candidate that ranked highest in terms of the voters' assessment of unsuitability, rather than suitability for the role. After the ensuing confusion, the votes were retallied... and the winner was unchanged. The same candidate was judged to be at once the best and worst for the role. The authors' story ended with the chair being promoted to a higher position, but in a sporting context we could expect a less polite outcome.

The relevant axiom here is \textbf{top-winner reversal bias}, which states that if a candidate $a$ is the unique winner with voters' preferences $R_1,\dots,R_n$, then if we invert the preferences of every voter then $a$ will no longer be the unique winner. Note that the axiom asks for less than one might expect -- we are not asking that the candidate formerly judged the best is now judged the worst, but merely that the same candidate cannot be the best in both cases.

By itself, this axiom is quite weak. If we, without loss of generality, set $s_1=1$ and $s_m=0$, then the only restriction is that for $1\leq j\leq m/2$, $s_{m-j+1}=1-s_j$ \citep[theorem~1]{SaariBarney03}.\footnote{\citet{Gardenfors73} introduces duality (the most demanding variant of reversal bias, also known as inversion and reversal symmetry) and verifies it for different modifications of the Borda rule. \citet{Morkeliunas77} studies interrelations between variants of neutrality, independence, and duality. \citet[theorem~4.2]{FineFine74b} and \citet[corollary~5]{LlamazaresPena15} independently describe the scoring rules that satisfy inversion and reversal symmetry. \citet[theorem~1]{SaariBarney03} describe the scoring rules that are invulnerable to different variants of the reversal bias paradox but omit the formal proof. 

Similar axioms are studied in the case of single-winner (as opposed to ranking) rules \citep[p.~157]{Morkeliunas82,Fishburn73book}. \citet[theorem~4]{HeckelmanRagan21} characterise the scoring rules that satisfy duality (invertibility in their terminology). However, the set of axioms in their characterisation is not complete. For instance, for the case of four candidates, a generalised scoring rule with the Borda vector in the first round and $(2,1,1,0)$ vector in the second round satisfies all their axioms, but is not even a scoring rule. Their characterisation will become correct if we add, for instance, an axiom of continuity introduced by \citet{Young75}.} In other words, we have $\floor{\frac{m-2}{2}}$ degrees of freedom. However, once we add either one of our independence axioms, we get the Borda rule uniquely.

\bigskip

\begin{restatable}{theorem}{Borda}\label{thm:borda}
Borda is the unique scoring rule that satisfies top-winner reversal bias and one of independence of unanimous winners or independence of unanimous losers.
\end{restatable}

\bigskip

Within the context of geometric scoring rules, there is an easier way to see that Borda is the unique rule satisfying top-winner reversal bias. If $g_p$ is the geometric scoring rule with parameter $p$, and $p\neq 1$, then by a result of \citet[theorem~1]{Fishburn81DAM} there exists a profile $Q$ such that $g_p(Q)$ is the reverse order of $g_{1/p}(Q)$, $g_p(Q)=\mathit{rev}(g_{1/p}(Q))$. At the same time, it is easy to see that $g_p(\mathit{rev}(Q))=\mathit{rev}(g_{1/p}(Q))$. Combining the two we get $g_p(Q)=g_p(\mathit{rev}(Q))$, which violates top-winner reversal bias.

A more demanding version of reversal bias, known as \textbf{duality} \citep{Gardenfors73} and \textbf{inversion} \citep{FineFine74b}, states that if we invert the individual rankings of every voter then the aggregate ranking will be also inverted. \citet[corollary to theorem~4.3]{FineFine74b} use inversion to provide a characterisation of Borda similar to \autoref{thm:borda}.\footnote{In addition to using inversion instead of top-winner reversal bias to characterise the Borda rule, the independence axioms of \citet{FineFine74b} do not require the ranking of a unanimous winner first and a unanimous loser last, as our axioms do, but instead their definition of scoring rules directly requires non-decreasing scores.

In this paper we consider Borda, and all other scoring rules, as ranking rules -- they produce an ordering of the athletes from best to worst. If we are only interested in the version of Borda that selects the winner(s), then a difficult to find paper of \citet{Morkeliunas82} presents a characterisation of the Borda rule using variants of duality and independence of unanimous losers. \citet[theorem~5]{HeckelmanRagan21} rediscovered the characterisation, but used a more demanding axiom than the duality of \citet{Morkeliunas82} and additionally required an axiom of positive responsiveness. Moreover, the proof of \citeauthor{HeckelmanRagan21} relies on their theorem~4, which requires additional axioms, as we have shown in the previous note.}

\subsection*{$p<1$: Concave rules and majority loser paradox}

The beginning of modern social choice theory is often dated to Borda's memorandum to the Royal Academy \citep{Borda1781}, where he demonstrated that electing a winner by plurality could elect a \textbf{majority loser} -- a candidate that is ranked last by an absolute majority of the voters.\footnote{\citet[corollary~4]{LlamazaresPena15} and \citet[theorem~10]{KondratevNesterov20} provide equivalent characterisations of the scoring rules that never rank the majority loser first (in their notation, immune to the absolute loser paradox, and satisfy the majority loser criterion, respectively).} The extent to which such a result should be viewed as paradoxical depends on the context in which a ranking rule is used. In sports, this may be acceptable -- a sprinter who has three false starts and one world record is still the fastest man in the world. In a political context, however, voting is typically justified by identifying the will of the majority with the will of the people; it would be odd to argue that the will of the majority is to pick a candidate that the majority likes the least. Likewise, should a group recommendation system suggest that a group of friends watch a film that the majority detests, soon it would be just a group.

It turns out that the weak version of the criterion -- that the majority loser is never ranked first -- is characteristic of the concave geometric rules ($p\leq 1$). The strong version -- that the majority loser is always ranked last -- is satisfied only by generalised antiplurality ($p\ria~0$).


\begin{restatable}{theorem}{concaveGSR}\label{thm:concaveGSR}
Geometric scoring rules with parameter $0 < p\leq 1$ are the only scoring rules that satisfy independence of unanimous winners and independence of unanimous losers and never rank the majority loser first.
\end{restatable}

\begin{restatable}{theorem}{concaveGSRTwo}\label{thm:concaveGSRTwo}
Generalised antiplurality is the only generalised scoring rule that satisfies independence of unanimous losers and always ranks the majority loser last.
\end{restatable}


\section{Optimal scoring rules}\label{sec:optimal}

The practical relevance of the previous sections is that if the organiser accepts that our two axioms are desirable -- and they are very natural axioms -- then the problem of choosing a scoring rule is reduced to the choice of a single parameter,~$p$.

Unfortunately, the choice of even a single parameter is far from trivial. In the previous section we saw how an axiomatic approach can pin down the edge cases of generalised plurality ($p\ria\infty$), Borda ($p=1$) or generalised antiplurality ($p\ria~0$).\footnote{Theoretical frameworks where the best scoring rule was found to be different from Borda, plurality, and antiplurality are rare indeed. Some examples include \citet{Lepelley95}, \citet{LepelleyPierronValognes00,LepelleyMoyouwouSmaoui18}, \citet{CervoneGehrleinZwicker05}, \citet{Sitarz13}, \citet{Kamwa19}, \citet{DissKamwaMoyouwouSmaoui21} and \citet{Kilgour22}.} In applications where the properties these axioms represent are paramount, the question is then settled: if you are after a scoring rule that satisfies independence of unanimous losers and top-winner reversal bias, you must use Borda. There is no other. However, in the case of sports these extreme rules are rarely used. Whatever goals the organisers are pursuing, these are more complicated than simply satisfying an axiom.

In the remainder of this paper, we will take an empirical approach to selecting a scoring rule for an event. We introduce a model of the organiser's objective, assuming the goal is to select an athlete that maximises some measure of quality, which aggregates the athlete's cardinal results. By imposing four axioms we see that this aggregation function ($F_\lambda$) must be determined solely by a parameter $\lambda$, which can be interpreted as the organiser's preferences for peak performance versus consistency. It turns out that among all ordinal procedures for producing a ranking of athletes, it is precisely the scoring rules which rank the athletes in accordance to the expected values of $F_\lambda$, and these scores can be computed from empirical data. If the distribution of the athletes' cardinal results is uniform, then the optimal scoring rules are approximately geometric, but in general the two will differ. We conclude by computing these optimal scores for the IBU World Cup biathlon, PGA TOUR golf, and IAAF Diamond League athletics, and compare them to the best approximation via a geometric scoring rule.

\subsection*{The organiser's objective}

An organiser's goals can be complex. For a commercial enterprise the end goal is profit, whether from ad revenue or spectator fees. To that end they would prefer that athletes take risks and keep the audience on edge, rather play a safe and sure strategy. If a tournament lasts for a long time, the presence of consistently strong athletes -- crowd favourites -- could help hold the viewers' attention throughout the season. In this case the organiser would want a system that encourages athletes to perform consistently well in every race.
In a youth racing league, the focus could be that the drivers finish the race with engines and bodies intact -- the goal being that the drivers learn to finish the race, before trying to finish it in record time.\footnote{We have previously mentioned that Bernie Ecclestone is outspoken about the need of a ranking system that motivates racers to overtake in every race. For contrast, in the Tour de France the prize for the most aggressive athlete -- the prix de la combativit\'e -- is relatively unknown, while the most prestigious award is the Yellow Jersey, awarded to the athlete with the best overall time.

The Castrol Toyota Racing Series positions itself as incubating the next generation of racing talent, and the competitors tend to be very young. The 2018 season consisted of 14 drivers, and the score for the fourteenth position was 24. The winner was Robert Shwartzman with a total score of 916. Richard Verschoor came second with 911, but failed to complete the third race (and thus got no points). Had Verschoor finished in any position whatsoever, he would have taken the championship.}

Because of this, we want our model of the organiser's objective to be as general as possible. We assume that in each of the $n$ events, an athlete's performance in an event $i$ can be assessed as a cardinal quantity, $x_i$ (e.g.\ finishing time in a race, strokes on a golf course, score in target shooting). The athlete's aggregate performance is measured by a function $F:\mathbb{R}^n\ria\mathbb{R}$ that maps these $n$ cardinal quantities, $\bs{x}=(x_1,\ldots,x_n)$, into an overall measure of quality -- an aggregation function \citep{Grabisch09book,Grabisch11}. The space of such functions is too vast to be tractable, so we shall narrow it down by imposing four axioms on how a measure of quality should behave.

The first axiom has to do with the measurement of the cardinal qualities $x_i$. Suppose an athlete competes in the javelin throw, and in the $i$th round throws a distance of 95 metres. There are two natural ways in which we could record this. The first is to simply set $x_i=95$, the second is to compare the throw to the current world record of 98.48 and set $x_i=95-98.48=-3.48$. It would be absurd if the two approaches would rank our athlete differently vis-\`a-vis the other athletes. Thus we require the condition of \textbf{independence of the common zero}, which states that whenever $F(\bs{x})\geq F(\bs{y})$, it is also the case that $F(\bs{x}+\bs{c})\geq F(\bs{y}+\bs{c})$, where $\bs{c}=(c,\dots,c)$ and the notation $\bs{x}+\bs{c}$ denotes $(x_1+c,\dots,x_n+c)$.

The next two axioms deal with the intuition that our aggregation function is intended to measure quality, and hence higher values of $x_i$, the performance in an individual event, should contribute to a higher level of $F(\boldsymbol{x})$, the overall quality. The least we could ask for is that if an athlete performs (strictly) better in \emph{every} event, then their overall quality should also be (strictly) higher. This is the condition of \textbf{unanimity}, requiring that whenever $x_i\geq  y_i$ $(x_i>y_i)$ for all $i$, it is also the case that $F(\boldsymbol{x})\geq F(\boldsymbol{y})$ $(F(\boldsymbol{x})>F(\boldsymbol{y}))$.

Next, consider the admittedly odd situation where two javelin throwers, $a$ and $b$, obtain potentially different results on the first $q$ throws, but throw the javelin the exact same distance as each other in throws $q+1$ through $n$. For example, let $a$'s results on the first three throws be $(94, 90, 89)$, $b$'s results -- $(93, 93, 92)$, and for the sake of argument let us suppose that $F$ assigns a higher quality to $a$. It is natural to assume that this decision does not change if throws 4 through 6 are identical. As such, if the complete results are $(94, 90, 89, \mathit{70, 94, 90})$ for $a$ and $(93, 93, 92, \mathit{70, 94, 90})$ for $b$, we would still expect $F$ to assign a higher quality to $a$. This is the property of \textbf{separability}, stating that for $\boldsymbol{x}=(x_1,\ldots,x_q),\boldsymbol{y}=(y_1,\ldots,y_q)$, and $\boldsymbol{z}=(z_{q+1},\ldots,z_n)$, whenever $F(\boldsymbol{x})\geq F(\boldsymbol{y})$, it is also the case that $F(\boldsymbol{x}\boldsymbol{z})\geq F(\boldsymbol{y}\boldsymbol{z})$, where the notation $\boldsymbol{x}\boldsymbol{z}$ denotes $(x_1,\dots,x_q,z_{q+1},\dots,z_n)$.

The final condition perhaps has the most bite. We assume that the order of the results does not matter -- it should not matter whether an athlete throws 93 in round $i$ and 92 in round $q$, or vice versa; \textbf{anonymity} requires that $F(\boldsymbol{x})=F(\pi\boldsymbol{x})$, for any permutation $\pi$. This would have been an innocuous assumption in a political context, where it is standard to assume that all voters are equal, but it is a real restriction in sports as it is entirely natural for different events to be weighted differently. However, we justify this assumption since the three categories we examine in the next section (IBU World Cup biathlon, PGA TOUR golf, IAAF Diamond League athletics) do not distinguish between their events in scoring.

It turns out that the only continuous solution satisfying these four properties \citep[theorem~2.6, p.~44]{Moulin1991} is defined with respect to a parameter $\lambda$ and is the following:\footnote{Our version of the separability axiom implies Moulin’s (\citeyear{Moulin1991}) separability; this argument is well known, see e.g. lemma~18 in \citet{Kothiyal14}.}
\begin{equation*}\label{exponential-quality}
    F_{\lambda}(\boldsymbol{x})=  \sum\limits_{i=1}^n {u_\lambda(x_i)} =
    \begin{cases}
    \sum\limits_{i=1}^n {\lambda^{x_i}}, &\lambda>1,\\
    \sum\limits_{i=1}^n x_i, &\lambda=1,\\
    \sum\limits_{i=1}^n -{\lambda^{x_i}}, &0<\lambda<1.
    \end{cases}
\end{equation*}
As an added bonus, $F_{\lambda}$ enjoys a version of scale invariance. Since $\lambda^{\alpha x_i}=(\lambda^\alpha)^{x_i}$, it does not matter whether the race is measured in minutes or seconds, provided the organiser adjusts the value of $\lambda$ accordingly.

The parameter $\lambda$ can be interpreted as the organiser's preferences for peak performance versus consistency. With $\lambda=1$, the organiser values consistency and assesses athletes by their average performance. As $\lambda$ increases, the organiser is more willing to tolerate poor average performance for the possibility of observing an exceptional result, culminating in the lexmax rule as $\lambda\ria\infty$. As $\lambda$ decreases, the organiser is increasingly concerned about subpar performance, tending to the lexmin rule as $\lambda\ria0$. Other factors concerning the choice of $\lambda$ are discussed in Appendix~\ref{app:lambda}.

We shall thus assume that the organiser assesses the quality of the athletes via $F_{\lambda}$, and wishes to choose a sequence of scores such that the athletes with the highest quality have the highest total score.

\bigskip

\subsection*{Why scoring rules?}

At this point one may ask, if we have access to the cardinal values $x_i$, why bother with a scoring rule at all? In a political context, cardinal voting is problematic since voters may not know their utilities exactly, and in any case would have no reason to report them sincerely, but in sport these are non-issues -- we can measure $x_i$ directly, and a race protocol is incapable of strategic behaviour. Nevertheless, a cardinal approach has its problems even in sport. In a contest where athletes are operating near the limits of human ability, the cardinal difference between first and second place could be minuscule, and a race decided by milliseconds. On the other hand failing to complete a race, or completing it poorly for whatever reason, would be an insurmountable penalty. Ordinal rankings also allow the comparison of results between different races, while cardinal results would be skewed by external factors like wind, rain, or heat. This can explain why in practice ordinal procedures are more popular.

The advantages of a scoring rule over other ordinal procedures is, in addition to the axiomatic properties discussed before, the fact that if we are interested in maximising a sum of cardinal utilities (such as $F_{\lambda}$), then the optimal voting rule is a scoring rule, provided the utilities are drawn i.i.d.\ from a distribution symmetric with respect to athletes.

\clearpage

\begin{theorem}[{\citealp[p. 277--279]{ApesteguiaBallesterFerrer2011,Boutilier2015,Laplace1795}}]\label{thm:optimal}
Denote by $u_i^a$ the cardinal quality of athlete~$a$ in race~$i$. Denote by $u_i=(u_i^1,\ldots,u_i^m)$ the vector of cardinal qualities in race $i$ and $(u_i^{(1)},\ldots,u_i^{(m)})$ its reordering in non-increasing order. Suppose $u_1,\ldots,u_n$ are drawn independently and identically from a distribution with a symmetric joint cumulative distribution function (i.e., permutation of arguments does not change the value of this c.d.f.).

Consider a scoring rule with scores equal to the expected value of the corresponding order statistics:
\begin{equation*}
s_j=\EE[u_i^{(j)}]. 
\end{equation*}
Then the winner under this scoring rule is the athlete with the highest expected overall quality:
\begin{equation*}
    \max_{a}{\EE\left[\sum\limits_{i=1}^n u_i^a \,\middle\vert\, R_1(u_1),\ldots,R_n(u_n)\right]},
\end{equation*}
where expectation is conditional on $R_i(u_i)$ -- the ordinal ranking induced by $u_i$. If we make the further assumption that the cardinal qualities $u_i^a$ are drawn independently and identically (i.e.\ we further assume that the performances of athletes in a race are independent) from a distribution with a continuous density function, it is also the case that the total score of $a$ is equal to $a$'s expected overall quality.
\end{theorem}

\bigskip

Substituting $u_i^a=\lambda^{x_i^a}$ for $\lambda>1$, $u_i^a=x_i^a$ for $\lambda=1$ and $u_i^a=-\lambda^{x_i^a}$ for $0<\lambda<1$, it follows that if the organiser wishes to choose a winner based on $F_\lambda$, they should use a scoring rule.\footnote{The theorem statements of \citet[theorem~3.1]{ApesteguiaBallesterFerrer2011} and \citet[theorem~4.2]{Boutilier2015} assume that all values $u_i^a$ are non-negative, but the proofs work for any value of~$u_i^a$.} The \textbf{optimal scoring rule} for a given $\lambda$ can be computed by evaluating $\EE[u_i^{(j)}]$ on historical data.

\begin{example}

In \autoref{table:optimalscorescalculation}, we demonstrate how the optimal scoring sequence for the men's 100m sprint could be computed, assuming the only data we have available is from the 2015 IAAF Diamond League. If the organiser values consistent performance ($\lambda=1$), then $u_i^a=x_i^a$, so by \autoref{thm:optimal} the score awarded for the first position should equal the expected performance of the first-ranked athlete. Evaluating this on our data, we have $(-0.16-0.30-0.17-0.54-0.23-0.29)/6=-0.28$. Repeating the calculations for the remaining positions, the optimal scoring vector is $(-0.28, -0.38, -0.43, -0.46, -0.49, -0.54, -0.60, -0.63)$. If we desire a more visually appealing vector, recall that affinely equivalent scores produce identical rankings, so we can normalise the scores to range from 0 to 100, namely $(100,73,59,49,42,27,11,0)$.

If the organiser values the chance of exceptional performance more than consistency, then their measure of athlete quality is parameterised by a $\lambda>1$. The exact value is exogenous to our model, but as a consequence of \autoref{thm:optimal}, $\lambda$ has a natural numerical interpretation -- how much is an extra unit of performance worth? Choosing a $\lambda>1$ displays a willingness to award an athlete who completes a race with $x+1$ units of performance $\lambda$ times as many points as the athlete that completes the race with $x$ units. In \autoref{table:optimalscorescalculation} we measure performance in seconds, and one second is a colossal difference in the 100m sprint. Thus choosing a $\lambda$ as high as $100$ seems perfectly reasonable. With $\lambda=100$, $u_i^a=100^{x_i^a}$, so the score awarded for the first position ought to be $(100^{-0.16}+100^{-0.30}+100^{-0.17}+100^{-0.54}+100^{-0.23}+100^{-0.29})/6=0.31$, and the normalised vector is $(100,51,34,26,20,11,5,0)$.
\end{example}

\begin{table}
\centering
\caption{Men’s 100m of the IAAF Diamond League 2015}
\begin{tabular}{ccccccc|cc|cc}
\toprule
Position&\multicolumn{6}{c}{Event: lag behind world record}&\multicolumn{4}{c}{Optimal scores}\\
&Doha&Eugene&Rome&New York&Paris&London&\multicolumn{2}{c}{$\lambda=1$}&\multicolumn{2}{c}{$\lambda=100$}\\
\hline
1&-0.16&-0.30&-0.17&-0.54&-0.23&-0.29&-0.28&100&0.31&100\\
2&-0.38&-0.32&-0.40&-0.55&-0.28&-0.32&-0.38&73&0.19&51\\
3&-0.43&-0.41&-0.40&-0.57&-0.41&-0.34&-0.43&59&0.15&34\\
4&-0.45&-0.41&-0.48&-0.60&-0.44&-0.38&-0.46&49&0.13&26\\
5&-0.46&-0.44&-0.49&-0.66&-0.47&-0.40&-0.49&42&0.11&20\\
6&-0.49&-0.55&-0.50&-0.70&-0.50&-0.49&-0.54&27&0.09&11\\
7&-0.52&-0.69&-0.50&-0.82&-0.54&-0.50&-0.60&11&0.07&5\\
8&-0.56&-0.70&-0.56&-0.87&-0.60&-0.51&-0.63&0&0.06&0\\
\bottomrule
\end{tabular}
\label{table:optimalscorescalculation}
\vspace{0.2cm}
\justify 
\footnotesize{\textit{Notes}: The numbers on the left represent the difference in seconds between the world record (9.58) and the time of the athlete that finished first through eighth. On the right we see the raw and normalised optimal scoring sequence computed on this data for parameters $\lambda=1$ and $\lambda=100$.}
\end{table}

\subsection*{Parallels to geometric scoring rules}

The reader will notice that $F_\lambda$ bears a resemblance to a geometric scoring rule -- for $p,\lambda>1$ we raise a certain parameter to the power of a measure of performance in a given race (whether cardinal or ordinal), and sum the result across the races.

We arrived at similar results because we started with similar axioms. Scoring rules are characterised by anonymity, neutrality, and electoral consistency (\citealp{Smith73,Young74,Young75}, see \autoref{app:complete}). Anonymity and neutrality require that scoring rules treat races and athletes equally; in the cardinal setting we impose anonymity directly, and neutrality is implicit in the fact that we use the same $F_\lambda$ to measure the quality of every athlete. Electoral consistency guarantees that if an athlete is leading in the first $q$ and the last $n-q$ races of the tournament taken separately, then he is also the champion overall. Separability is similar in that it allows us to interpret $F_\lambda(\boldsymbol{x})\geq F_\lambda(\boldsymbol{y})$ as meaning that the first athlete is better in the first $q$ races ($\boldsymbol{x}$ versus $\boldsymbol{y}$) and (weakly) better in the last $n-q$ ($\boldsymbol{z}$ versus $\boldsymbol{z}$), then he is also better overall.

Crucially, independence of the common zero allows us to raise or lower the performance of all athletes by a common $c$ without affecting their relative ranking. It seems that padding the profile above or below with unanimous winners/losers is in some sense the ordinal equivalent of adding~$c$.

Formally, we can show that in the case of a uniform distribution, optimal scores are in fact approximately geometric.

\bigskip

\begin{restatable}{theorem}{OptimalUniform}\label{thm:OptimalUniform}
Let $x^1,\ldots,x^m$ be independently and uniformly distributed on $[a,b]$, and $x^{(1)}\geq\ldots\geq x^{(m)}$ be their reordering in non-increasing order. As $m\ria\infty$, the optimal scores $s_j=\EE\left[\lambda^{x^{(j)}}\right]$ for $\lambda>1$ and $s_j=\EE\left[-\lambda^{x^{(j)}}\right]$ for $0<\lambda<1$ converge to geometric scores with parameter $p=\lambda^{\frac{b-a}{m+1}}$.

For $\lambda=1$ the optimal scoring rule is exactly Borda, which has been known since \citet[p. 277--279]{Laplace1795}.
\end{restatable}

\bigskip

In the limit cases of $\lambda\ria\infty$ and $\lambda\ria0$, the optimal scoring rule tends to generalised plurality and antiplurality for a wide class of distributions.

\bigskip

\begin{restatable}{theorem}{pluralityLimit}\label{thm:pluralityLimit}
Let the number of potential athletes be fixed and finite ($M<\infty$), and for every number $m$ of athletes, $m\leq M$, their performances $x_i^1,\ldots,x_i^m$ in each competition~$i$ be drawn independently and identically from a distribution on $[a,b]$ such that density function $f$ and its derivative are bounded and continuous.

Suppose $b$ is finite and $f(b)>0$. Then as $\lambda\ria\infty$, the optimal scores $s_j=\EE\left[\lambda^{x_i^{(j)}}\right]$ tend towards (but are never equivalent to) plurality. If, in addition, the number of competitions $n$ is bounded from above and the first $M-1$ derivatives of $f$ are bounded and continuous, then there exists a finite $\overline{\lambda}$ such that for each $\lambda>\overline{\lambda}$ the optimal scoring rule is equivalent to generalised plurality.

Suppose $a$ is finite, and $f(a)>0$. Then as $\lambda\ria0$, the optimal scores $s_j=\EE\left[-\lambda^{x_i^{(j)}}\right]$ tend towards (but are never equivalent to) antiplurality. If, in addition, the number of competitions $n$ is bounded from above and the first $M-1$ derivatives of $f$ are bounded and continuous, then there exists a finite $\underline{\lambda}$ such that for each $0<\lambda<\underline{\lambda}$ the optimal scoring rule is equivalent to generalised antiplurality.
\end{restatable}

\bigskip


\section{Empirical evaluation}\label{sec:empirical}

How realistic is our assumption that the organiser assesses athlete performance by the aggregation function $F_{\lambda}$? We compared the actual scores used in the IBU World Cup biathlon (\autoref{OptimalGraphIBU}), the PGA TOUR golf, and the IAAF Diamond League athletics (\autoref{OptimalGraphPGA}). Details about the data and calculations can be found in \autoref{app:love_sport}.

\begin{figure}
\begin{center}
\includegraphics[width=8cm]{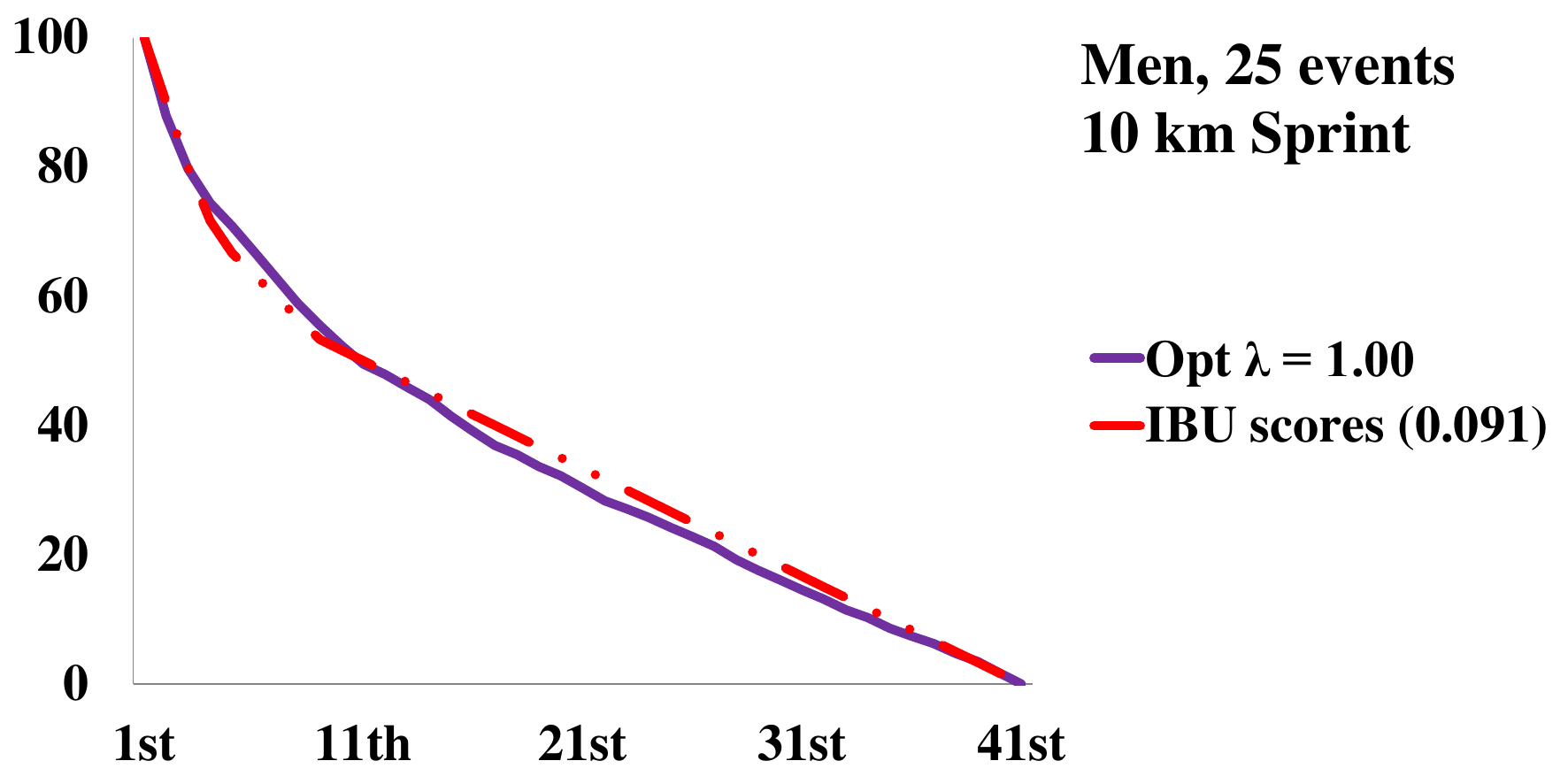}
\includegraphics[width=8cm]{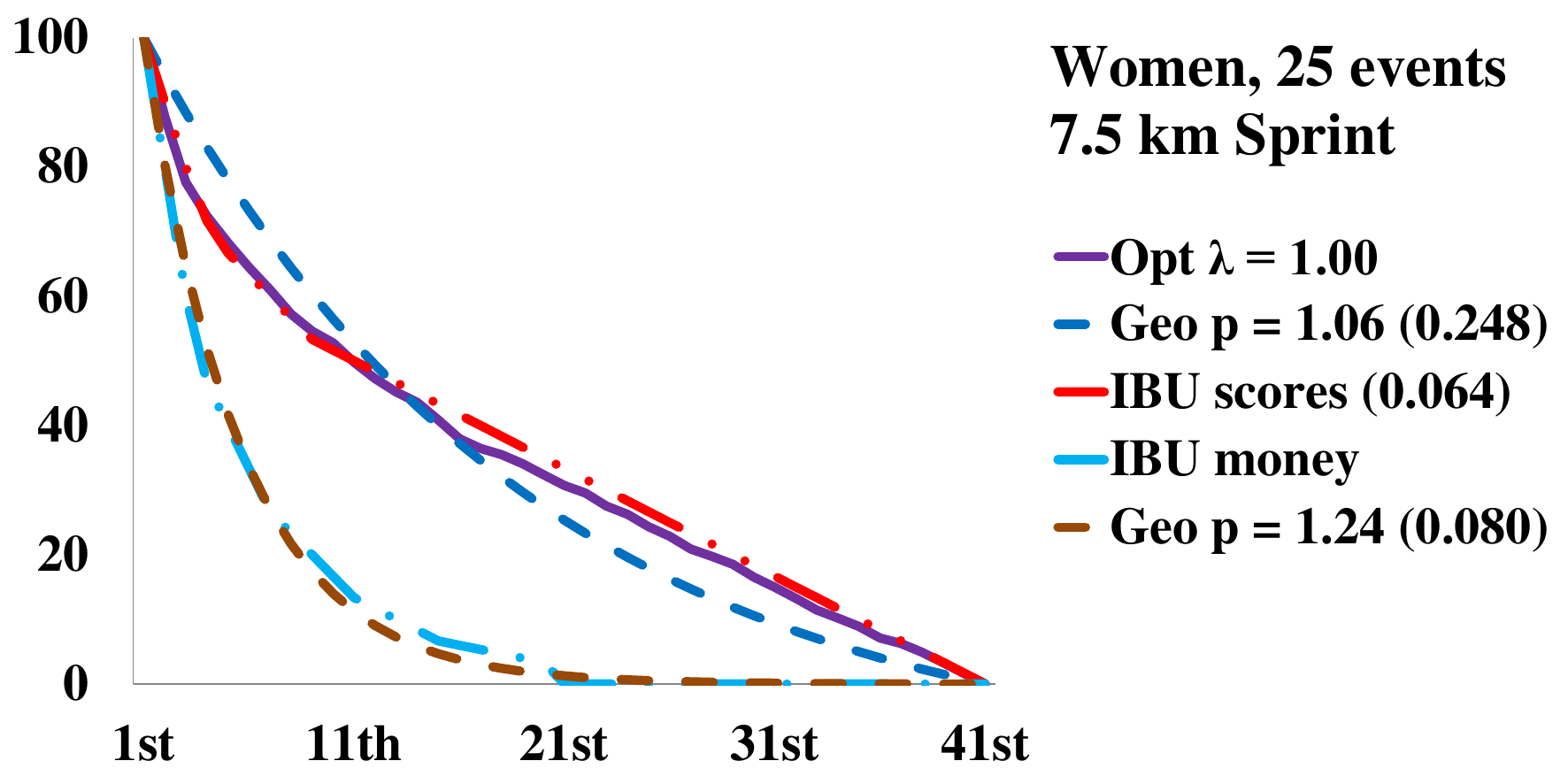}
\includegraphics[width=8cm]{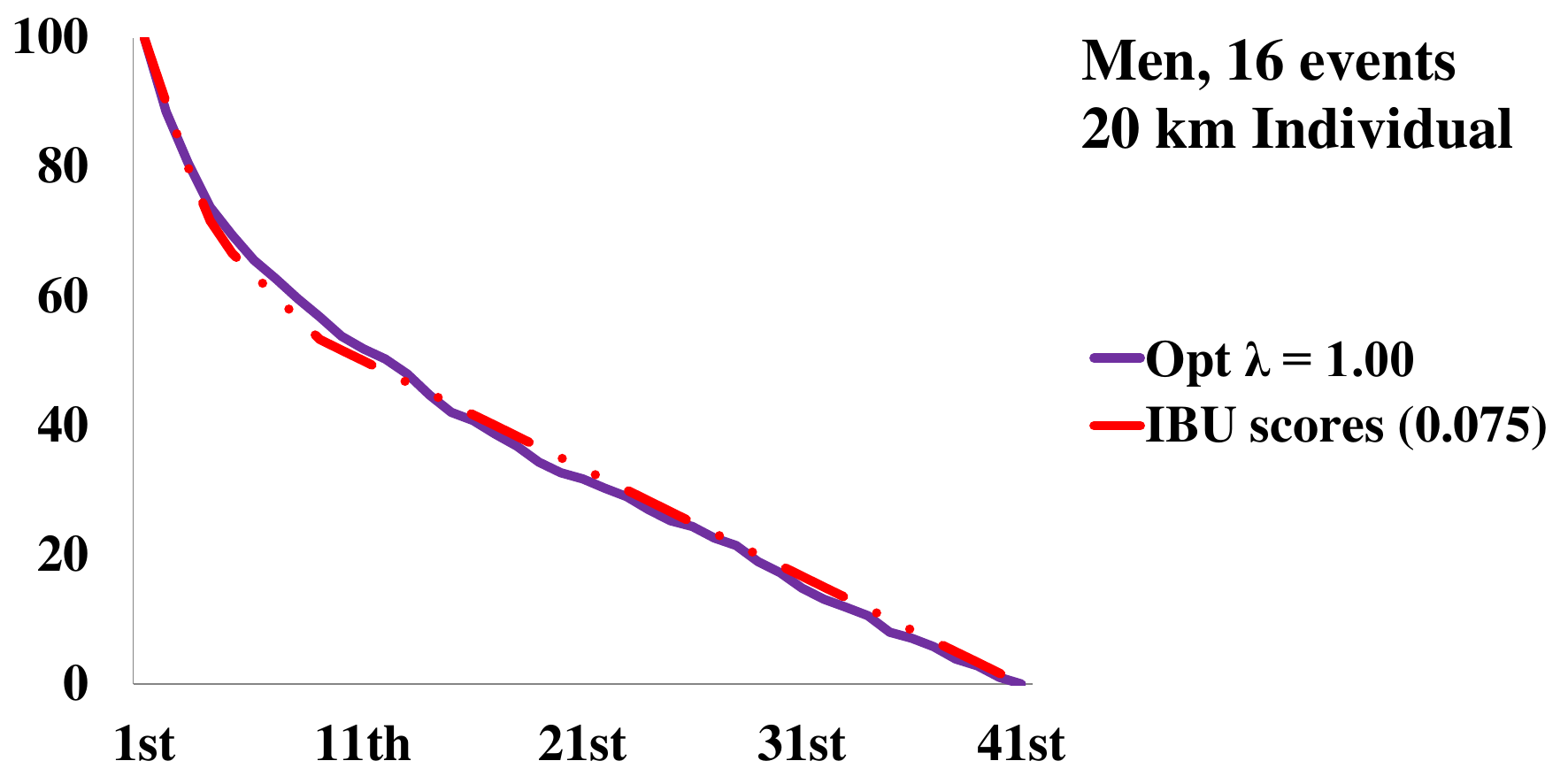}
\includegraphics[width=8cm]{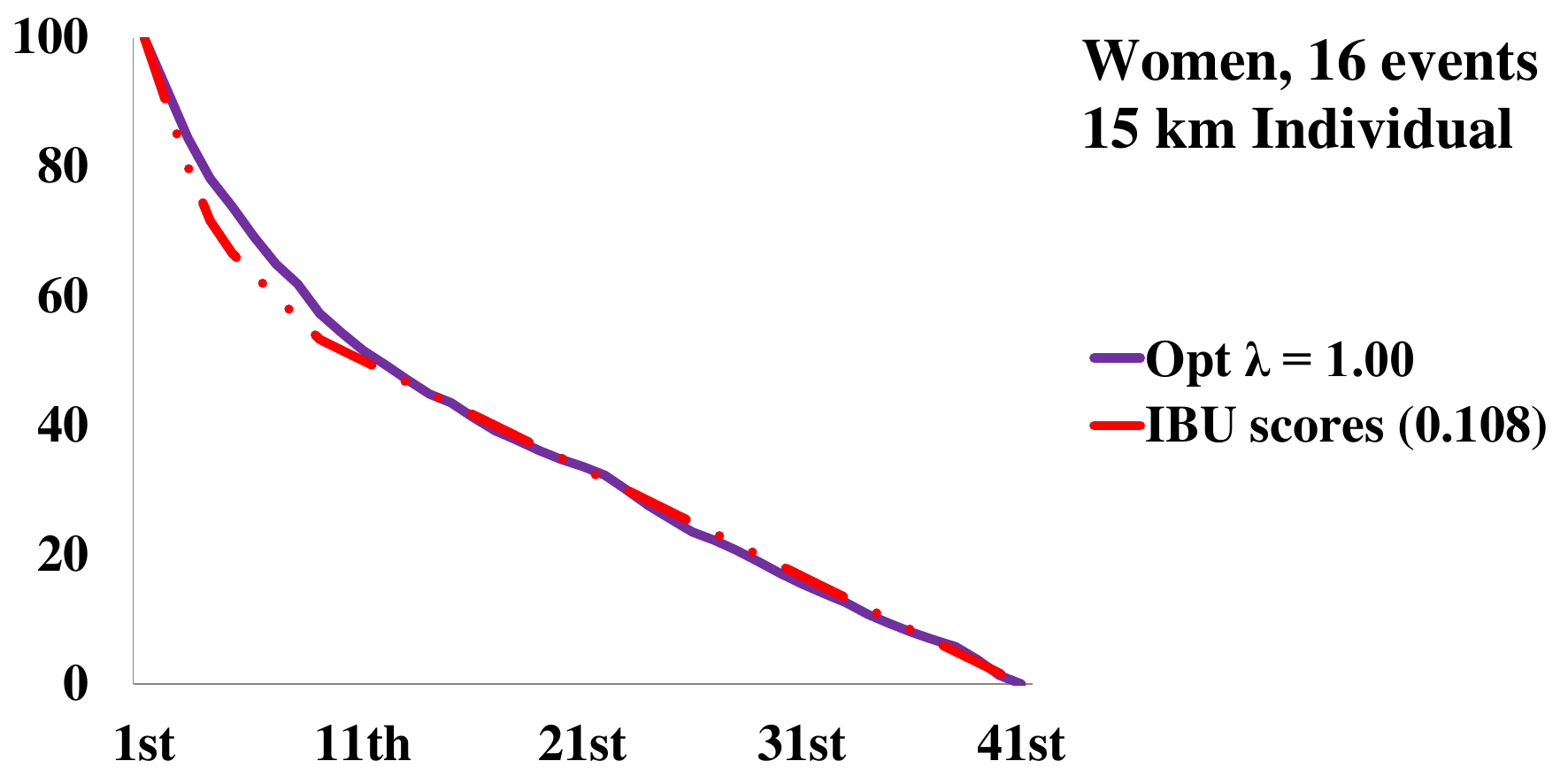}
\includegraphics[width=8cm]{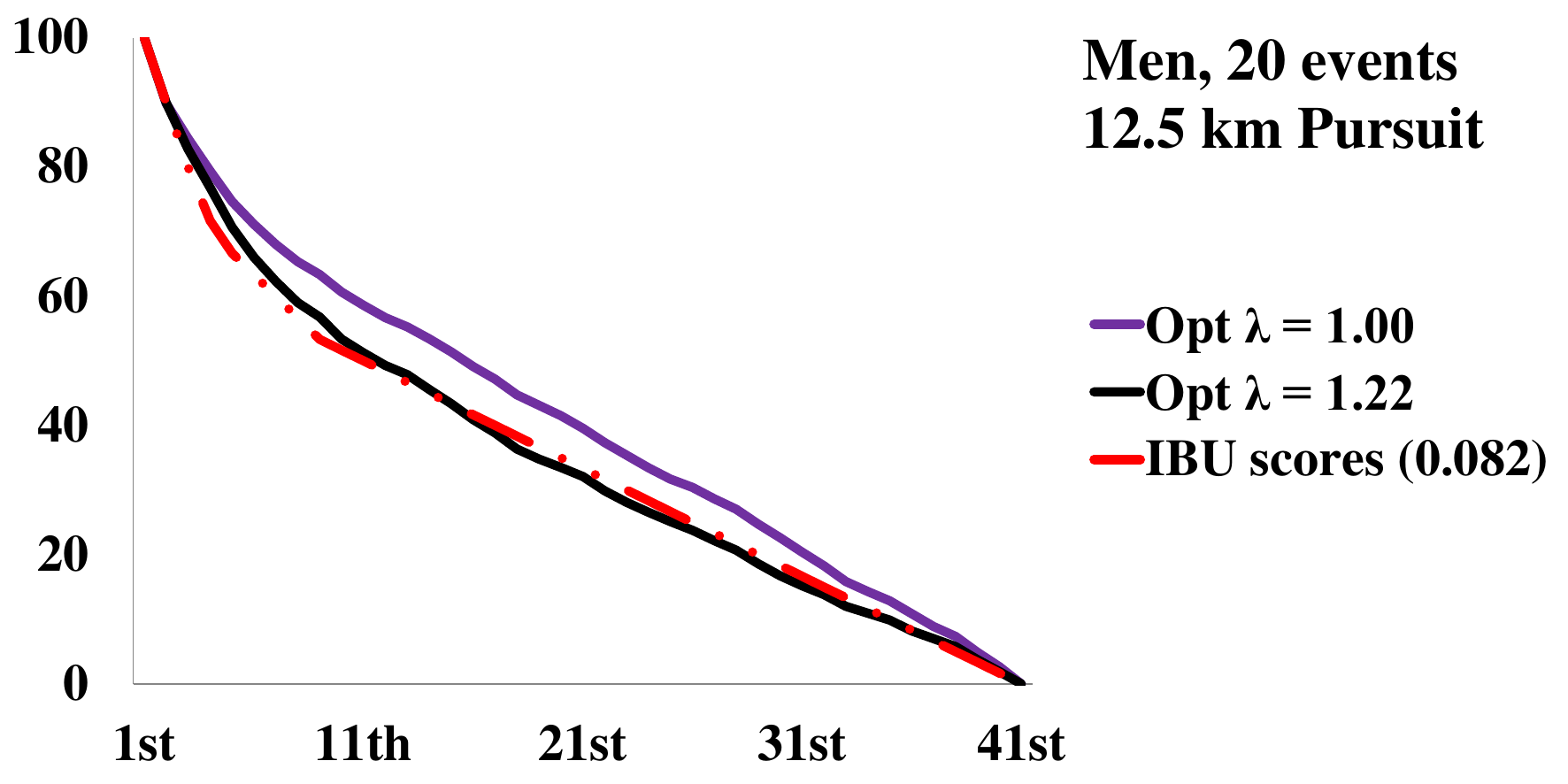}
\includegraphics[width=8cm]{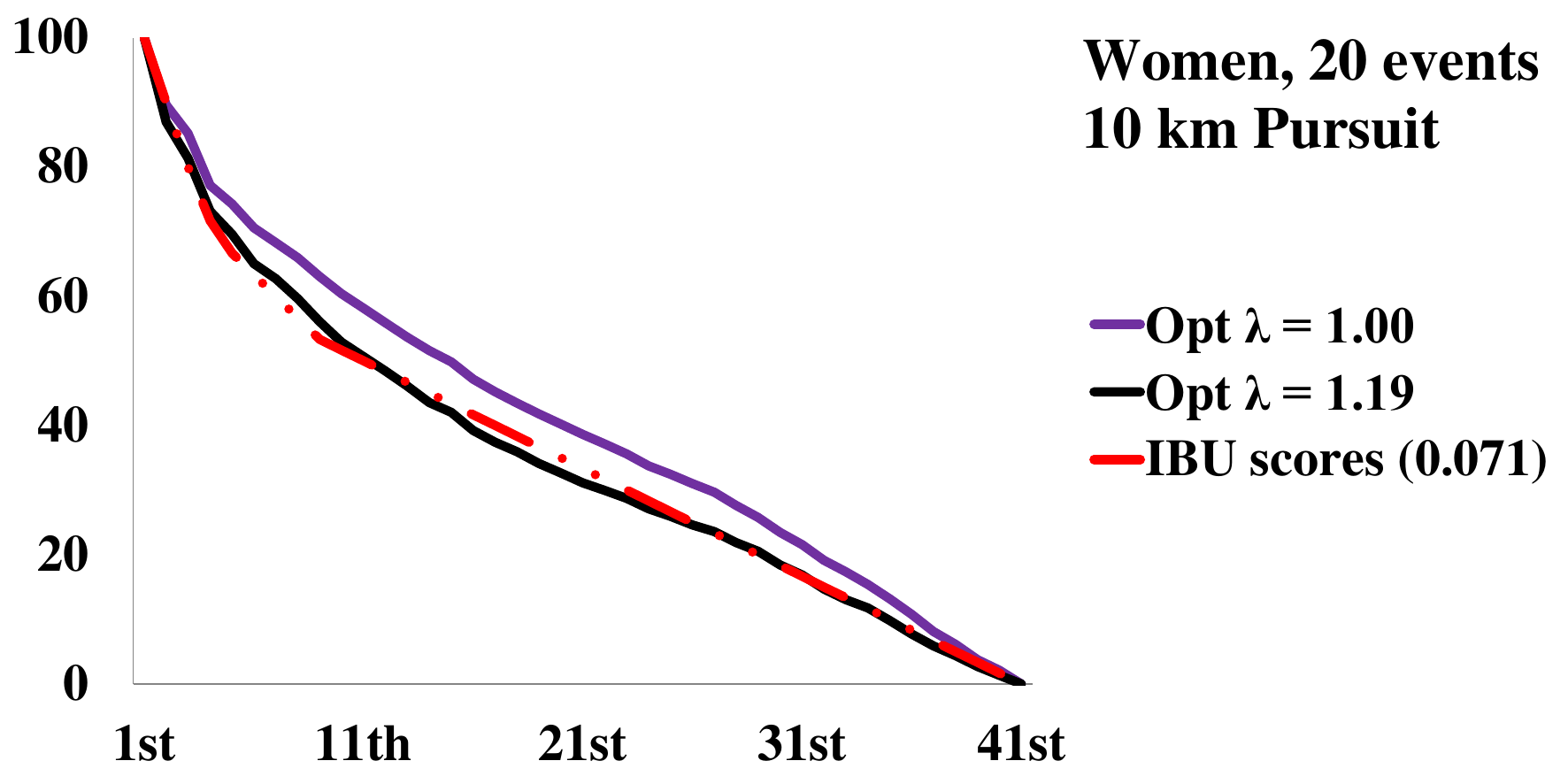}
\includegraphics[width=8cm]{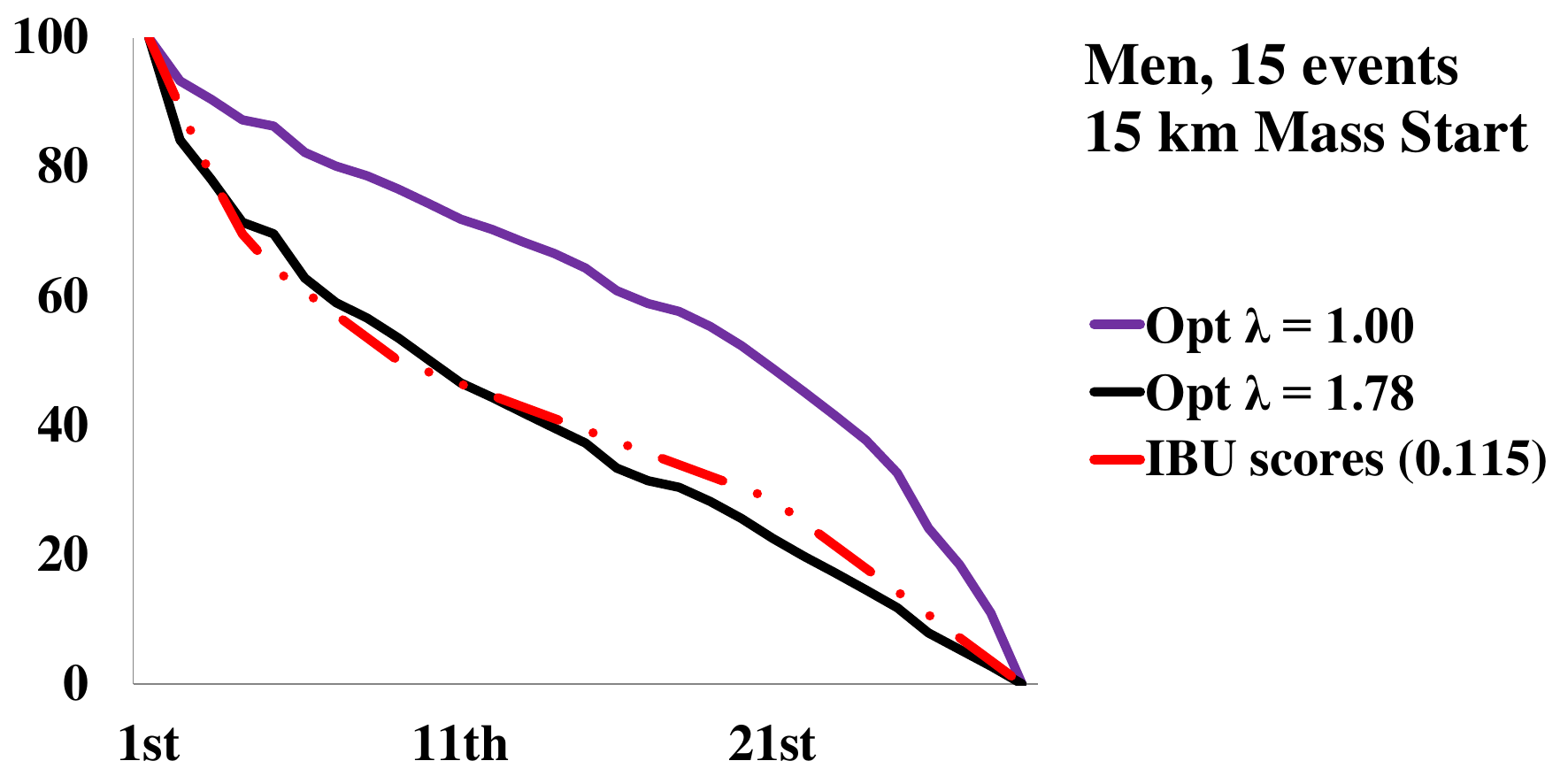}
\includegraphics[width=8cm]{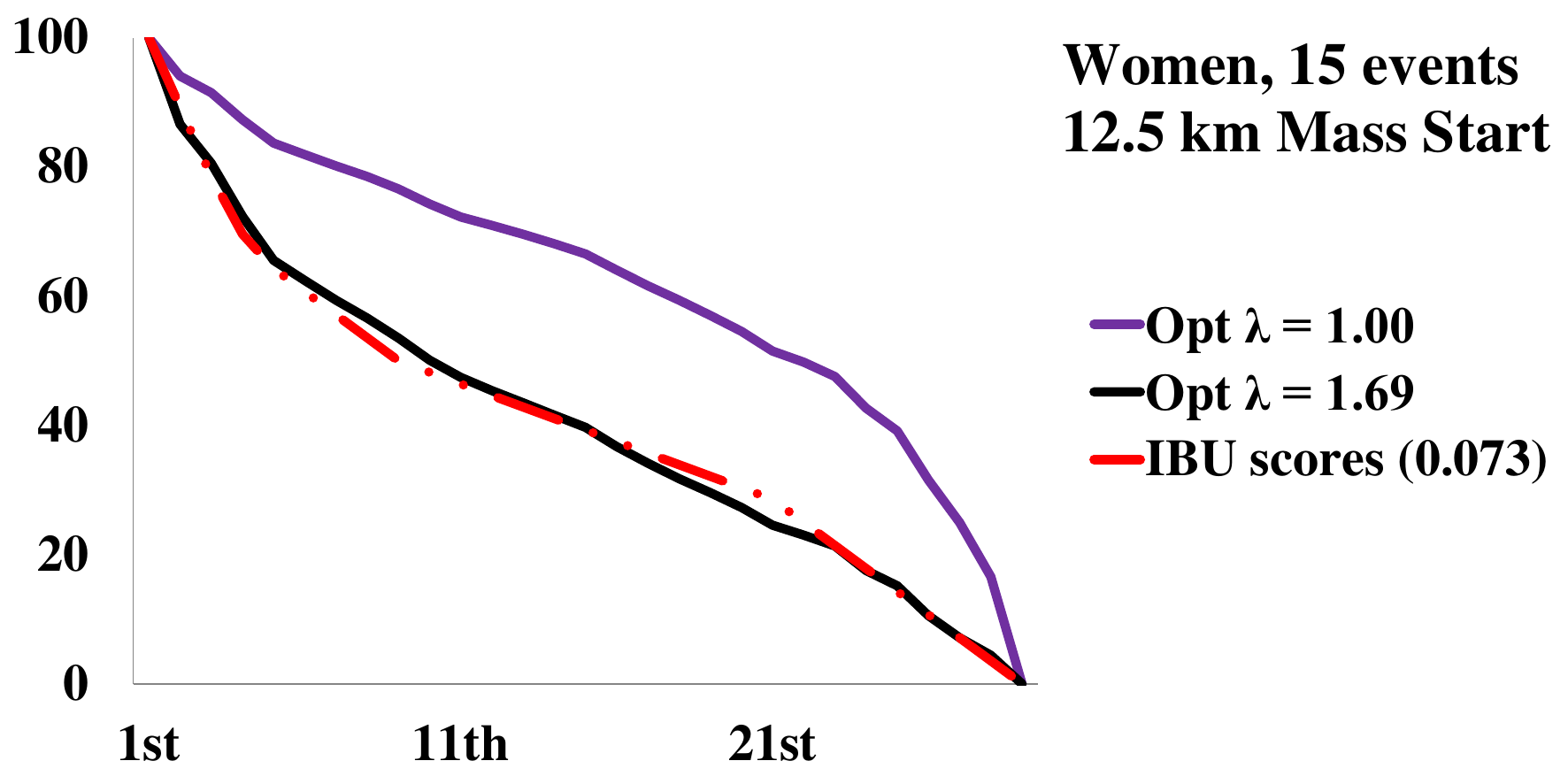}
\end{center}
\caption{Scores and prize money in IBU World Cup biathlon}
\label{OptimalGraphIBU}
\vspace{0.2cm}
\justify
\footnotesize{\emph{Notes}: Scores and prize money used in 2017/18, 2018/19 and 2019/20 seasons compared with the best approximations by geometric and optimal scores. Since there were only 7 Individual races in the three seasons (these figures can be found in \autoref{app:love_sport}), here we present results for 16 Individual races from 2014/15 to 2019/2020 seasons. The $x$-axis is the position, the $y$-axis the normalised score. Scores for first position were normalised to 100, for forty-first (or twenty-ninth in the mass start) position to~0. The optimal scores for $\lambda=1$ (purple solid, higher curve) and $\lambda>1$ (black solid, lower curve, performance measured in minutes) approximate the actual IBU scores used (red long dash two dots). Observe that the best approximations by p = 1.06 (blue dash, higher curve) and p~=~1.24 (brown dash, lower curve) illustrate that the actual IBU prize money awarded (light blue long dash dot) is close to be geometric, while the optimal scores are not. The approximation distance is in brackets and calculated by formula~(\ref{distance}), and denotes the distance to the first curve without brackets above the approximation in the legend.}
\end{figure}

\begin{figure}
\begin{center}
\includegraphics[width=8cm]{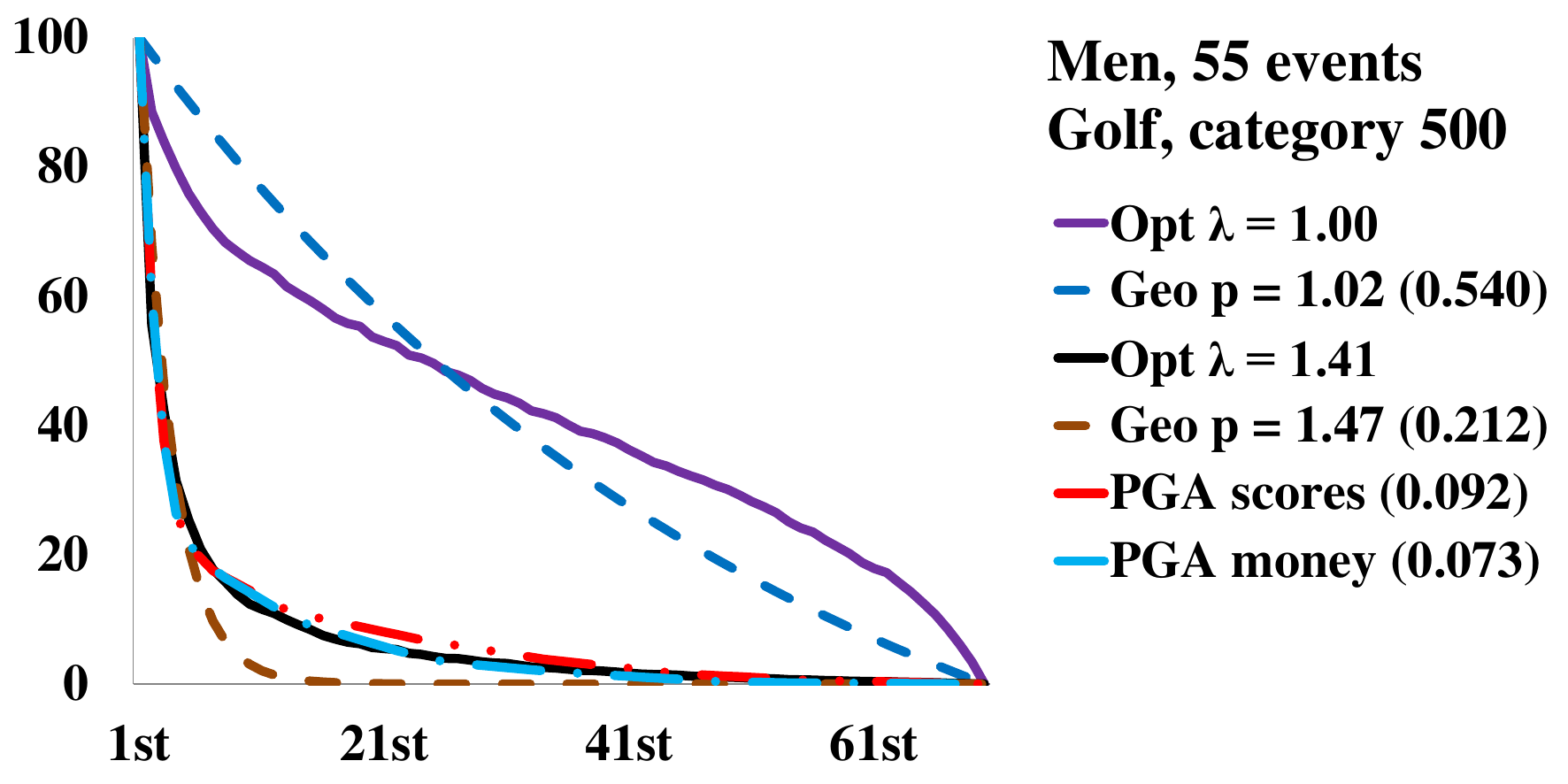}
\includegraphics[width=8cm]{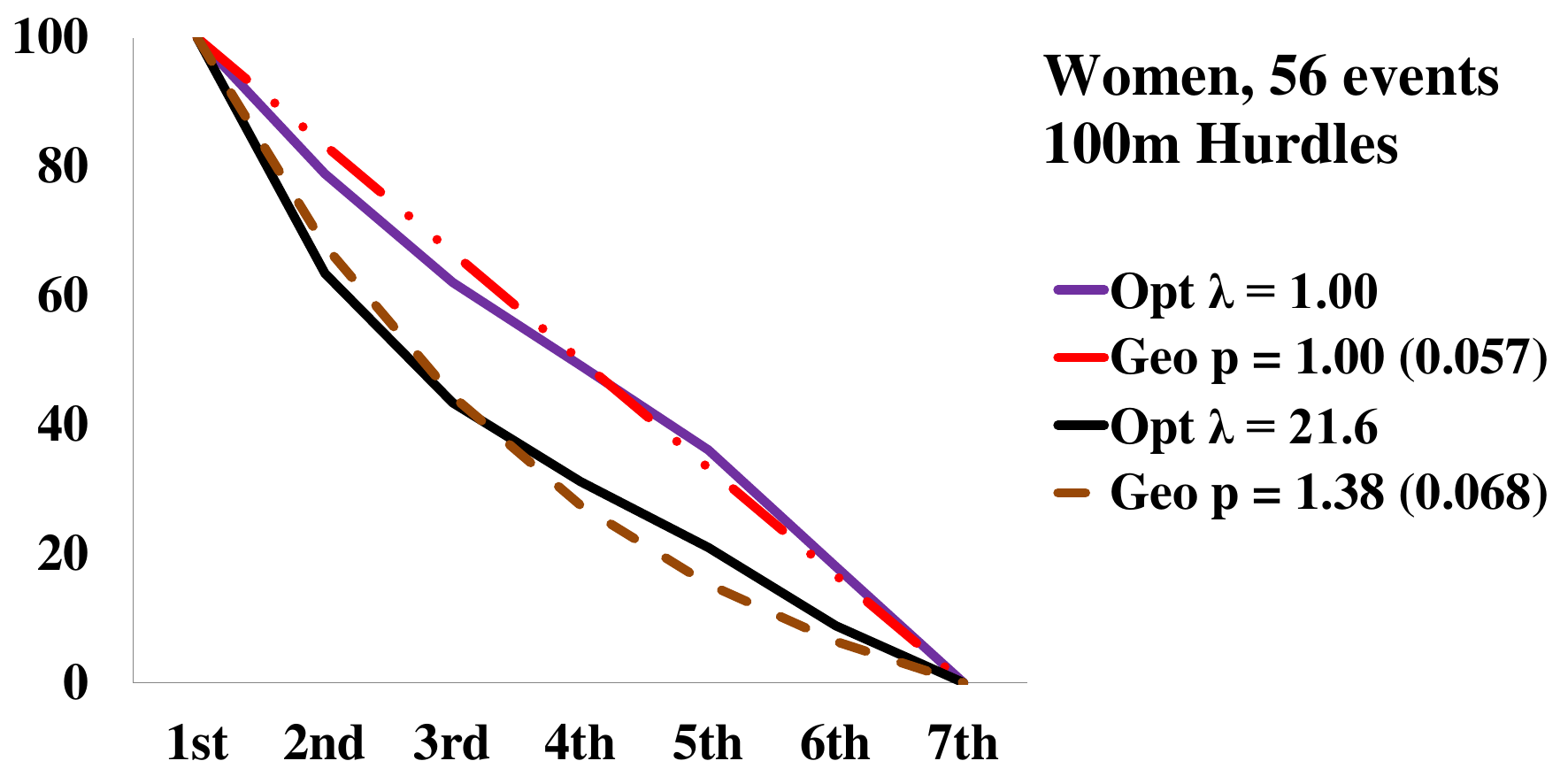}
\includegraphics[width=8cm]{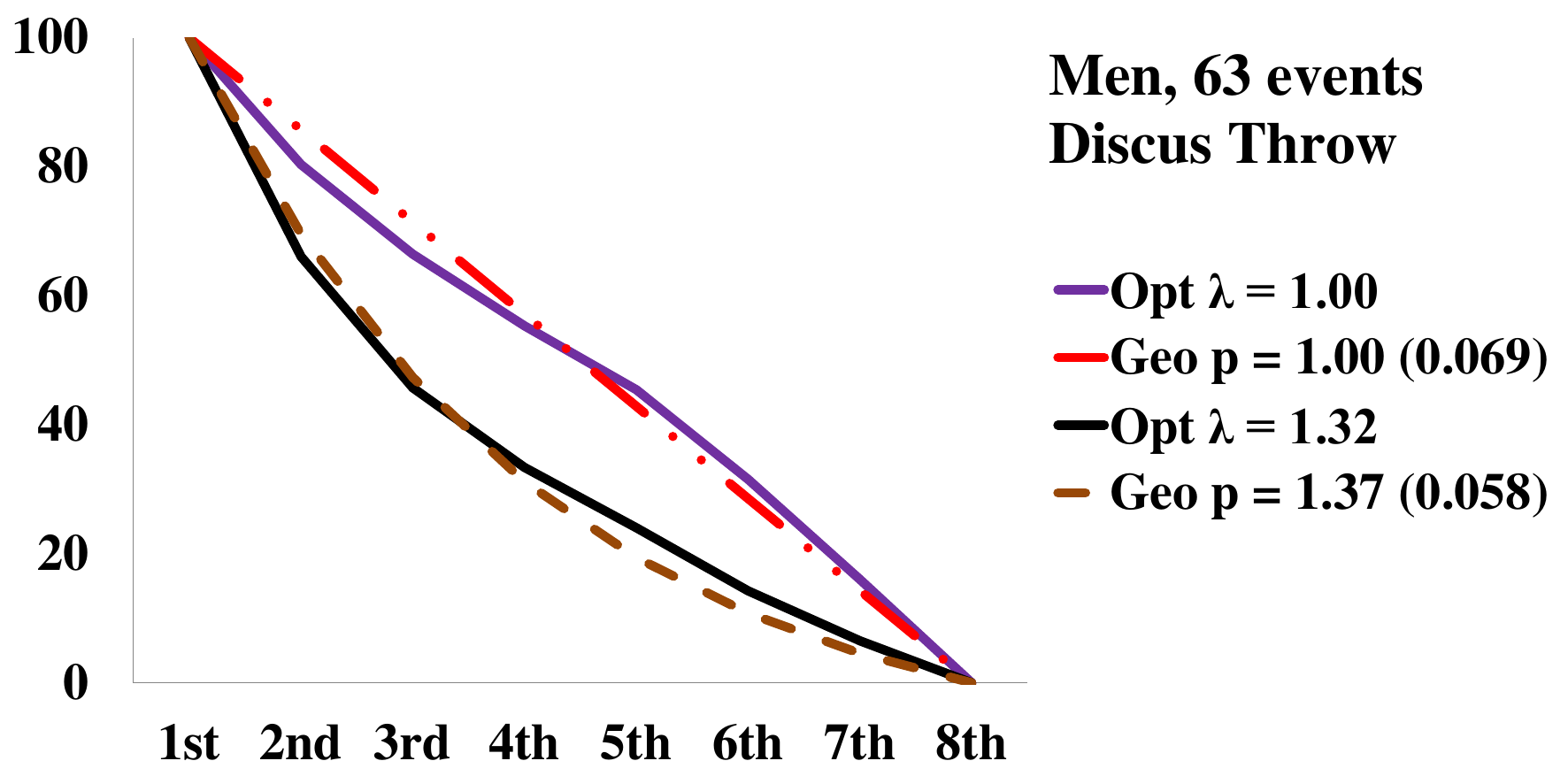}
\includegraphics[width=8cm]{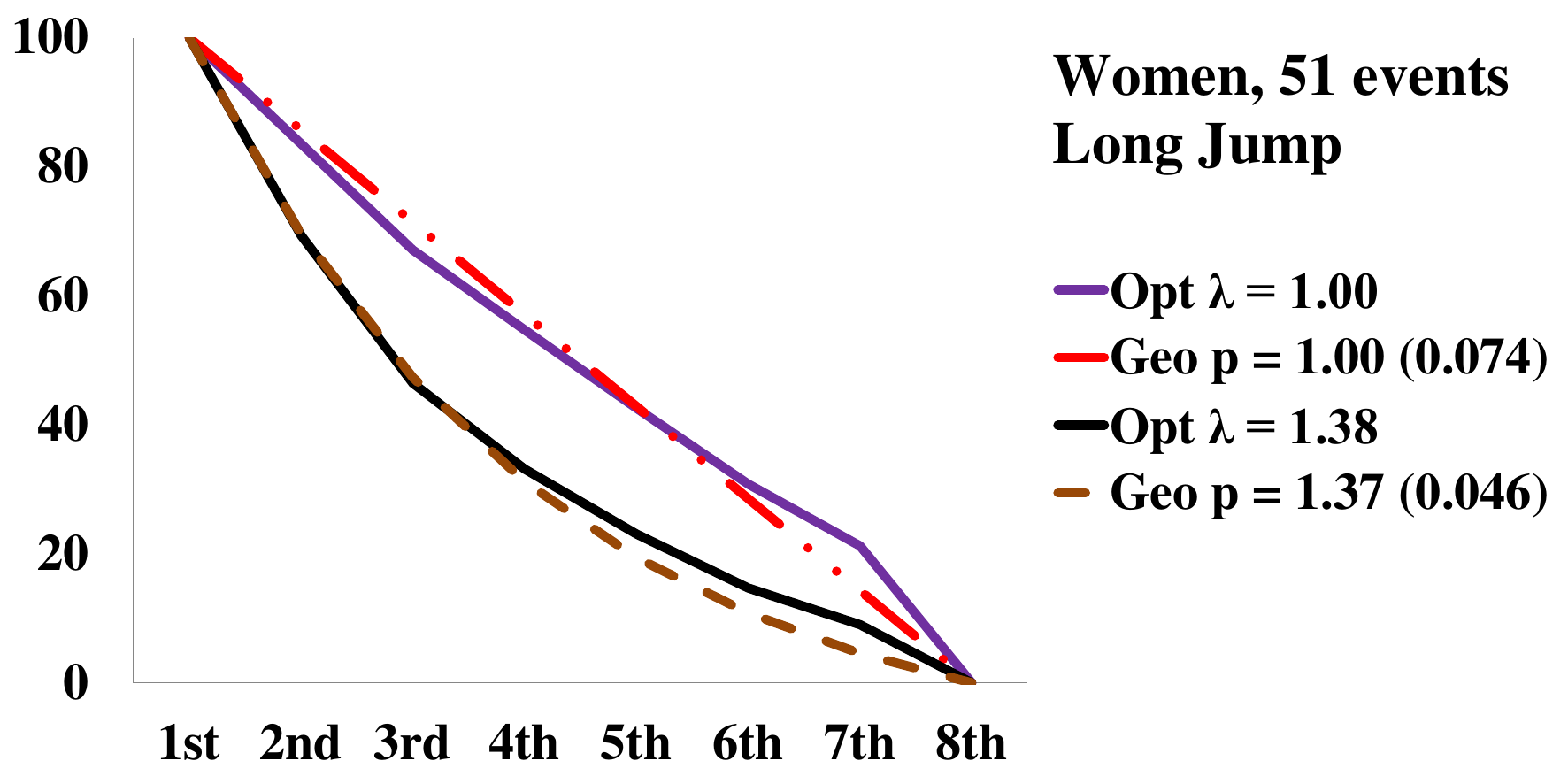}
\end{center}
\caption{Scores in PGA TOUR golf and IAAF Diamond League athletics}
\label{OptimalGraphPGA}
\vspace{0.2cm}
\justify
\footnotesize{\emph{PGA}: Scores and prize money used in 2017/18 and 2018/19 seasons compared with geometric and optimal scores. The $x$-axis is the position, the $y$-axis the normalised score. Scores for first position were normalised to 100, for seventieth position to~0. Observe that the optimal scores for $\lambda=1$ (purple solid, higher curve) illustrate the concave-convex nature of the performance distribution. The optimal scores for $\lambda=1.41$ (black solid, lower curve, performance measured in strokes) closely approximate both the actual PGA scores used (red long dash two dots) and prize money awarded (light blue long dash dot). The best approximations by p = 1.02 (blue dash, higher curve) and p = 1.47 (brown dash, lower curve) illustrate that the optimal scores are far from geometric. The approximation distance is in brackets and calculated by formula~(\ref{distance}), and denotes the distance to the first curve without brackets above the approximation in the legend.\\
\emph{IAAF}: The optimal scores for three athletic disciplines in 2010--2021 seasons approximated by geometric scores. The $x$-axis is the position, the $y$-axis the normalised score. Scores for first position were normalised to 100, for seventh (or eighth) position to~0. The eighth position is excluded to account for the discouragement effect in running (\citealp{Krumer21}). The effect is pronounced in our data, see \autoref{app:love_sport}. Observe that the actual Borda scores used since 2017 (geometric p~=~1, red long dash two dots) closely approximate the optimal scores for $\lambda=1$ (purple solid, higher curve). The curves for $\lambda>1$ (black solid, lower curve, performance measured in seconds for running, metres for throw, and decimetres for jump) illustrate how closely other geometric scores (brown dash) can approximate the optimal scores. The approximation distance is in brackets and calculated by formula~(\ref{distance}). Figures for all 24 analysed athletic disciplines can be found in \autoref{app:love_sport}.}
\end{figure}

We used the following distance measure to find the best approximations. Given a pair of scoring sequences, $s_1,\ldots,s_m$ and $t_1,\ldots,t_m$, we first normalise the scores so that $s_1=t_1=1$ and $s_m=t_m=0$. Then the distance is defined by:
\begin{equation}\label{distance}
    d(s,t)=\sqrt{\frac{1}{4(m-2)}\sum\limits_{j\neq z}(s_j-s_z-t_j+t_z)^2}.
\end{equation}
The factor $1/(4(m-2))$ normalises the distance between plurality and antiplurality to 1. Motivation for such a distance measure can be found in \autoref{app:love_sport}.

\subsection*{Actual scores are optimal in biathlon and golf}

In the sprint and individual categories of biathlon, the scores used are very closely approximated by the optimal scores for $\lambda=1$. In the pursuit category $\lambda=1$ is a passably close fit, but $\lambda=1.22$ and 1.19 for men and women respectively is much better. For the mass start, the scores for $\lambda=1$ are completely off the mark, but $\lambda=1.78$ and 1.69 fit the actual scores well. Geometric scores do a poor job of approximating both the optimal and actual scores, but the prize money is approximately geometric (see Women's Sprint in \autoref{OptimalGraphIBU}).

In golf (\autoref{OptimalGraphPGA}) both the actual scores and prize money are closely approximated by $\lambda=1.41$ (distance 0.092 and 0.073 respectively), while the closest geometric approximation (p = 1.56, distance 0.251 and 0.236) does not come close. The good fit of optimal scores in biathlon and golf is perplexing -- the focal case of $\lambda=1$ is not an issue, one can easily imagine that an organiser took a look at the average finishing times when deciding the scores. But it is at once hard to believe that an organiser decided to raise $1.41$ to the average numbers of strokes and sum the results across historical data, or that the similarity of the scores is a matter of chance. One may be tempted to suppose that the optimal scores are sufficiently flexible to approximate \emph{any} curve with the right choice of $\lambda$, but that is not the case -- if PGA used a geometric sequence with $p=1.02$, then the best approximation with an optimal scoring rule would be $\lambda=0.99$, with a distance of $0.538$. We discuss this phenomenon in the conclusion.

\clearpage

\subsection*{Optimal scores explain phenomena in golf and biathlon}

The resemblance of the scores and prize money in golf to optimal scores with $\lambda=1.4$ also shed light on empirical phenomena in the sport. A single ``race'' in golf (called a tournament) consists of four rounds. In a famous study \citet{Ehrenberg1990} find that a golfer who finishes the first three rounds trailing behind the other competitors is likely to perform poorly in the final round. The authors attribute this to the fact that the marginal monetary return on effort spent for a golfer who can expect to rank low is lower than for a golfer who can expect to rank high, which disincentivises those who are trailing from further effort. But why do marginal returns display this behaviour? The authors argue that this is due to the convexity of the prize structure -- ``the marginal prize received from finishing  second instead of third was 4.0 percent of the total tournament prize money, while the marginal prize received from finishing twenty-second instead of twenty-third was 0.1 percent of the total tournament prize money''. 

We can now see that there is more to the story. What is key here is not the convexity of the scores per se, but how the scores relate to the distribution of the athletes' cardinal performance. Intuitively, one can imagine that the convexity of the rewards is offset by the convexity of athlete performance -- while climbing from the 2nd to the 1st position will net a larger reward than climbing from the 50th to the 49th, climbing from the 50th to the 49th is a lot easier.

Observe that it is possible for optimal scores with parameter $\lambda=1$ to be convex (\autoref{OptimalGraphIBU}), but we argue that had PGA assigned prize money according to $\lambda=1$, we would not observe the effect of \citet{Ehrenberg1990}. Suppose athlete $a$ is performing poorly and knows their final cardinal quality in this race, $x_i^a$, will be low. The athlete must decide whether to accept $x_i^a$, or expend the extra bit of effort to finish with $x_i^a+\varepsilon$. By \autoref{thm:optimal}, at the end of all $n$ races $a$ can expect his total earnings to equal his overall quality -- the sum of $x_1^a,\dots,x_n^a$. If the athlete's performance in the $i$th race is $x_i^a+\varepsilon$ rather than $x_i^a$, this will translate to an expected $\varepsilon$ extra in prize money, regardless of the value of $x_i^a$. On the other hand, with $\lambda=1.4$, the athlete can expect to earn the sum of $1.4^{x^a_1},\ldots,1.4^{x^a_n}$. By putting in the extra effort he can substitute $1.4^{x_i^a+\varepsilon}$ for $1.4^{x_i^a}$, but the extra money here will very much depend on the value of $x_i^a$, and we could expect an athlete that is lagging to not expend the extra effort. 

 Optimal scores also explain the result of \citet{Shmanske07} and \citet{Hood2008}, who observed that golfers with a high variance in the number of strokes earn more than more consistent golfers, even if the mean performance of the consistent golfers is slightly better. The authors attribute this effect to the convexity of the prize money used, but again we claim that such a phenomenon would be absent with $\lambda=1$, regardless of how convex the prize money distribution may be. As a consequence of \autoref{thm:optimal}, we would expect a golfer's earnings to be determined \emph{solely} by their average performance. Variance does not enter into the equation. This reaffirms our interpretation of the choice of $\lambda$ being linked to the organiser's attitude towards peak performance -- by using $\lambda=1.4$ the organisers of the PGA TOUR are willing to reward inconsistent golfers for the possibility of exceptional performance, even if their mean performance suffers.

In the case of biathlon, we observed that the actual scores used resemble optimal scores with $\lambda=1$ in the case of the sprint, and $\lambda>1$ in mass start and pursuit (\autoref{OptimalGraphIBU}). A recent paper of \citet{Gurtler22} studied risk-taking in tournaments and in their interpretation athletes used a riskier strategy in mass start and pursuit than in sprint. This agrees with our interpretation of $\lambda=1$ as giving athletes an incentive to play a consistent strategy, and $\lambda>1$ to motivate them to aim for peak performance, even if it involves the risk of finishing poorly.

\subsection*{Geometric scores are optimal in athletics}

In \autoref{thm:OptimalUniform} we have shown that with a uniform distribution of athlete performance, the optimal scoring rule is approximately geometric. We can see this phenomenon in the data of the IAAF Diamond League (\autoref{OptimalGraphPGA}). The actual scoring rule used since 2017 in these events is Borda, which would be the optimal scoring rule for $\lambda=1$ if the distribution were uniform. In the 24 athletic disciplines studied (see \autoref{app:love_sport}), only in 5 was the distance between Borda and $\lambda=1$ greater than 0.1, and the largest distance was 0.146 (women's high jump, \autoref{OptimalGraphIAAF3} in the appendix). The distribution of athlete performance is remarkably uniform in most events. Presumably, this is because this is a well-understood sport where athletes perform near the limits of human performance -- the athletes are sampled from a very narrow slice of the distribution of possible human performance, and we would expect such a slice to be approximately uniform. To further demonstrate the convergence guaranteed by \autoref{thm:OptimalUniform}, we plot hypothetical curves for scores with a higher value of $\lambda$, and the best geometric approximation, in \autoref{OptimalGraphPGA}. Note that even though the theorem states that the rules converge as the number of athletes tends to infinity, the fit is very good even with $m=7,8$.


\section{Conclusion}

Scoring rules are omnipresent. They are used in group decisions \citep{DyerMiles76}, group recommender systems \citep{Masthoff15}, meta-search engines, multi-criteria selection, word association queries \citep{Dwork01}, sports competitions \citep{Stefani11,Csato21book}, awarding prizes \citep{Benoit92,Stein94,Corvalan18}, arbitrator selections \citep{BloomCavanagh86}, and even for aggregating results from gene expression microarray studies \citep{Lin10}. Many countries use scoring rules in political elections: most of them use plurality, while Slovenia, Nauru and Kiribati use non-plurality scores  \citep{Reilly02,FraenkelGrofman14}.

It is likely that scoring rules are popular because of their simplicity, yet choosing a scoring rule for a specific application is by no means simple. An axiomatic approach simplifies this search by narrowing the scope to the set of rules satisfying a certain combination of properties. In this paper, we establish that:\\
\begin{itemize}
    \item Two natural independence axioms reduce the search to a single parameter family -- the choice of $p$ determines the scores we need (Theorem~\ref{thm:geometricrules}). To our knowledge, this is the first characterisation of a non-trivial family of scoring rules, rather than a specific rule, in the literature.\footnote{\citet{ChebotarevShamis98} provide an extensive overview of previous axiomatisations of scoring rules. While there have been many relaxations of independence in the literature, e.g., independence of Pareto-dominated alternatives or reduction principle by \citet[p.~288]{LuceRaiffa57} and \citet[p.~148]{Fishburn71,Fishburn73book}, and independence of clones by \citet{Tideman87}, for scoring rules they rarely led to positive results. \citet{Richelson78} proved that plurality is the only scoring rule that satisfies independence of Pareto-dominated alternatives, \citet{Ching96} generalised it by dropping one of the redundant axioms, and for a more general framework \citet{Morkeliunas82} proved a theorem which implies the results of Richelson and Ching as corollaries. \citet[theorems~1, 2]{Ozturk20} showed that plurality is the only scoring rule that satisfies weaker versions of independence of clones. Borda's rule is the only scoring rule that satisfies a modified independence defined by \citet{Maskin20}.} This family is sufficiently broad: not only does it include a continuum of convex and concave scores, but also three of the most popular scoring rules: the Borda count, generalised plurality (medal count) and generalised antiplurality (threshold rule).\\
    \item We demonstrate how the choice of the parameter $p$ is constrained by the presence of other desirable axioms. The majority winner criterion pins down generalised plurality (\autoref{thm:generalisedplurality}), top-winner reversal bias -- Borda (\autoref{thm:borda}), and majority loser -- generalised antiplurality (\autoref{thm:concaveGSRTwo}). In \autoref{app:complete}, we provide a full characterisation of these rules among all ordinal ranking procedures.\\
    \item Finally, we consider the choice of $p$ in the context of a sporting competition on historical data. We introduce a model of the organiser's goal, and derive the optimal scoring rules for biathlon (\autoref{OptimalGraphIBU}), golf, and athletics (\autoref{OptimalGraphPGA}). These scores closely resemble the actual scores used by the organisers, and provide an explanation for the phenomena observed by \citet{Ehrenberg1990}, \citet{Shmanske07}, and \citet{Hood2008}. We see that geometric scoring rules approximate the optimal scores well in events where the distribution of athlete's performances is roughly uniform (\autoref{thm:OptimalUniform}).
\end{itemize}

\medskip

Our independence axioms have not received much attention in the literature, perhaps because of how weak they are individually. However, the points incenter \citep{Sitarz13}, best-worst \citep{GarciaLapresta10}, and antiplurality scoring rules violate independence of unanimous losers by \autoref{prop:affinelyequivalent}. In \autoref{app:violation}, we show that Nanson's rule \citep[p.~21]{Nanson1882,FelsenthalNurmi18book}, the proportional veto core \citep{Moulin81}, and even certain generalised scoring rules used in practice, such as average without misery \citep{Masthoff15} and veto-rank \citep{BloomCavanagh86}, also violate independence of unanimous losers.\footnote{Recently, \cite{BarberaCoelho22} showed that the shortlisting procedure \citep{deClippelEliazKnight14} and the voting by alternating offers and vetoes scheme \citep{Anbarci93} violate independence of unanimous losers.} It would be interesting to see where else these axioms can provide some insight. In the weighted version of approval-based multiwinner voting \citep{Thiele1895,Janson18}, if we apply independence of always-approved alternatives (analogous to our independence of unanimous winners), we will obtain geometric sequences of scores which include the top-$k$ rule and a refinement of the Chamberlin–Courant rule as particular cases. Similarly, in the weighted version of approval-based single-winner voting \citep{AlcaldeUnzuVorsatz09}, this axiom will lead to geometric sequences of scores which include approval voting and a refinement of plurality as particular cases. Recently, \citet[theorem~9]{BrandlPeters22} characterised approval voting by independence of never-approved alternatives (analogous to our independence of unanimous losers).

\subsection*{Future directions}

The most striking empirical finding in this paper is the close agreement of optimal scoring rules and the scores used in practice. The case of $\lambda=1$ could be explained away -- it is not a stretch to imagine that an organiser decided to look at average times when deciding on a scoring vector. It is less credible to suppose that an organiser decided to raise $\lambda$ to the power of the result, and take the sum of the outcomes, especially if $\lambda$ takes on seemingly random values like 1.22 and 1.78. To make things worse, consider that the IBU uses only two scoring vectors for eight categories (\autoref{OptimalGraphIBU}); the vectors are optimal in each case, but for different values of $\lambda$ (indeed, different values for men and women). We suspect there is some empirical process going on that causes athlete's results to converge to the scoring vector over time. This is a possibility that should be explored.

The problem of rank aggregation arises in many contexts, but historically the field was largely viewed through the lens of political elections. As a consequence the assumption that we should treat candidates and voters equally -- neutrality and anonymity -- generally goes unquestioned. In a sporting context both are much more demanding suppositions \citep{Stefani11,Csato20ijgs,Csato22unfairness}. Anonymity demands that we weigh every race equally, while there are compelling reasons why we might want to place greater weight on some events than others -- perhaps to recognise their difficulty, or to modulate viewer interest over the course of the championship. Relaxing anonymity raises the question of how we can axiomatise weighted counterparts of geometric scoring rules, and whether our independence axioms can provide additional insight on non-anonymous rules. Neutrality may be perfectly natural when it comes to ranking athletes, but the assumption of symmetric a priori performance of athletes in \autoref{thm:optimal} is a strong one. Clearly some athletes can be expected to perform better than others \citep{Broadie12}, and even the mere presence of an exceptional athlete can be enough to change the performance of the competitors \citep{Brown11}. It would be interesting to see what the optimal ranking rule would be in a more general setting.

Another peculiar feature of many sporting events is that both points and prize money are awarded after each event, and the principles governing the two could be very different. We have seen that, while in golf the scores and prize money are almost identical (\autoref{OptimalGraphPGA}), in biathlon the two are completely different (\autoref{OptimalGraphIBU}). This can lead to the phenomenon where the athlete that earns the most money is not, in fact, the champion.\footnote{In the PGA TOUR, the winner of the 2018 FedEx cup was Justin Rose, while the money leader was Justin Thomas, with 8.694 million to Rose's 8.130 (\url{https://web.archive.org/web/20200318190006/https://www.pgatour.com/news/2018/10/01/2018-2019-pga-tour-full-membership-fantasy-rankings-1-50.html},  \url{https://web.archive.org/web/20180924033715/https://www.pgatour.com/daily-wrapup/2018/09/23/tiger-woods-wins-2018-tour-championship-justin-rose-wins-fedexcup-playoffs.html}).} It would be interesting to see whether such incidents could be avoided, as well as what are other desirable features of prize structures. It does not appear that the axiomatic approach has been applied to prize structures, barring the recent works of \citet{DietzenbacherKondratev22} and \citet{PetroczyCsato21revenue}.

This paper was motivated by sports, where extreme results are valued, so we had little to say about concave geometric rules $(0<p\leq 1)$. An area where they may be of interest is group recommendation systems, where one of the guiding principles is balance between achieving high average utility in the group, and minimising the misery of the least happy member. It is easy to see that Borda ($p=1$) maximises rank-average utility, while generalised antiplurality ($p\ria~0$) minimises the misery of the least happy member. It is natural to suppose that rules with $0<p<1$ will find a middle ground between these two extremes, and it would be interesting to compare them to other procedures for achieving balance, such as average without misery \citep{Masthoff15}, the Nash product \citep{DyerMiles76,Airiau19}, or veto-based approaches \citep{IanovskiKondratev21}.

\bigskip

\printendnotes


\appendix

\section{Proofs}\label{app:proof}

\affinelyequivalent*
\begin{proof}
The result follows from theorem~1 of \citet{Fishburn81DAM} -- that if two scoring vectors are not affinely equivalent, then there exists a profile at which they lead to different rankings. We provide an independent, constructive proof.

The ``if'' part is straightforward. Let us prove the ``only if'' part.

\textbf{Step one:} That the scores are strictly decreasing.

For a fixed $k<m\leq M$, consider a profile $P_k$ consisting of just one race, $a_1\succ\ldots\succ a_k$. By independence of unanimous losers, $a_k$ must come last, so $s_k^k<s_j^k$ for all $j<k$. Moreover, the ranking of $a_1,\dots,a_{k-1}$ must be the same as the ranking in the profile $P_{k-1}$ with the single race $a_1\succ\ldots\succ a_{k-1}$. By independence of unanimous losers, $a_{k-1}$ must come last in $P_{k-1}$, so $a_{k-1}$ must come second-to-last in $P_k$, and thus $s_{k-1}^k<s_j^k$ for all $j<k-1$. By repeating this argument we establish that $s^k_1>\ldots>s^k_k$ for all $k\leq m$.

\textbf{Step two:} That the scores for $k$ athletes are affinely equivalent to the first $k$ scores for $m$ athletes.

For a fixed $k<m\leq M$, consider $s^m_1,\dots,s^m_k$ and $s^k_1,\dots,s^k_k$. Let $\alpha=(s_1^m-s_2^m)/(s_1^k-s_2^k)$ and $\beta=(s_1^k s_2^m-s_2^k s_1^m)/(s_1^k-s_2^k)$. Observe that the scores $\alpha s_1^k+\beta,\dots,\alpha s_k^k+\beta$ are affinely equivalent to $s_1^k,\dots,s_k^k$, and moreover:
\begin{align*}
    \alpha s_1^k+\beta =\frac{s_1^m-s_2^m}{s_1^k-s_2^k}s_1^k +\frac{s_1^k s_2^m-s_2^k s_1^m}{s_1^k-s_2^k} =\frac{s_1^ms_1^k-s_1^ms_2^k}{s_1^k-s_2^k} =s_1^m,\\
    \alpha s_2^k+\beta =\frac{s_1^m-s_2^m}{s_1^k-s_2^k}s_2^k +\frac{s_1^k s_2^m-s_2^k s_1^m}{s_1^k-s_2^k} =\frac{s_2^ms_1^k-s_2^ms_2^k}{s_1^k-s_2^k} =s_2^m.
\end{align*}

For convenience, we write $t_j = \alpha s_j^k+\beta$ for $j=3,\ldots,k$. It remains to show that $t_j=s_j^m$ for $j=3,\ldots,k$ to prove that the scores are affinely equivalent. 

Suppose for contradiction that $t_j > s_j^m$ for some $j$ (the case where $t_j < s_j^m$ is analogous).

Choose integers $n_2>0$ and $n_1$ such that:
\begin{equation}
\frac{s_2^m-t_j}{s_1^m-s_2^m} < \frac{n_1}{n_2} < \frac{s_2^m-s_j^m}{s_1^m-s_2^m}.
\label{eqn:sandwich}
\end{equation}
Let $n = |n_1|+n_2$. Construct a profile with $3n$ races and $m$ athletes as follows.

If $n_1\leq 0$, then in $|n_1|$ races $a$ has position~$1$ and $b$ has position~$2$. In $n_2$ races $a$ has position~$1$ and $b$ has position~$j$. In $n$ races $a$ has position~$2$ and $b$ has position~$j$. In $-n_1$ races $a$ has position~$j$ and $b$ has position~$2$. In $n+n_1$ races $a$ has position~$j$ and $b$ has position~$1$.

If $n_1>0$, then in $n$ races $a$ has position~$1$ and $b$ has position~$j$. In $n_2$ races $a$ has position~$2$ and $b$ has position~$j$. In $n_1$ races $a$ has position~$2$ and $b$ has position~$1$. In $n$ races $a$ has position~$j$ and $b$ has position~$1$. 


In both cases there are $m-k$ athletes who come last in the order $a_{k+1}\succ\ldots\succ a_m$ in every race, and the other athletes are ranked arbitrarily.

Observe than in a profile so constructed athlete $a$ finishes $n$ times in positions $1,2,j$. Athlete $b$ finishes first $n+n_1$ times, second $|n_1|-n_1$ times, and $j$-th $n+n_2$ times. The total score of $a$ is thus $n s_1^m +n s_2^m + n s_j^m$ and the total score of $b$ is $(n+n_1)s_1^m + (|n_1|-n_1)s_2^m + (n+n_2)s_j^m$. The difference between the total scores of $a$ and $b$ is $(n_2 s_2^m - n_2 s_j^m) - (n_1 s_1^m - n_1 s_2^m)$, which is positive by formula~(\ref{eqn:sandwich}). Thus, $a$ beats~$b$.

Now suppose we drop $a_{k+1},\ldots,a_m$ from the races. In the new race, $a$ attains $S_a=n s_1^k + n s_2^k + n s_j^k$ points, and $b$ attains $S_b=(n+n_1)s_1^k + (|n_1|-n_1)s_2^k + (n+n_2)s_j^k$. Clearly, $S_a-S_b > 0$ if and only if $\alpha S_a+3n\beta-(\alpha S_b +3n\beta) > 0$, so we multiply both totals by $\alpha$ and add $3n\beta$. We obtain $n s_1^m + n s_2^m + n t_j$ for $a$, and $(n+n_1)s_1^m + (|n_1|-n_1)s_2^m + (n+n_2)t_j$ for $b$. This gives us a difference of $(n_2 s_2^m - n_2 t_j) - (n_1 s_1^m - n_1 s_2^m)$, which is negative by (\ref{eqn:sandwich}), meaning that dropping the unanimous losers made $b$ overtake~$a$.

The argument for independence of unanimous winners is analogous.
\end{proof}

\bigskip

\alzamoraparadox*

\begin{proof}
We proceed by cases on the value of $p$. 

\textbf{Case one: $p<1$.}

Consider a profile of $n=m-1$ races, $m\geq 3$, where athlete $a$ comes second in every race, and has a total score of $(m-1)(1-p^{m-2})$. Every other athlete comes first, third, fourth, and so on, exactly once. This gives them a total score of $1-p^{m-1}+1-p^{m-3}+\ldots+1-1=m-1-p^{m-1}-(p^{m-2}-1)/(p-1)$. We want to show that the difference between the total scores of $a$ and every other athlete is positive, which is true if and only if:
\begin{align*}
 (m-1)(1-p^{m-2})&>   m-1-p^{m-1}-\frac{p^{m-2}-1}{p-1},\\
 -mp^{m-2}+p^{m-2}+p^{m-1}+\frac{p^{m-2}-1}{p-1}&>0,\\
 mp^{m-1}-mp^{m-2} - p^{m} +1&>0.
\end{align*}

If we take the derivative with respect to $p$, we get $m(m-1)p^{m-2}-m(m-2)p^{m-3}-mp^{m-1}$ that has the same sign as $(m-1)p-(m-2)-p^2$. This is a parabola with vertex at $p=(m-1)/2$ and roots at $1, m-2$. Thus for $0\leq p<1$, this is a monotonely decreasing function, reaching a minimum as $p\ria~1$. At $p=1$, $m\cdot1^{m-1}-m\cdot1^{m-2} - 1^{m} +1=0$, so for the relevant values of $p$ the difference is positive.

\textbf{Case two: $p=1$.}

Consider the profile of case one. Athlete $a$ has a total score of $(m-1)(m-2)$ while the other athletes $(m-1)+(m-3)+\ldots+1= m-1 + (m-3)(m-2)/2$. We want to show that the difference between the total scores is positive:
\begin{align*}
 (m-1)(m-2)&>   m-1 + \frac{(m-3)(m-2)}{2},\\
m^2 -3m +2 &> m-1 + \frac{m^2-5m +6}{2},\\
m^2-3m&>0.
\end{align*}
Which is true for $m\geq 4$.

\textbf{Case three: $p>1$.}

Consider the profile of case one, but with $m>p^2/(p-1)$. Athlete $a$ has a total score of $(m-1)p^{m-2}$, the other athletes $p^{m-1}+p^{m-3}+\ldots+1=p^{m-1}+(p^{m-2}-1)/(p-1)$. We want to show that the difference is positive:
\begin{align*}
 (m-1)p^{m-2}&>   p^{m-1}+\frac{p^{m-2}-1}{p-1},\\
 mp^{m-2}-p^{m-2}-p^{m-1}-\frac{p^{m-2}-1}{p-1}&>   0,\\
 mp^{m-1}-mp^{m-2} -p^{m} +1&>   0,\\
 mp^{m-2}(p-1) -p^{m} +1&>   0.
\end{align*}
Since we assumed that $m>p^2/(p-1)$:
\begin{align*}
    mp^{m-2}(p-1) -p^{m} +1&>   p^{m-2}p^2 -p^{m} +1>0.
\end{align*}
\end{proof}

\bigskip

\generalisedplurality*

\begin{proof}

That generalised plurality satisfies the majority criterion and independence of unanimous winners is straightforward. We shall prove the other direction.

Suppose a generalised scoring rule satisfies independence of unanimous winners and the majority criterion. Fix any $k\leq M$. We proceed by induction on rounds~$r$.

\textbf{Inductive hypothesis:} Suppose for all $l<r$, in the $l$th round the scores are 1 for the first $l$ positions and 0 elsewhere. We will show that in the $r$th round the scores $(s_1^{k,r},\ldots,s_k^{k,r})$ must rank the candidates that made it to the $r$th round in exactly the same order as $(\overbrace{1,\ldots,1}^r, \overbrace{0, \ldots, 0}^{k-r})$. 

For the base case we choose $r=0$, which is satisfied trivially.

In the $r$th round we are concerned with those candidates that were tied in the first $r-1$ rounds, and thus have exactly the same number of first places, second places, through to $(r-1)$th places. If $r=k$ we have a perfect tie, and there is nothing more we can do with scoring rules. Thus, we can assume that $1\leq r\leq k-1$.

Since the relevant candidates have the same number of $l$th places for all $l<r$, we can without loss of generality assume that $s_1^{k,r}=\ldots=s_r^{k,r}$, since the first $r-1$ scores will not change the relative total scores in any way. For convenience, we write $s_j = s_j^{k,r}$ for $j=1,\ldots,k$.

\textbf{Step one:} We shall first show that $s_r>s_j$, for all $j>r$.

Consider a profile consisting of one race, $a_1 \succ \ldots \succ a_k$. By the inductive hypothesis, it is clear that the candidates that made it to the $r$th round are $\{a_r,\ldots,a_k\}$. By repeatedly applying independence of unanimous winners, it follows that the aggregate ranking must be $a_1 \succ \ldots \succ a_k$, so $a_r$ must be ranked first among the remaining candidates. It follows that $s_r\geq s_j$ for all $j>r$.

A round with all scores equal is redundant. Hence, without loss of generality, assume $s_r>s_z$ for some $z>r$.

If $r=k-1$, the scores must be affinely equivalent to $(1,\ldots,1,0)$, and we are done. Assume then that $r<k-1$.

Suppose for contradiction $s_r=s_j$ for some $j>r$. Consider a profile consisting of three races. In all three races $a_{l}$, for $l<r$, is ranked in position $l$. In two races $a_r$ is ranked in position $r$ and $a_j$ in position $j$. In one race $a_r$ is ranked in position $z$ and $a_j$ in position $j$. Since neither $a_j$ nor $a_r$ have any $l$th positions for $l<r$, they have made it through to round $r$. The difference between the total scores of $a_j$ and $a_r$ is positive, $3 s_j - 2 s_r - s_z = s_r - s_z$. Thus, $a_j$ beats $a_r$. However, by applying independence of unanimous winners $r-1$ times we can delete $a_1$ through $a_{r-1}$ without changing the relative ranking of the remaining candidates, but at that point we run into a contradiction because $a_r$ is now the majority winner and should be ranked first. Hence, $s_r>s_j$ for all $j>r$.

\textbf{Step two:} Next, we will show that all other scores are equal.

Suppose for contradiction $s_j > s_z$ for some $j,z>r$. Choose an integer $n > (s_r - s_j) / (s_j - s_z) > 0$. Consider a profile consisting of $2 n + 1$ races. As before, the candidates $a_1\succ \ldots \succ a_{r-1}$ hold the first $r-1$ positions in all races, meaning $a_j$ and $a_r$ have made it to round $r$. In $n+1$ races $a_r$ has position $r$ and $a_j$ has position $j$. In $n$ races $a_r$ has position $z$ and $a_j$ has position~$r$. The difference between the total scores of $a_j$ and $a_r$ is positive, $n s_r + (n+1) s_j - (n+1) s_r - n s_z = s_j - s_r + n (s_j - s_z)>0$. Again we apply independence of unanimous winners and find a contradiction that $a_j$ beats the majority winner $a_r$. It must follow that, $s_j = s_z$ for all $j,z> r$ and the scores $(s_1,\ldots,s_k)$ are affinely equivalent to $(\overbrace{1,\ldots,1}^r, \overbrace{0, \ldots, 0}^{k-r})$.
\end{proof}

\bigskip

\Borda*

\begin{proof}

That Borda satisfies independence of unanimous losers follows from \Cref{thm:geometricrules}. Let us show that the rule satisfies top-winner reversal bias. For a race, if an athlete~$a$ gets $s_j^k=k-j$ points, then for the reversed result of the race the athlete gets $j-1=k-1-(k-j)$ points. Hence, for a profile with $n$ races, if $a$ gets $S_a$ total points, then for the reversed profile the athlete gets $(k-1)n-S_a$ total points. Thus, for a profile, if an athlete is the unique winner and has a higher total score than every other athlete, then for the reversed profile this athlete has a lower total score than every other athlete. We shall show that it is the only scoring rule which has these properties.

Suppose a scoring rule satisfies independence of unanimous losers and top-winner reversal bias. We proceed by induction on the number of athletes~$k$. 

\textbf{Inductive hypothesis:} Suppose that for $k-1$ athletes the scores are $(k-2,\ldots,1,0)$. We will show that for $k$ athletes the scores must be affinely equivalent to $(k-1,\ldots,1,0)$.

In the base case $k=2$, and the only scoring rule which satisfies independence of unanimous losers has scores affinely equivalent to $(1,0)$.

Consider $k\geq 3$. By \Cref{prop:affinelyequivalent}, the scores for $k$ athletes must be affinely equivalent to $(k-1,\ldots, 2, 1, s_k)$, with $s_k<1$.

If $s_k=0$, we are done. We will consider the two cases $s_k>0$ and $s_k<0$, and show that both lead to contradiction.

{\textbf{Case one:} $s_k>0$}

Consider a profile consisting of $2(k-1)$ races. In $k-1$ races $a_1$ finishes first and every other $a_j$ finishes once at every position except for the first position. In the other $k-1$ races the reverse is true -- $a_1$ always finishes last and every other $a_j$ finishes once at every position except for the last position. The total score of $a_1$ is thus $(k-1)^2 + (k-1)s_k$. The total score of every other $a_j$ is $k-2+\ldots+1+s_k+k-1+\ldots +1 = (k-1)^2 + s_k$ which is less than the total score of $a_1$. For the reversed profile the total scores are the same. Hence, $a_1$ wins in both profiles which contradicts top-winner reversal bias.

{\textbf{Case two:} $s_k<0$}

We consider subcases based on whether $k$ is odd or even.

\textit{Subcase one: $k$ is odd.}

Consider a profile consisting of $k-1$ races. Athlete $a_1$ always finishes in the middle position $(k+1)/2$, and every other $a_j$ finishes once at every position except this middle position. The total score of $a_1$ is thus $(k-1)^2/2$. The total score of every other $a_j$ is $k-1+\ldots+1+s_k - (k-1)/2 = (k-1)^2/2 + s_k$ which is less than the total score of $a_1$. For the reversed profile the total scores are the same, contradicting top-winner reversal bias.

\textit{Subcase two: $k$ is even.}

Consider a profile consisting of $2(k-1)$ races. In $k-1$ races $a_1$ finishes in position $k/2$ and every other $a_j$ finishes once at every position except the position $k/2$. The result of other $k-1$ races is the reverse -- $a_1$ always finishes in position $(k+2)/2$ and every other $a_j$ finishes once at every position except position $(k+2)/2$. The total score of $a_1$ is thus $(k-1)^2$. The total score of every other $a_j$ is $2(k-1 + \ldots+1+s_k) -k/2 - (k-2)/2 = (k-1)^2 + 2s_k$ which is less than the total score of $a_1$. For the reversed profile the total scores are the same, again contradicting top-winner reversal bias.

Both cases lead to contradiction. Hence, $s_k=0$ and we get the Borda scores for $k$ athletes.

The proof of the case of independence of unanimous winners is analogous.
\end{proof}

\bigskip

\concaveGSR*

To prove the theorem we will exploit the following auxiliary statement.

\begin{claim}\label{claim:concave}
For each $k\geq 3$ and $p>0$, the function below is strictly increasing in $p$:
\begin{equation*}
    (k-2)p^k-kp^{k-1}+kp-k+2.
\end{equation*}
\end{claim}
\begin{proof}
Let us check the first order condition:
\begin{align*}
    (k-2)kp^{k-1}-k(k-1)p^{k-2}+k&>0,\\
    (k-2)p^{k-1}-(k-1)p^{k-2}+1&>0.
\end{align*}
To see that the above inequality is true for all $p\neq 1$, consider the second derivative:
\begin{align*}
    &(k-2)(k-1)p^{k-2}-(k-1)(k-2)p^{k-3}\\
    =&(k-2)(k-1)p^{k-3}(p-1),
\end{align*}
This is negative for $0<p<1$ and positive for $p>1$. Thus the first derivative is decreasing before hitting 0 at $p=1$, after which it increases -- meaning the first derivative is positive for all $p\neq 1$.
\end{proof}

\begin{proof}[Proof of \autoref{thm:concaveGSR}]
 By \autoref{thm:geometricrules}, we can restrict our attention to geometric scoring rules. We proceed by cases on the value of $p$. 

\textbf{Case one: $p<1$.}

The scores are $1-p^{k-1},\ldots,1-p,1-1$. The average total score is
\begin{equation*}
    \frac{n(k-1-p-p^2-\ldots-p^{k-1})}{k}=n\left(1-\frac{1-p^k}{k(1-p)}\right).
\end{equation*}
A majority loser gets zero points in more than half of the races and hence has a total score lower than $(1-p^{k-1})n/2$. We will show that this is lower than the average total score, and thus the majority loser cannot be ranked first. We wish to show:
\begin{align*}
    \frac{n(1-p^{k-1})}{2}&<n\left(1-\frac{1-p^k}{k(1-p)}\right),\\
    k(1-p)(1-p^{k-1})&<2k(1-p)-2(1-p^k),\\
    k-kp^{k-1}-kp+kp^k&<2k-2kp-2+2p^k,\\
    (k-2)p^k-kp^{k-1}+kp-k+2&<0.
\end{align*}
This is precisely the function from \autoref{claim:concave}. At $p=1$ the function is 0, and elsewhere it is increasing, thus it must be negative for $0<p<1$.

\textbf{Case two: $p=1$.}

The scores are $k-1,\ldots,1,0$. Any majority loser gets zero points in more than half the races and hence has a total score lower than $(k-1)n/2$, which equals to the average total score. This fact was the motivation behind Borda's proposal of his voting system.

\textbf{Case three: $p>1$.}

The scores are $p^{k-1},\ldots,p,1$. For each $k\geq 3$, we can construct a counterexample profile consisting of $2n_k(k-1)+1$ races. In $n_k(k-1)$ races athlete $a_k$ finishes first and every other $a_j$ finishes $n_k$ times at every position except for the first position. In the other $n_k(k-1)$ races the reverse is true -- $a_k$ always finishes last and every other $a_j$ finishes $n_k$ times at every position except for the last position. In the final race the ranking is $a_1,\ldots,a_k$. This will guarantee that $a_1$ is the highest scoring athlete out of $a_1,\dots,a_{k-1}$

We will show that for a large enough $n_k$, the total score of the majority loser $a_k$ is higher than the total score of $a_1$, the best of the other athletes. We wish to show:
\begin{align*}
    n_k(k-1)(p^{k-1}+1)+1&>n_k(p+1)(p^{k-2}+\ldots+1)+p^{k-1},\\
    n_k\left((k-1)(p^{k-1}+1)-\frac{(p+1)(p^{k-1}-1)}{p-1}\right)&>p^{k-1}-1,\\
    n_k\left((k-1)(p^{k-1}+1)(p-1)-(p+1)(p^{k-1}-1)\right)&>(p^{k-1}-1)(p-1),\\
    n_k((k-2)p^k-kp^{k-1}+kp-k+2)&>(p^{k-1}-1)(p-1).
\end{align*}
The coefficient of $n_k$ on the left is positive since this is the function from \autoref{claim:concave}, which is 0 at $p=1$ and increasing elsewhere, and $p>1$ in this case. Thus for a large enough $n_k$ the left hand side will dominate the right.
\end{proof}

\bigskip

\concaveGSRTwo*

\begin{proof}

It is clear that generalised antiplurality satisfies these properties. To see that it is the only generalised scoring rule to do so, suppose $f$ is a generalised scoring rule that satisfies independence of unanimous losers and always ranks the majority loser last. Let $g$ be a generalised scoring rule defined by $g(P)=\mathit{rev} (f(\mathit{rev}(P)))$, where $\mathit{rev}(P)$ is the profile formed by reversing every race result in $P$, and $\mathit{rev}(R)$ is the ranking formed by reversing $R$.

Observe that we can obtain the scoring vector for round $r$ in $g$ by reversing the vector for round $r$ in $f$, and multiplying the entries by -1. The majority winner in $P$ is the majority loser in $\mathit{rev}(P)$ and the unanimous winner in $P$ is the unanimous loser in $\mathit{rev}(P)$, so $g$ satisfies independence of unanimous winners and always ranks the majority winner first. By \autoref{thm:generalisedplurality}, $g$ must be generalised plurality. Thus the scoring vector in round $r$ of $g$ is $(\overbrace{1,\ldots,1}^r, \overbrace{0, \ldots, 0}^{k-r})$, and since $f(P)=\mathit{rev} (g(\mathit{rev}(P)))$, the scoring vector of $f$ is $(\overbrace{0,\ldots,0}^{k-r}, \overbrace{-1, \ldots, -1}^{r})$, which is affinely equivalent to the scores for generalised antiplurality.
\end{proof}

\bigskip

\OptimalUniform*

\begin{proof}
Let $m>4$ and $c=b-a$. We will show that for $\lambda>1$, $\lambda=1$, and $1>\lambda>0$ respectively:
\begin{equation*}
    1<\frac{\EE\left[\lambda^{x^{(j)}}\right]}{\lambda^{a+\frac{c(m+1-j)}{m+1}}}<1+\frac{(c\ln\lambda)^2}{8(m+2)}+\frac{(c\ln\lambda)^3}{18\sqrt{3}(m+2)(m+3)}+\frac{\lambda^c(c\ln\lambda)^4}{128(m+2)(m+4)},
\end{equation*}
\begin{equation*}
    \EE\left[x^{(j)}\right]=a+\frac{c(m+1-j)}{m+1},
\end{equation*}
\begin{equation*}
    1<\frac{\EE\left[-\lambda^{x^{(j)}}\right]}{-\lambda^{a+\frac{c(m+1-j)}{m+1}}}<1+\frac{(c\ln\lambda)^2}{8(m+2)}+\frac{(-c\ln\lambda)^3}{18\sqrt{3}(m+2)(m+3)}+\frac{\lambda^{-c}(c\ln\lambda)^4}{128(m+2)(m+4)}.
\end{equation*}

For the $[0,1]$-uniform distribution, from \citet[p.~36]{DavidNagaraja04} we get:
\begin{align*}
    \EE\left[x^{(j)}\right]=&\frac{m+1-j}{m+1}=p_j,\\
    \EE\left[(x^{(j)}-p_j)^2\right]=&\frac{p_j(1-p_j)}{m+2}\leq\frac{1}{4(m+2)},\\
    \left|\EE\left[(x^{(j)}-p_j)^3\right]\right|=&\left|\frac{2p_j(1-2p_j)(1-p_j)}{(m+2)(m+3)}\right|\leq\frac{\sqrt{3}}{9(m+2)(m+3)},\\
    \EE\left[(x^{(j)}-p_j)^4\right]=&\frac{3p_j^2(1-p_j)^2}{(m+2)^2}\\
    &+\frac{6p_j(1-p_j)}{(m+2)(m+3)(m+4)}\left[(1-2p_j)^2-\frac{(m+3)p_j(1-p_j)}{m+2}\right]\\
    =&\frac{3p_j(1-p_j)\left(p_j(1-p_j)(m-5)+2\right)}{(m+2)(m+3)(m+4)}.
\end{align*}
$\EE\left[(x^{(j)}-p_j)^4\right]$ is maximised at $p_j=1/2$. It follows that:
\begin{equation*}
    \EE\left[(x^{(j)}-p_j)^4\right]\leq\frac{3}{16(m+2)(m+4)}.
\end{equation*}

The $k$-th derivative of $\lambda^x$ is $\lambda^x(\ln\lambda)^k$. If we apply Taylor series expansion with the remainder in Lagrange's form we get:
\begin{multline*}
    \lambda^{x^{(j)}}=\lambda^{p_j}+(x^{(j)}-p_j)\lambda^{p_j}\ln\lambda+\frac{1}{2}(x^{(j)}-p_j)^2\lambda^{p_j}(\ln\lambda)^2+\frac{1}{6}(x^{(j)}-p_j)^3\lambda^{p_j}(\ln\lambda)^3\\+\frac{1}{24}(x^{(j)}-p_j)^4\lambda^{z}(\ln\lambda)^4,
\end{multline*}
for some $z\in[0,1]$. Hence, for $\lambda>1$:
\begin{multline*}
    \lambda^{x^{(j)}}\leq\lambda^{p_j}+(x^{(j)}-p_j)\lambda^{p_j}\ln\lambda+\frac{1}{2}(x^{(j)}-p_j)^2\lambda^{p_j}(\ln\lambda)^2+\frac{1}{6}(x^{(j)}-p_j)^3\lambda^{p_j}(\ln\lambda)^3\\+\frac{1}{24}(x^{(j)}-p_j)^4\lambda(\ln\lambda)^4,
\end{multline*}
\begin{align*}
    \EE\left[\lambda^{x^{(j)}}\right]&\leq\lambda^{p_j}+\frac{\lambda^{p_j}(\ln\lambda)^2}{8(m+2)}+\frac{\lambda^{p_j}(\ln\lambda)^3}{18\sqrt{3}(m+2)(m+3)}+\frac{\lambda(\ln\lambda)^4}{128(m+2)(m+4)},\\
    \frac{\EE\left[\lambda^{x^{(j)}}\right]}{\lambda^{p_j}}&<1+\frac{(\ln\lambda)^2}{8(m+2)}+\frac{(\ln\lambda)^3}{18\sqrt{3}(m+2)(m+3)}+\frac{\lambda(\ln\lambda)^4}{128(m+2)(m+4)}.
\end{align*}

And for $0<\lambda<1$:
\begin{multline*}
    \lambda^{x^{(j)}}\leq\lambda^{p_j}+(x^{(j)}-p_j)\lambda^{p_j}\ln\lambda+\frac{1}{2}(x^{(j)}-p_j)^2\lambda^{p_j}(\ln\lambda)^2+\frac{1}{6}(x^{(j)}-p_j)^3\lambda^{p_j}(\ln\lambda)^3\\+\frac{1}{24}(x^{(j)}-p_j)^4(\ln\lambda)^4,
\end{multline*}
\begin{align*}
    \EE\left[\lambda^{x^{(j)}}\right]&\leq\lambda^{p_j}+\frac{\lambda^{p_j}(\ln\lambda)^2}{8(m+2)}+\frac{\lambda^{p_j}(-\ln\lambda)^3}{18\sqrt{3}(m+2)(m+3)}+\frac{(\ln\lambda)^4}{128(m+2)(m+4)},\\
    \frac{\EE\left[\lambda^{x^{(j)}}\right]}{\lambda^{p_j}}&<1+\frac{(\ln\lambda)^2}{8(m+2)}+\frac{(-\ln\lambda)^3}{18\sqrt{3}(m+2)(m+3)}+\frac{\lambda^{-1}(\ln\lambda)^4}{128(m+2)(m+4)}.
\end{align*}
This establishes the upper bound. For the lower bound, observe that for each $\lambda\neq 1$, $\lambda^x$ is convex and hence:
\begin{equation*}
    \EE\left[\lambda^{x^{(j)}}\right]>\lambda^{\EE\left[x^{(j)}\right]}=\lambda^{p_j}.
\end{equation*}

For the $[a,b]$-uniform distribution, $c=b-a$, and we have
\begin{equation*}
    \frac{\EE\left[\lambda^{a+cx^{(j)}}\right]}{\lambda^{a+cp_j}}=\frac{\EE\left[(\lambda^c)^{x^{(j)}}\right]}{(\lambda^c)^{p_j}},
\end{equation*}
from which the desired bounds follow.
\end{proof}

\bigskip

\pluralityLimit*

To prove the theorem, we will exploit the following auxiliary statements.

\begin{claim}\label{claim:pluralityLimit}
Consider a fixed $m>j\geq 1$. Let $x^1,\ldots,x^m$ be drawn independently and identically from a distribution on $[a,b]$ such that density function $f$ and its first $j$ derivatives are bounded and continuous, $b$ is finite, and $f(b)>0$. Let $x^{(1)}\geq\ldots\geq x^{(m)}$ be their reordering in non-increasing order. Then
\begin{equation*}
    \lim_{\lambda\rightarrow\infty}{\frac{\EE\left[\lambda^{x^{(j+1)}}\right]}{\EE\left[\lambda^{x^{(j)}}\right]}}=0.
\end{equation*}
\end{claim}

\begin{proof}

From \citet[p.~34]{DavidNagaraja04} we have:
\begin{equation*}
    \EE\left[\lambda^{x^{(j)}}\right]=\frac{m!}{(j-1)!(m-j)!}\int_a^b{\lambda^xf(x)(F(x))^{m-j}(1-F(x))^{j-1}dx}.
\end{equation*}
It follows that:
\begin{align*}
    \frac{\EE\left[\lambda^{x^{(j+1)}}\right]}{\EE\left[\lambda^{x^{(j)}}\right]}=\frac{(m-j)}{j}&\left(\frac{\int_a^b{\lambda^xf(x)(F(x))^{m-j-1}(1-F(x))^{j}dx}}{\int_a^b{\lambda^xf(x)(F(x))^{m-j}(1-F(x))^{j-1}dx}}\right)\\
    =\frac{(m-j)}{j}&\left(\frac{\int_a^b{(1-F(x))\lambda^xf(x)(F(x))^{m-j-1}(1-F(x))^{j-1}dx}}{\int_a^b{\lambda^xf(x)(F(x))^{m-j}(1-F(x))^{j-1}dx}}\right)\\
    =\frac{(m-j)}{j}&\left(\frac{\int_a^b{\lambda^xf(x)(F(x))^{m-j-1}(1-F(x))^{j-1}dx}}{\int_a^b{\lambda^xf(x)(F(x))^{m-j}(1-F(x))^{j-1}dx}}\right.\\
    &\quad\left.-\frac{\int_a^b{\lambda^xf(x)(F(x))^{m-j}(1-F(x))^{j-1}dx}}{\int_a^b{\lambda^xf(x)(F(x))^{m-j}(1-F(x))^{j-1}dx}}\right)
    \\
    =\frac{(m-j)}{j}&\left(\frac{\int_a^b{\lambda^xf(x)(F(x))^{m-j-1}(1-F(x))^{j-1}dx}}{\int_a^b{\lambda^xf(x)(F(x))^{m-j}(1-F(x))^{j-1}dx}}-1\right).
\end{align*}

Hence, it is sufficient to show that
\begin{equation*}
    \lim_{\lambda\rightarrow\infty}{\frac{\int_a^b{\lambda^xg_{m-j-1,j-1}(x)dx}}{\int_a^b{\lambda^xg_{m-j,j-1}(x)dx}}}=1,
\end{equation*}
where $g_{k,s}(x)=f(x)(F(x))^k(1-F(x))^s$, for $k=0,\ldots,m-1$, and $s=0,\ldots,m-2$. 

Since $F(b)=1$, we have:
\begin{enumerate}
    \item $g_{k,0}(b)=f(b)$,
    \item for $s>0$, $g_{k,s}(b)=\ldots=g_{k,s}^{(s-1)}(b)=0$,
    \item $g_{k,s}^{(s)}(b)=s!(-1)^s(f(b))^{s+1}$.
\end{enumerate} Integrating by parts $j$ times, 
\begin{align*}
    \int_a^b{\lambda^xg_{k,j-1}(x)dx}&=\frac{\lambda^bg_{k,j-1}(b)}{\ln\lambda}-\frac{\lambda^ag_{k,j-1}(a)}{\ln\lambda}-\frac{1}{\ln\lambda}\int_a^b{\lambda^xg_{k,j-1}^{(1)}(x)dx}\\
    &=\sum\limits_{r=0}^{j-1}{\frac{(-1)^r\lambda^bg_{k,j-1}^{(r)}(b)}{(\ln\lambda)^{r+1}}}+o\left(\frac{\lambda^b}{(\ln\lambda)^j}\right)\\
    &=\frac{\lambda^b(j-1)!(f(b))^j}{(\ln\lambda)^j}(1+o(1)),
\end{align*}
which proves the claim.
\end{proof}

\clearpage

\begin{claim}\label{claim:antipluralityLimit}
Consider a fixed $m>j\geq 1$. Let $x^1,\ldots,x^m$ be drawn independently and identically from a distribution on $[a,b]$ such that density function $f$ and its first $j$ derivatives are bounded and continuous, $a$ is finite, and $f(a)>0$. Let $x^{(1)}\geq\ldots\geq x^{(m)}$ be their reordering in non-increasing order. Then
\begin{equation*}
    \lim_{\lambda\rightarrow 0}{\frac{\EE\left[-\lambda^{x^{(m-j)}}\right]}{\EE\left[-\lambda^{x^{(m-j+1)}}\right]}}=0.
\end{equation*}
\end{claim}

\begin{proof}

For $k=1,\ldots,m$, let $z^k=-x^k$ and $z^{(m-k+1)}=-x^{(k)}$ be their reordering in non-increasing order. By letting $\lambda=1/\alpha$, we have that:
\begin{equation*}
    \lim_{\lambda\rightarrow 0}{\frac{\EE\left[-\lambda^{x^{(m-j)}}\right]}{\EE\left[-\lambda^{x^{(m-j+1)}}\right]}}=\lim_{\alpha\rightarrow\infty}{\frac{\EE\left[\alpha^{-x^{(m-j)}}\right]}{\EE\left[\alpha^{-x^{(m-j+1)}}\right]}}  =\lim_{\alpha\rightarrow\infty}{\frac{\EE\left[\alpha^{z^{(j+1)}}\right]}{\EE\left[\alpha^{z^{(j)}}\right]}}=0,
\end{equation*}
where the last equality is true by \autoref{claim:pluralityLimit}: we have a fixed $m>j\geq 1$; $z^1,\ldots,z^m$ are drawn independently and identically from a distribution on $[-b,-a]$ whose density function $f_z(x)=f(-x)$ and its first $j$ derivatives are bounded and continuous; $-a$ is finite; $f_z(-a)=f(a)>0$; and $z^{(1)}\geq\ldots\geq z^{(m)}$ is the reordering in non-increasing order. 
\end{proof}

\begin{proof}[Proof of \autoref{thm:pluralityLimit}]
 It follows immediately by placing $j=1$ in the claims above and from the fact that whenever the scores are positive and $\frac{s_{j+1}}{s_j}<\frac{1}{n}$ for all $j=1,\ldots,m-1$, the scoring rule is equivalent to generalised plurality, and whenever the scores are negative and $\frac{s_{m-j}}{s_{m-j+1}}<\frac{1}{n}$ for all $j=1,\ldots,m-1$, the scoring rule is equivalent to generalised antiplurality.
\end{proof}

\bigskip

\section{Complete characterisations}\label{app:complete}

Let $A$ be the countable set of potential athletes (either finite or infinite). For a finite set of $m$ athletes, $C\subset A$, a \textbf{profile} $P$ on $C$ is a vector of $m!$ nonnegative integers (indexed by the set of strict rankings on $C$, each such integer denotes the number of races with the corresponding strict ranking of athletes). An \textbf{anonymous ranking rule} $R$ associates with each finite set of athletes $C\subset A$ and each profile $P$ on $C$ a weak ranking $R(P)$ on $C$.

For a bijection $\sigma:C\ria \sigma(C)$ and a weak ranking $L$ on $C$, denote by $\sigma(L)$ the ranking on $\sigma(C)$ that ranks athlete $\sigma(a)$ higher than $\sigma(b)$ if and only if $L$ ranks $a$ higher than $b$. Given finite sets of athletes $C\subset A,\sigma(C)\subset A$, a bijection $\sigma:C\ria \sigma(C)$ and a profile $P$ on $C$, denote by $\sigma(P)$ the profile on $\sigma(C)$ such that the number of races with a strict ranking $L$ in $P$ equals the number of races with strict ranking $\sigma(L)$ in $\sigma(P)$, for all $L$ in~$P$. 

An anonymous ranking rule $R$ satisfies \textbf{neutrality} if $R(\sigma(P))=\sigma(R(P))$, for all finite sets of athletes $C\subset A,\sigma(C)\subset A$, each profile $P$ on $C$, and each bijection $\sigma:C\ria \sigma(C)$. 

An anonymous ranking rule $R$ satisfies \textbf{electoral consistency} if for each finite set of athletes $C\subset A$, each pair of profiles $P$ and $Q$ on $C$ and each pair of athletes $a$ and $b$ from $C$ the next two conditions hold:
\begin{enumerate}
    \item If $R(P)$ and $R(Q)$ rank $a$ higher or equal to $b$ then $R(P+Q)$ ranks $a$ higher or equal to $b$. $P+Q$ is understood as standard vector addition;
    \item If $R(P)$ ranks $a$ higher than $b$ and $R(Q)$ ranks $a$ higher or equal to $b$ then $R(P+Q)$ ranks $a$ higher than $b$.
\end{enumerate}

An anonymous ranking rule $R$ satisfies the \textbf{Archimedean property} if for each finite set of athletes $C\subset A$, each pair of profiles $P$ and $Q$ on $C$, whenever $R(P)$ ranks athlete $a$ higher than $b$, there exists an $n'$ such that $R(nP+Q)$ ranks $a$ higher than $b$, for all integers $n>n'$.

\begin{proposition}
An anonymous ranking rule satisfies neutrality and electoral consistency if and only if it is a generalised scoring rule. 

An anonymous ranking rule satisfies neutrality, electoral consistency and Archimedean property if and only if it is a scoring rule.
\end{proposition}

\begin{proof}
For each finite set of athletes $C\subset A$, we can apply the theorem~1 of \citet{Smith73}. By our definition of neutrality, the scoring vectors will be the same for every set of athletes $C'\subset A$ whenever $|C'|=|C|$.
\end{proof}

Using the proposition above, we can generalise \autoref{prop:affinelyequivalent}, \autoref{thm:geometricrules}, \autoref{thm:generalisedplurality}, \autoref{thm:borda}, \autoref{thm:concaveGSR}, \autoref{thm:concaveGSRTwo}.

\begin{proposition}
An anonymous ranking rule satisfies neutrality, electoral consistency, Archimedean property and independence of unanimous losers if and only if it is a scoring rule with $s^m_1>\ldots>s^m_m$, and the scores for $k$ athletes, $s^k_1,\dots,s^k_k$, are affinely equivalent to the first $k$ scores for $m$ athletes, $s^m_1,\dots,s^m_k$, for all $k<m\leq M$.

An anonymous ranking rule satisfies neutrality, electoral consistency, Archimedean property and independence of unanimous winners if and only if it is a scoring rule with $s^m_1>\ldots>s^m_m$ and the scores for $k$ athletes, $s^k_1,\dots,s^k_k$, are affinely equivalent to the last $k$ scores for $m$ athletes, $s^m_{m-k+1},\dots,s^m_m$, for all $k<m\leq M$.
\end{proposition}

\begin{proposition}
An anonymous ranking rule satisfies neutrality, electoral consistency, Archimedean property, independence of unanimous winners and independence of unanimous losers if and only if it is a geometric scoring rule with parameter $0 < p <\infty$.
\end{proposition}

\begin{proposition}
Generalised plurality is the only anonymous ranking rule that satisfies neutrality, electoral consistency, independence of unanimous winners and always ranks the majority winner first.
\end{proposition}

\begin{proposition}
Borda is the only anonymous ranking rule that satisfies neutrality, electoral consistency, Archimedean property, top-winner reversal bias and one of independence of unanimous winners or independence of unanimous losers.
\end{proposition}

\begin{proposition}
Geometric scoring rules with parameter $0 < p\leq 1$ are the only anonymous ranking rules that satisfy neutrality, electoral consistency, Archimedean property, independence of unanimous winners and independence of unanimous losers and never rank the majority loser first.
\end{proposition}

\begin{proposition}
Generalised antiplurality is the only anonymous ranking rule that satisfies neutrality, electoral consistency, independence of unanimous losers and always ranks the majority loser last.
\end{proposition}

\bigskip

\clearpage

\section{Violation of independence of unanimous losers}\label{app:violation}

In this section we show that some well-known ordinal procedures do not satisfy independence of unanimous losers.

\textbf{Nanson's rule} (\citeyear{Nanson1882}) eliminates candidates round by round. In each round the candidates with more than the average Borda scores proceed to the next round, until the remaining candidates get equal Borda scores and are declared the winners. Consider a profile where 7 individual rankings are $acdbe$ (from first-ranked $a$ to the last-ranked $e$), 7 -- $bacde$, 7 -- $cdbae$, 1 -- $bcade$, and 1 -- $bacde$. In the first round, $a$ gets 61 points, $b$ -- 57, $c$ -- 68, $d$ -- 44, $e$ -- 0, the average score is 46, and thus $a, b$, and $c$ proceed to the second round, where $a$ gets 22 points, $b$ -- 25, $c$ -- 22, the average score is 23 and thus $b$ wins. However, if we remove the unanimous loser $e$, then in the first round $a$ gets 38 points, $b$ -- 34, $c$ -- 45, $d$ -- 21, the average score is 34.5 and hence $b$ cannot win anymore.

The \textbf{proportional veto core} is defined as follows by \citet{Moulin81}. For a profile with $n$ voters and $m$ candidates, a candidate $a$ is \textbf{blocked} if there exists a coalition of $t$ voters and a subset of $k$ candidates such that each voter in the coalition ranks each candidate in the subset higher than $a$, and $n(m-k)<mt$. All candidates that are not blocked are declared the winners. Consider a profile with $n=3$ voters and $m=3$ candidates, where 2 individual rankings are $bac$, and 1 -- $abc$. It is easy to verify that $a$ is not blocked. However, if we remove the unanimous loser $c$, then $a$ is blocked, because $t=2$ voters rank $k=1$ candidates ($b$) higher than $a$, $m=2$ and $n(m-k)=3\cdot(2-1)<2\cdot 2=mt$.

For the case of $n=2$ voters and $m=2k+1$ candidates, the \textbf{veto-rank} used in arbitrator selection \citep{BloomCavanagh86} can be defined as a generalised scoring rule. It assigns 1 point for the first $k+1$ positions, and 0 points for the last $k$ positions. The tie-breaking scores are $k,k-1,\ldots,1,0,\ldots,0$. This rule also can be seen as an ordinal variant of average without misery used in group recommendations \citep{Masthoff15}. Consider a profile with individual rankings $abcdefg$ and $decbafg$. The best candidates in the first round, $b,c$, and $d$, get 2 points, and in the tie-breaking round $b$ and $c$ get 2 points, whereas $d$ gets 3 points and thus wins. However, if we remove the unanimous losers $g$ and $f$, then in the first round only $c$ gets 2 points and thus wins.

\bigskip

\section{Computing optimal scores (Online Appendix)}\label{app:love_sport}

\subsection{Data description}\label{app:DataDescription}

The results for the IBU World Cup biathlon are based on the 2017/18, 2018/19 and 2019/20 seasons, however since there were only 7 Individual races in the three
seasons, we use the 2014/15 to the 2019/2020 seasons for the Individual category. The data was downloaded from \url{https://www.biathlonworld.com}.  The actual IBU scores used for sprint, pursuit, and individual are 60, 54, 48, 43, 40, 38, 36, 34, 32, 31, $\dots$, 1, then 0 for the remaining positions. The scores used for the mass start are 60, 54, 48, 43, 40, 38, 36, 34, 32, 31, 30, $\dots$, 22, 21, 20, 18, 16, $\dots$, 2. The actual IBU prize-money (in euros) awarded in 2019/20  for the first twenty positions is 15,000, 12,000, 9,000, 7,000, 6,000, 5,000, 4,000, 3,500, 3,000, 2,500, 2,000, 1,750, 1,500, 1,250, 1,000, 900, 800, 700, 600, 500, and then 0 for the remaining positions. Since at least 29 biathletes completed each mass start, we restricted ourselves to 29 positions in this category and 41 positions in other categories. In race $i$, the cardinal performance $x_i^a$ of biathlete $a$ was calculated as their lag behind the race winner in minutes.

The results of Category 500 golf events of the PGA TOUR in 2017/18 (29 events) and 2018/19 (26 events) seasons were downloaded from \url{https://www.pgatour.com}. The actual PGA scores for first seventy positions are 500, 300, 190, 135, 110, 100, 90, 85, $\ldots$, 60, 57, 55, $\ldots$, 37, 35.5, $\ldots$, 22, 21, $\ldots$, 11, 10.5, $\ldots$, 6, 5.8, $\ldots$ 3. The actual PGA prize-money (in percent of the total purse) is 18, 10.9, 6.9, 4.9, 4.1, 3.625, 3.375, 3.125, 2.925, $\ldots$, 1.925, 1.825, $\ldots$, 1.125, 1.045, 0.965, 0.885, 0.805, 0.775, $\ldots$, 0.595, 0.57, 0.545, 0.52, 0.495, 0.475, $\ldots$, 0.295, 0.279, 0.265, 0.257, 0.251, 0.245, 0.241, 0.237, 0.235, $\ldots$, 0.205. We restricted ourselves to seventy positions because at least seventy competitors completed each event. In event $i$, the cardinal performance $x_i^a$ of competitor $a$ was taken as his lag behind the event winner in the number of strokes.

The results of twenty four athletic disciplines of the IAAF Diamond League in the 2010--2021 seasons were downloaded from \url{https://www.diamondleague.com}. The 2020 season contains no data because of the COVID pandemic. We dropped the results of long and triple jump, shot put, discus and javelin throw in the last season because after the rule change in 2021 the final ranking of the top three athletes in each event reflects only one last attempt and thus can be different from the order of their best attempts. We dropped events where less than eight athletes finished or where the result for the eighth position was lower than the standard for a ``Candidate for Master of Sport'' under the Unified Sports Classification System of Russia (there were three of these: 6.41 and 5.99 metres for  men's long jump, 5.32 metres for women's long jump). The descriptive statistics are in \autoref{table:descriptive}

The actual IAAF scores since 2017 are 8, 7, 6, 5, 4, 3, 2, 1, 0 (no points in season finales); in 2016 the vector 10, 6, 4, 3, 2, 1, 0 was used (double points in the season finale); and in 2010--2015 the vector 4, 2, 1, 0 (double points in season finales). In each event $i$, the cardinal performance $x_i^a$ of athlete $a$ was taken as their final result in seconds (running), metres (throw and put), and decimetres (jump and vault).

\begin{table}
\centering
\caption{IAAF Diamond League descriptive statistics}
\begin{tabular}{rccccc}
\toprule
&Total&Included&Unit of&Mean&Standard\\
Discipline&events&events&measure&result&deviation\\
\hline
men 100&77&65&seconds&10.06&0.128\\
women 100&72&60&seconds&11.11&0.167\\
men 200&72&58&seconds&20.34&0.318\\
women 200&74&56&seconds&22.76&0.381\\
men 400&72&56&seconds&45.22&0.584\\
women 400&71&53&seconds&51.16&0.897\\
men 110H&72&55&seconds&13.30&0.185\\
women 100H&72&56&seconds&12.78&0.192\\
men 400H&72&57&seconds&49.07&0.769\\
women 400H&72&58&seconds&55.05&0.968\\
men high jump&73&68&decimetres&22.82&0.470\\
women high jump&71&62&decimetres&19.20&0.498\\
men pole vault&73&58&decimetres&56.75&1.508\\
women pole vault&74&57&decimetres&46.06&1.369\\
men long jump&63&56&decimetres&80.46&2.034\\
women long jump&64&51&decimetres&66.68&1.816\\
men triple jump&64&51&decimetres&169.0&4.506\\
women triple jump&64&46&decimetres&142.9&3.672\\
men shot put&61&54&metres&20.92&0.734\\
women shot put&60&36&metres&18.67&0.801\\
men discus throw&64&63&metres&64.62&2.121\\
women discus throw&64&53&metres&62.73&2.893\\
men javelin throw&64&56&metres&83.12&3.531\\
women javelin throw&64&50&metres&62.25&2.917\\
\hline
All disciplines&1649&1335&&&\\
\bottomrule
\end{tabular}
\label{table:descriptive}
\vspace{0.2cm}
\justify 
\footnotesize{\textit{Notes}: Mean and standard deviation are calculated for positions from first to seven.}
\end{table}

In the case of the running disciplines, we restricted our analysis to the first seven positions. This is due to the discouragement effect \citep[][p.~525]{Ehrenberg1990,Krumer21,Frick03}, according to which athletes reduce their efforts when they perceive they are lagging behind the leaders.  To check the presence of this effect in our data, in \autoref{table:Running8Gaps} and \autoref{table:Running9Gaps} we calculated the average and median gaps between adjacent positions. When calculating the average (but not median) gaps, we ignored results of the eighth and ninth positions that were lower than the standard for a Candidate for Master of Sport. These were: 11.59, 12.08, 12.46 (men's 100m), 23.35, 23.69, 26.30, 28.80, 80.88 (men's 200m), 62.69 (men's 400m), 15.64, 16.82, 19.26 (men's 110m hurdles), 15.32, 15.80, 15.85, 20.81, 21.62, 25.11 (women's 100m hurdles), 78.90 (men's 400m hurdles), 65.78, 90.61 (women's 400m hurdles). Without ignoring such results, the average gaps between the last two positions would have been even larger. Both the average and median gaps are almost symmetric around the middle positions, with the only exception that the gap between the last two positions is about twice as large as the gap between the first two positions. Thus we can confirm that the discouragement effect is pronounced in our data. From \autoref{table:Running8GapsBefore2016}, \autoref{table:Running8GapsSince2017}, and \autoref{table:Running8and9MedianGapsRatio}, we can observe that the effect has intensified since 2017. The ratio of the last and second to last median gaps have increased in most cases, from two to three on average. At this stage we have insufficient data to speculate whether this is linked to the rule change in 2017.

\begin{table}
\centering
\caption{Running: 8 athletes, 2010--2021, median and average gaps}
\begin{tabular}{rcccccccc}
\toprule
Discipline&Included&\multicolumn{7}{c}{Adjacent positions: median gap in seconds}\\
&events&1-2&2-3&3-4&4-5&5-6&6-7&7-8\\
\hline
men 100&40&0.050&0.040&0.030&0.050&0.030&0.030&0.070\\
women 100&39&0.050&0.060&0.060&0.040&0.030&0.060&0.080\\
men 200&45&0.150&0.070&0.100&0.090&0.110&0.100&0.220\\
women 200&42&0.135&0.095&0.160&0.065&0.140&0.090&0.285\\
men 400&43&0.240&0.190&0.090&0.160&0.170&0.190&0.400\\
women 400&44&0.305&0.270&0.255&0.180&0.285&0.255&0.620\\
men 110H&32&0.070&0.045&0.050&0.050&0.050&0.085&0.150\\
women 100H&39&0.060&0.050&0.040&0.040&0.050&0.060&0.140\\
men 400H&43&0.200&0.250&0.220&0.210&0.200&0.270&0.590\\
women 400H&46&0.310&0.310&0.400&0.210&0.285&0.360&0.975\\
\hline
All disciplines&413&10.76&8.85&8.95&7.70&8.54&9.77&21.48\\
\bottomrule
\end{tabular}
\begin{tabular}{rcccccccc}
\\
\toprule
Discipline&\multicolumn{7}{c}{Adjacent positions: average gap in seconds}&Total\\
&1-2&2-3&3-4&4-5&5-6&6-7&7-8&gap\\
\hline
men 100&0.066&0.044&0.032&0.052&0.047&0.046&0.100&0.386\\
women 100&0.077&0.076&0.070&0.050&0.040&0.067&0.131&0.511\\
men 200&0.170&0.117&0.110&0.113&0.128&0.137&0.226&1.000\\
women 200&0.200&0.141&0.155&0.093&0.159&0.136&0.327&1.210\\
men 400&0.294&0.247&0.190&0.193&0.211&0.272&0.540&1.947\\
women 400&0.387&0.376&0.274&0.243&0.338&0.358&0.894&2.870\\
men 110H&0.109&0.073&0.066&0.065&0.059&0.100&0.188&0.660\\
women 100H&0.102&0.068&0.059&0.055&0.079&0.089&0.258&0.710\\
men 400H&0.313&0.403&0.264&0.296&0.243&0.348&0.762&2.629\\
women 400H&0.429&0.437&0.402&0.274&0.347&0.623&0.923&3.435\\
\hline
All disciplines&14.93&12.42&10.50&9.73&10.81&13.54&28.06&100\\
\bottomrule
\end{tabular}
\label{table:Running8Gaps}
\vspace{0.2cm}
\justify 
\footnotesize{\textit{Notes}: The normalised mean for all disciplines was calculated after each value in our data was multiplied by 100 and divided by the total gap.}
\end{table}

\begin{table}
\centering
\caption{Running: 9 athletes, 2010-2021, median and average gaps}
\begin{tabular}{rccccccccc}
\toprule
Discipline&Included&\multicolumn{8}{c}{Adjacent positions: median gap in seconds}\\
&events&1-2&2-3&3-4&4-5&5-6&6-7&7-8&8-9\\
\hline
men 100&25&0.050&0.040&0.020&0.030&0.020&0.010&0.030&0.090\\
women 100&21&0.050&0.080&0.050&0.030&0.030&0.030&0.060&0.080\\
men 200&13&0.100&0.080&0.050&0.050&0.060&0.080&0.070&0.150\\
women 200&14&0.110&0.115&0.085&0.115&0.125&0.055&0.125&0.260\\
men 400&13&0.270&0.140&0.200&0.090&0.150&0.220&0.210&0.740\\
women 400&9&0.180&0.300&0.200&0.170&0.090&0.080&0.180&0.810\\
men 110H&23&0.050&0.050&0.040&0.050&0.030&0.060&0.050&0.170\\
women 100H&17&0.070&0.050&0.040&0.040&0.060&0.030&0.050&0.140\\
men 400H&14&0.220&0.150&0.080&0.155&0.195&0.150&0.605&0.475\\
women 400H&12&0.490&0.415&0.205&0.175&0.240&0.345&0.130&0.560\\
\hline
All disciplines&161&9.94&9.04&6.24&6.14&6.46&6.21&9.35&21.40\\
\bottomrule
\end{tabular}
\begin{tabular}{rccccccccc}
\\
\toprule
Discipline&\multicolumn{8}{c}{Adjacent positions: average gap in seconds}&Total\\
&1-2&2-3&3-4&4-5&5-6&6-7&7-8&8-9&gap\\
\hline
men 100&0.067&0.051&0.042&0.033&0.031&0.037&0.043&0.095&0.399\\
women 100&0.068&0.095&0.076&0.040&0.048&0.044&0.080&0.111&0.562\\
men 200&0.154&0.149&0.065&0.115&0.112&0.096&0.135&0.188&1.015\\
women 200&0.151&0.136&0.114&0.126&0.150&0.099&0.192&0.291&1.259\\
men 400&0.328&0.160&0.188&0.156&0.177&0.258&0.307&0.639&2.214\\
women 400&0.408&0.290&0.413&0.226&0.199&0.250&0.406&0.950&3.141\\
men 110H&0.098&0.061&0.050&0.060&0.054&0.062&0.077&0.215&0.678\\
women 100H&0.068&0.086&0.049&0.061&0.077&0.055&0.081&0.143&0.621\\
men 400H&0.351&0.246&0.181&0.200&0.248&0.214&0.579&0.489&2.508\\
women 400H&0.547&0.433&0.365&0.216&0.272&0.403&0.322&0.933&3.490\\
\hline
All disciplines&13.90&11.68&9.40&8.40&9.16&9.21&13.71&24.54&100\\
\bottomrule
\end{tabular}
\label{table:Running9Gaps}
\vspace{0.2cm}
\justify 
\footnotesize{\textit{Notes}: The normalised mean for all disciplines was calculated after each value in our data was multiplied by 100 and divided by the total gap.}
\end{table}

\begin{table}
\centering
\caption{Running: 8 athletes, 2010--2016, median and average gaps}
\begin{tabular}{rcccccccc}
\toprule
Discipline&Included&\multicolumn{7}{c}{Adjacent positions: median gap in seconds}\\
&events&1-2&2-3&3-4&4-5&5-6&6-7&7-8\\
\hline
men 100&21&0.040&0.040&0.040&0.040&0.040&0.030&0.060\\
women 100&25&0.050&0.060&0.060&0.040&0.030&0.060&0.070\\
men 200&29&0.140&0.090&0.110&0.090&0.120&0.100&0.230\\
women 200&26&0.115&0.070&0.120&0.060&0.150&0.065&0.240\\
men 400&34&0.240&0.190&0.110&0.155&0.160&0.190&0.345\\
women 400&28&0.265&0.265&0.260&0.195&0.315&0.310&0.570\\
men 110H&23&0.070&0.040&0.050&0.050&0.060&0.090&0.130\\
women 100H&28&0.065&0.050&0.035&0.035&0.045&0.055&0.140\\
men 400H&29&0.180&0.190&0.170&0.190&0.220&0.190&0.590\\
women 400H&31&0.390&0.280&0.290&0.230&0.310&0.460&0.870\\
\hline
All disciplines&274&10.47&8.54&8.66&7.45&9.48&9.83&20.06\\
\bottomrule
\end{tabular}
\begin{tabular}{rcccccccc}
\\
\toprule
Discipline&\multicolumn{7}{c}{Adjacent positions: average gap in seconds}&Total\\
&1-2&2-3&3-4&4-5&5-6&6-7&7-8&gap\\
\hline
men 100&0.065&0.051&0.037&0.050&0.053&0.055&0.064&0.375\\
women 100&0.079&0.074&0.074&0.050&0.042&0.072&0.130&0.523\\
men 200&0.159&0.122&0.117&0.121&0.125&0.145&0.254&1.043\\
women 200&0.170&0.109&0.144&0.088&0.167&0.111&0.260&1.049\\
men 400&0.287&0.261&0.213&0.195&0.189&0.259&0.511&1.914\\
women 400&0.352&0.302&0.293&0.269&0.387&0.431&0.930&2.965\\
men 110H&0.106&0.073&0.073&0.054&0.067&0.103&0.161&0.636\\
women 100H&0.099&0.068&0.054&0.062&0.076&0.078&0.260&0.698\\
men 400H&0.242&0.283&0.237&0.241&0.262&0.305&0.849&2.420\\
women 400H&0.488&0.408&0.348&0.306&0.365&0.719&0.967&3.602\\
\hline
All disciplines&14.52&11.80&10.86&9.79&11.55&14.09&27.39&100\\
\bottomrule
\end{tabular}
\label{table:Running8GapsBefore2016}
\vspace{0.2cm}
\justify 
\footnotesize{\textit{Notes}: The normalised mean for all disciplines was calculated after each value in our data was multiplied by 100 and divided by the total gap.}
\end{table}

\begin{table}
\centering
\caption{Running: 8 athletes, 2017--2021, median and average gaps}
\begin{tabular}{rcccccccc}
\toprule
Discipline&Included&\multicolumn{7}{c}{Adjacent positions: median gap in seconds}\\
&events&1-2&2-3&3-4&4-5&5-6&6-7&7-8\\
\hline
men 100&17&0.060&0.040&0.020&0.050&0.030&0.040&0.080\\
women 100&12&0.055&0.070&0.060&0.040&0.030&0.055&0.095\\
men 200&13&0.160&0.050&0.080&0.060&0.100&0.060&0.180\\
women 200&14&0.210&0.190&0.180&0.075&0.110&0.110&0.360\\
men 400&9&0.300&0.140&0.090&0.200&0.180&0.150&0.590\\
women 400&13&0.230&0.290&0.200&0.170&0.290&0.160&0.500\\
men 110H&7&0.110&0.050&0.030&0.090&0.030&0.100&0.420\\
women 100H&10&0.075&0.055&0.060&0.035&0.060&0.075&0.225\\
men 400H&11&0.400&0.510&0.240&0.430&0.200&0.310&0.630\\
women 400H&12&0.155&0.565&0.445&0.125&0.285&0.195&1.020\\
\hline
All disciplines&118&12.29&10.94&8.45&8.23&7.94&8.75&26.76\\
\bottomrule
\end{tabular}
\begin{tabular}{rcccccccc}
\\
\toprule
Discipline&\multicolumn{7}{c}{Adjacent positions: average gap in seconds}&Total\\
&1-2&2-3&3-4&4-5&5-6&6-7&7-8&gap\\
\hline
men 100&0.066&0.038&0.024&0.055&0.038&0.038&0.145&0.404\\
women 100&0.078&0.069&0.062&0.050&0.037&0.058&0.132&0.487\\
men 200&0.198&0.128&0.090&0.082&0.136&0.135&0.169&0.938\\
women 200&0.251&0.188&0.175&0.088&0.156&0.179&0.397&1.434\\
men 400&0.319&0.196&0.100&0.186&0.296&0.322&0.649&2.067\\
women 400&0.354&0.517&0.242&0.198&0.282&0.193&0.744&2.530\\
men 110H&0.126&0.054&0.057&0.091&0.047&0.104&0.317&0.797\\
women 100H&0.120&0.062&0.073&0.035&0.083&0.126&0.271&0.770\\
men 400H&0.564&0.604&0.287&0.419&0.221&0.380&0.659&3.134\\
women 400H&0.269&0.499&0.531&0.179&0.344&0.433&0.780&3.036\\
\hline
All disciplines&15.87&13.08&9.83&9.08&10.28&12.73&29.12&100\\
\bottomrule
\end{tabular}
\label{table:Running8GapsSince2017}
\vspace{0.2cm}
\justify 
\footnotesize{\textit{Notes}: The normalised mean for all disciplines was calculated after each value in our data was multiplied by 100 and divided by the total gap. We included only the events with actual Borda scores used, i.e., season finales were excluded from this table.}
\end{table}

\begin{table}
\centering
\caption{Running: 8 and 9 athletes, median gaps between three last positions}
\begin{tabular}{rcccc|cccc}
\toprule
&\multicolumn{4}{c}{2010-2016}&\multicolumn{4}{c}{2017-2021}\\
&Included&Second&Last&&Included&Second&Last&\\
Discipline&events&to last&gap&Ratio&events&to last&gap&Ratio\\
\hline
men 100&31&0.030&0.090&3.000&24&0.025&0.075&3.000\\
women 100&33&0.060&0.080&1.333&18&0.055&0.095&1.727\\
men 200&31&0.100&0.190&1.900&17&0.060&0.180&3.000\\
women 200&29&0.070&0.280&4.000&18&0.105&0.345&3.286\\
men 400&36&0.190&0.465&2.447&13&0.180&0.590&3.278\\
women 400&30&0.310&0.575&1.855&16&0.205&0.760&3.707\\
men 110H&31&0.070&0.140&2.000&15&0.050&0.330&6.600\\
women 100H&32&0.050&0.140&2.800&15&0.070&0.240&3.429\\
men 400H&34&0.230&0.575&2.500&14&0.330&0.555&1.682\\
women 400H&34&0.465&0.660&1.419&15&0.200&0.790&3.950\\
\hline
All disciplines&321&&&2.207&165&&&3.140\\
\bottomrule
\end{tabular}
\label{table:Running8and9MedianGapsRatio}
\vspace{0.2cm}
\justify 
\footnotesize{\textit{Notes}: The ratio for all disciplines is computed as the geometric mean. All season finales were excluded from this table.}
\end{table}

\newpage

\subsection{Data analysis}\label{app:DataAnalysis}

We used $u_i^a=\lambda^{x_i^a}$ for $\lambda>1$, $u_i^a=x_i^a$ for $\lambda=1$ and $u_i^a=-\lambda^{x_i^a}$ for $0<\lambda<1$ as the measure of quality of athlete $a$ in event $i$. In each event $i$, these qualities were reordered in non-increasing order $u_i^{(1)},\ldots,u_i^{(m)}$, i.e., $u_i^{(j)}$ is quality of the athlete that finished at position~$j$. By Theorem~\ref{thm:optimal}, the optimal scores are the expectations of the corresponding random variables, so we estimate them according to the sample mean: $s_j=(u_1^{(j)}+\ldots+u_n^{(j)})/n$ for each $j=1,\ldots,m$.

In the proof of \autoref{thm:OptimalUniform} we bound the ratio of the optimal scores to the geometric approximation by a formula of $m,\lambda$, and $b-a$ from above. If we were to substitute the extreme values of the men's 100m ($\lambda=100, m=7, b-a=(10.192-9.916)\cdot8/6=0.369$) and women's shot put ($\lambda=2.091, m=8, b-a=(19.884-17.156)\cdot9/7=3.508$) disciplines into the formula, the bounds would be 1.045 and 1.128, respectively. Note that these bounds are generous and the convergence is faster in practice.

We picked the geometric approximation to a scoring sequence on the following basis. Let $s_1,\ldots,s_m$ be an arbitrary sequence of scores, and $g_1(p),\ldots,g_m(p)$ the geometric sequence with parameter $p$. We normalise the scores so that $s_1=g_1(p)=1$ and $s_m=g_m(p)=0$. To choose $p$, we first suppose that all athletes are a priori equally strong. Then we fix a pair of athletes $a,b$. In a given event, $a$ finishes at position $j$ and $b$ at position $z$ with probability $1/m(m-1)$. Their score differences in this event are $s_j-s_z$ and $g_j(p)-g_z(p)$ respectively. Both are random variables and their difference has expectation 0, so we choose a $p$ that minimises the variance:
\begin{equation*}
    p=\argmin_{p'}\sum\limits_{j\neq z}(s_j-s_z-g_j(p')+g_z(p'))^2.
\end{equation*}

The argument above justifies a distance measure that we use to compare closeness between any pair of normalised scoring sequences:
\begin{equation}\label{distance2}
    d(s,t)=\sqrt{\frac{1}{4(m-2)}\sum\limits_{j\neq z}(s_j-s_z-t_j+t_z)^2},
\end{equation}
where the factor $1/(4(m-2))$ normalises the distance between plurality and antiplurality to~1.

In other words, for the geometric approximation to a normalised scoring sequence, we choose a parameter $p$ that minimises the distance to the scoring sequence. Similarly, for the optimal approximation to a scoring sequence, we choose a parameter $\lambda$ that minimises the distance to the scoring sequence.

Note that if we normalised the scoring sequences so that their sums are equal, $s_1+\ldots+s_m=t_1+\ldots+t_m=1$, and $s_m=t_m=0$, then we could have motivated another distance measure:
\begin{equation*}\label{distance3}
    r(s,t)=\sqrt{\frac{(m-1)}{(m-2)}\sum\limits_{j}(s_j-t_j)^2}.
\end{equation*}

\subsection{Choice of $\lambda$}\label{app:lambda}

We have briefly argued that a choice of $\lambda>1$ can be interpreted as the organiser valuing the possibility of exceptional performance more than consistency, while a $\lambda<1$ represents that an organiser is more concerned that an athlete never performs poorly in any given event.

If we think in terms of prize money rather than score, there is a more direct interpretation of $\lambda>1$: how much more is an organiser willing to pay an athlete whose performance is one cardinal unit higher? Thus in golf a $\lambda=1.4$ displays a willingness to pay an athlete who completes a score with one strike less, 1.4 times more prize money. In the hypothetical example of using $\lambda=100$ in men's 100m sprint (\autoref{OptimalGraphIAAF}), an athlete that finishes the event one second earlier will be rewarded one hundred times more. Note that despite the incredibly high choice of $\lambda$, the resulting scoring vector is not particularly convex -- this reflects the fact that one second is a very long time in this event, so valuing it by a factor of 100 is not as extreme as it may sound. In women's 200m (\autoref{OptimalGraphIAAF2}), by contrast, a $\lambda$ of only 4.72 produced a similar degree of convexity. The intuition here is that variance in the 200m event is roughly three times as high as in 100m, so the reward for a one second lead in 100m should be commensurate with the reward for a three second lead in 200m, and $100^1=4.64^3$.

The particular values of $\lambda$ in these examples were chosen by assuming that the men's 100m standard deviation (0.128) has the same level of ``quality'' as the women's 200m standard deviation (0.381), so we wanted $\lambda_m,\lambda_w$ such that $\lambda_m^{0.128}=\lambda_w^{0.381}$.

\subsection{Results}\label{app:results}

Here we include the figures that were omitted in the main text of the paper. The results for Individual races of the IBU World Cup biathlon in 2017/18, 2018/19 and 2019/20 seasons are presented in \autoref{OptimalGraphIBUappendix}. The results for twenty four athletic disciplines of the IAAF Diamond League in the 2010--2021 seasons are presented in \autoref{OptimalGraphIAAF}, \autoref{OptimalGraphIAAF2}, and \autoref{OptimalGraphIAAF3}. 

All data and calculations are available from the authors on request.

\begin{figure}
\begin{center}
\includegraphics[width=8cm]{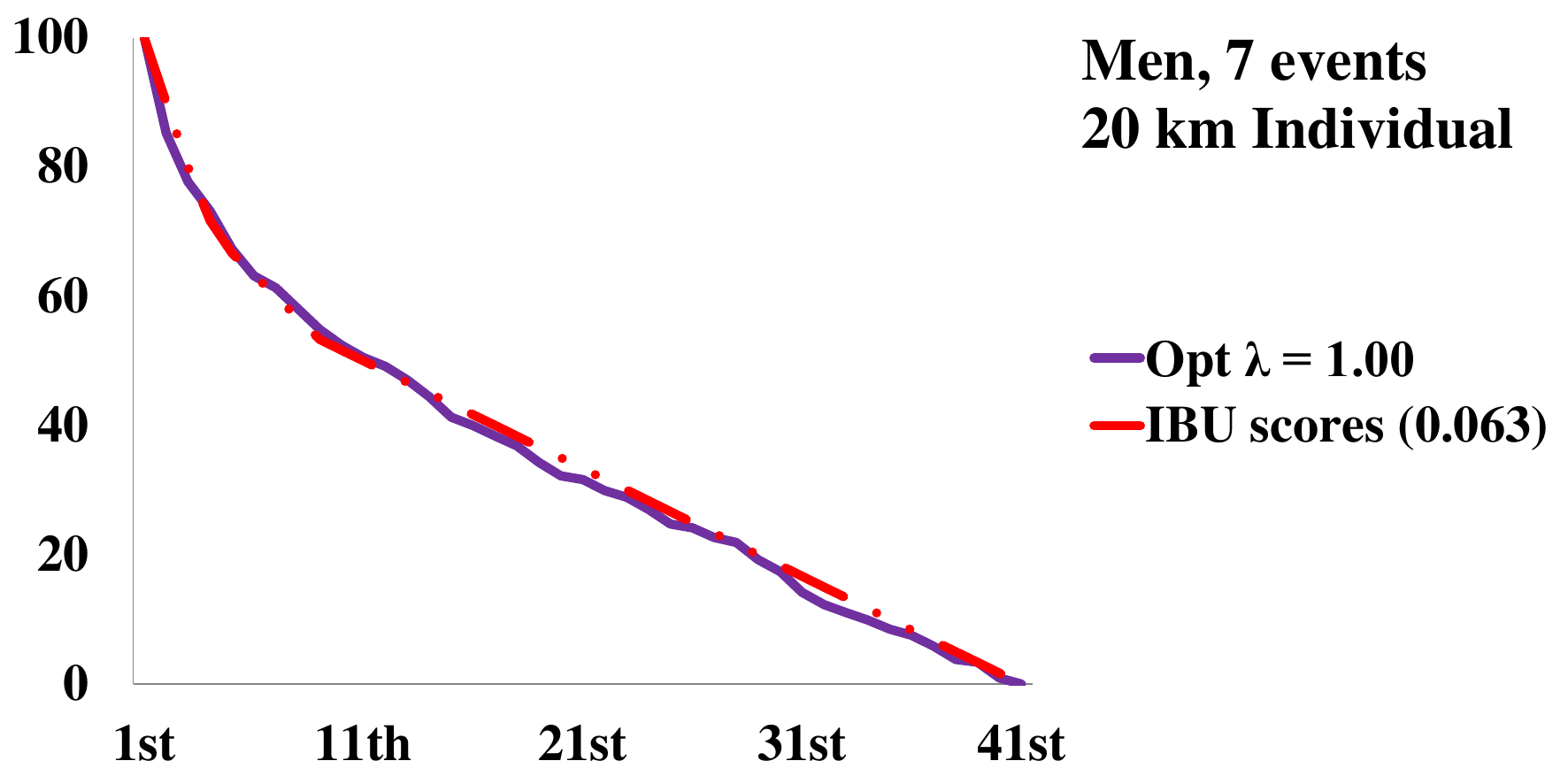}
\includegraphics[width=8cm]{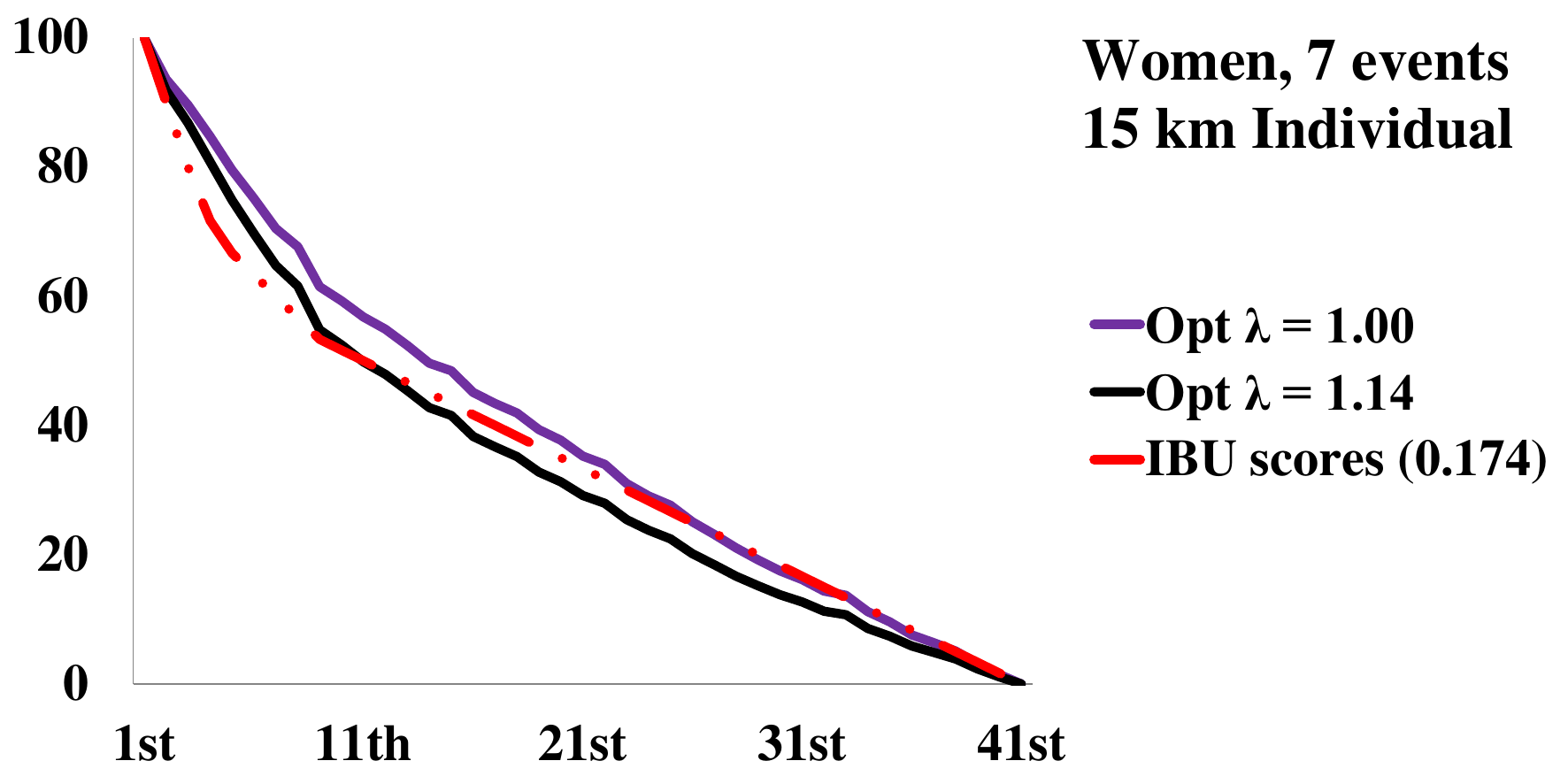}
\end{center}
\caption{Scores in IBU World Cup biathlon}
\label{OptimalGraphIBUappendix}
\vspace{0.2cm}
\justify
\footnotesize{\emph{Notes}: Scores used in 2017/18, 2018/19 and 2019/20 seasons compared with optimal scores. The $x$-axis is the position, the $y$-axis the normalised score. Scores for first position were normalised to 100, for forty-first position to~0. The optimal scores for $\lambda=1$ (purple solid, higher curve) and $\lambda=1.14$ (black solid, lower curve, performance measured in minutes) approximate the actual IBU scores used (red long dash two dots). The approximation distance is in brackets and calculated by formula~(\ref{distance2}).}
\end{figure}

\begin{figure}
\begin{center}
\includegraphics[width=8cm]{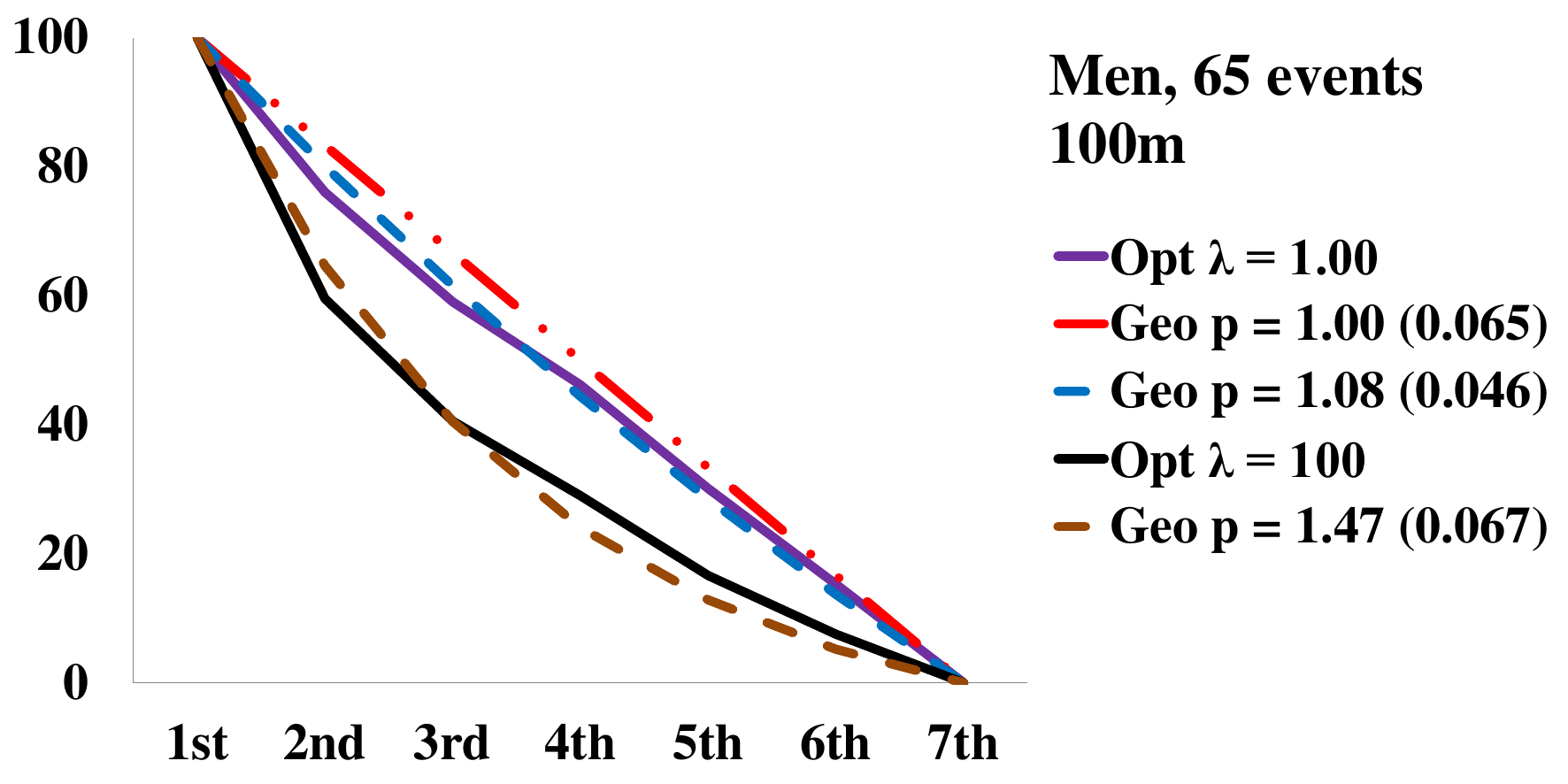}
\includegraphics[width=8cm]{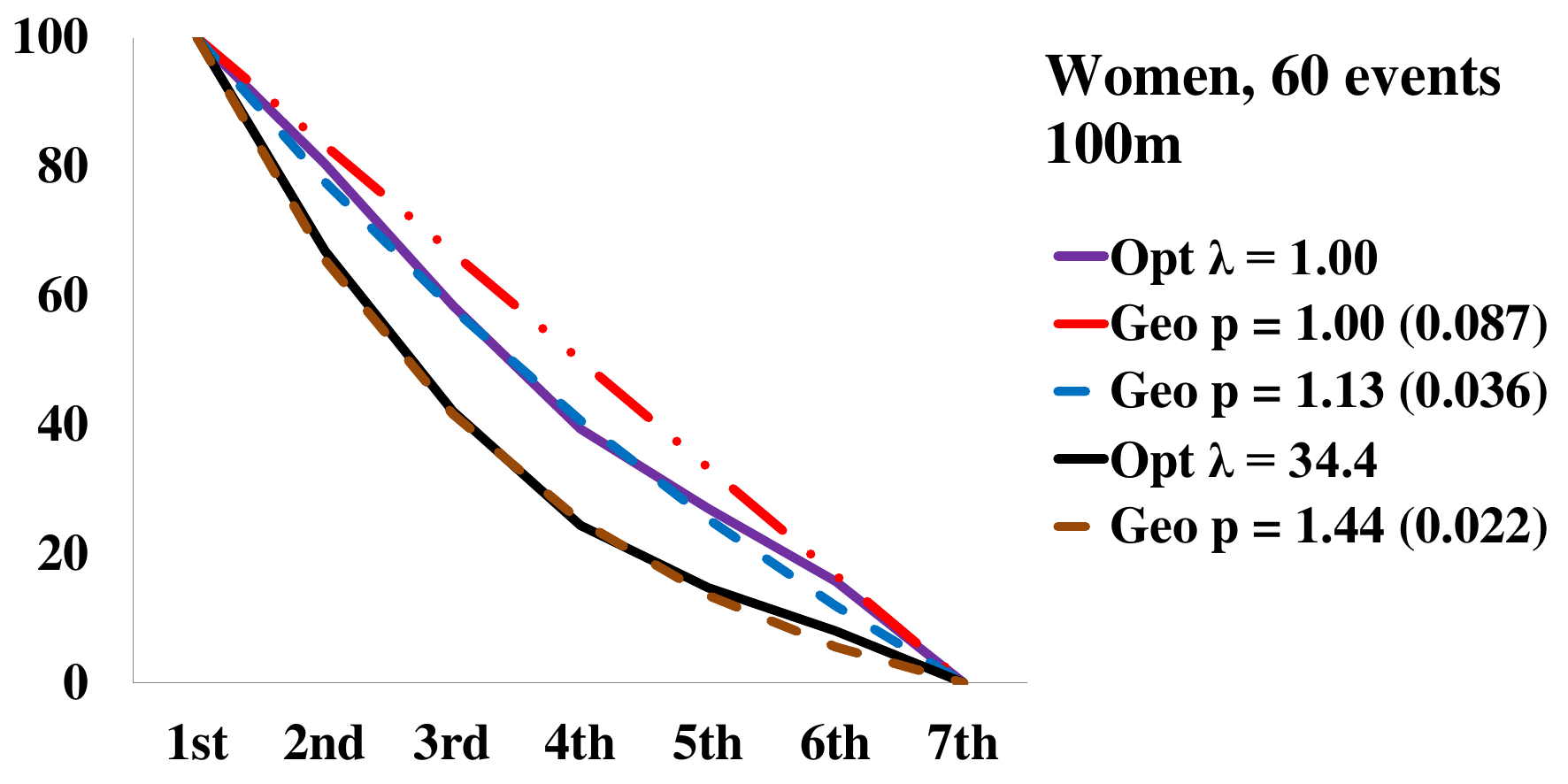}
\includegraphics[width=8cm]{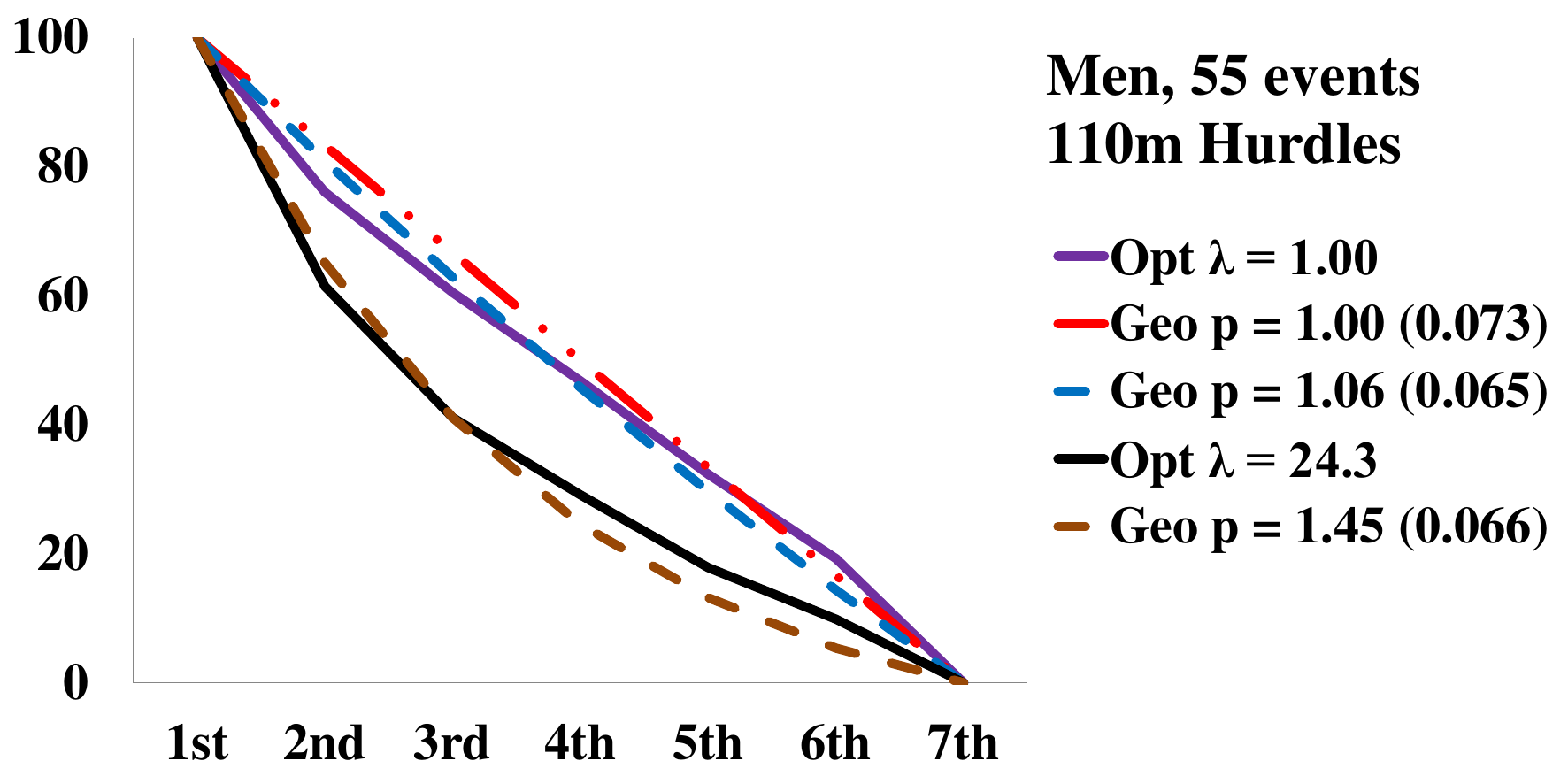}
\includegraphics[width=8cm]{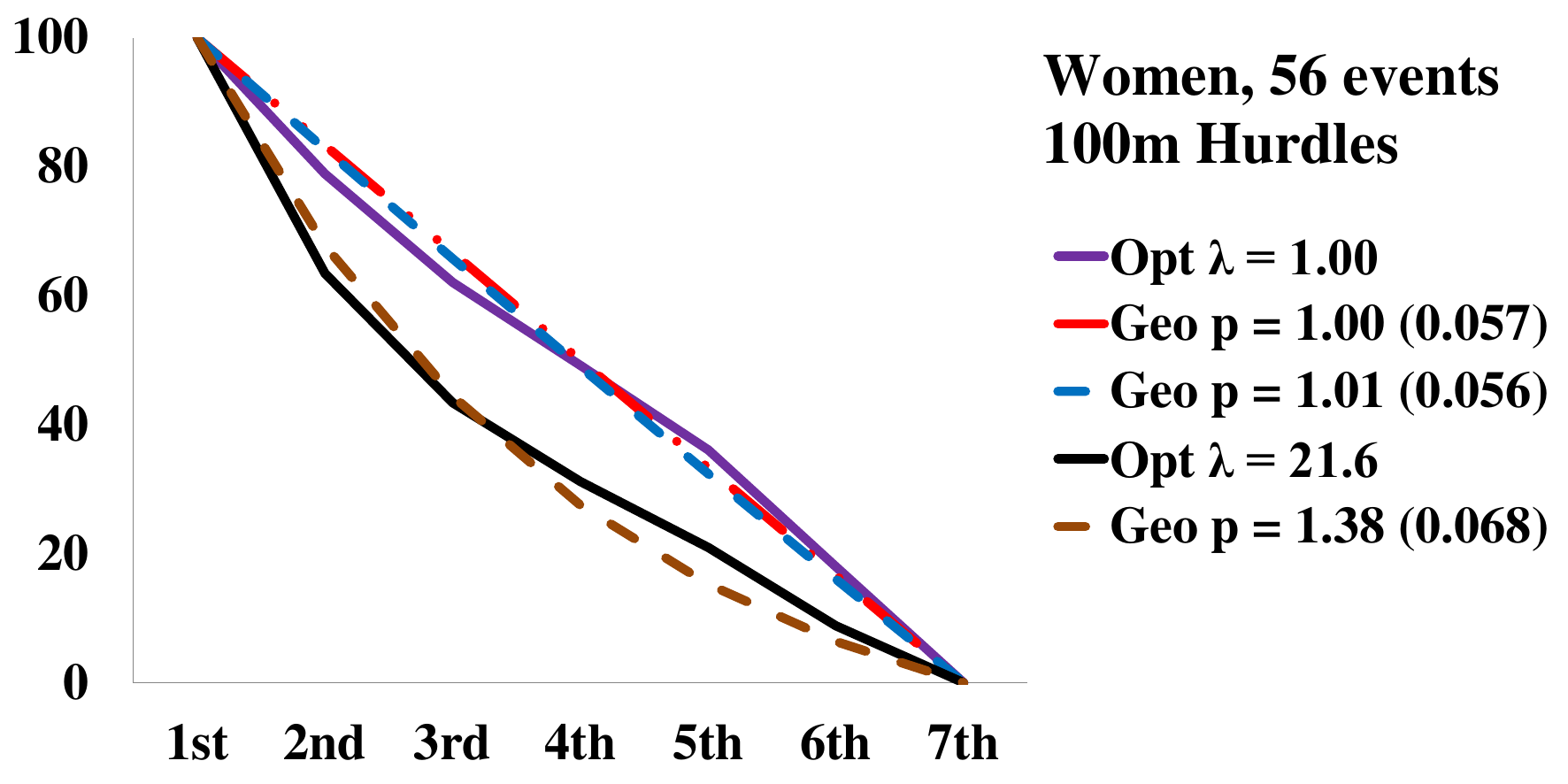}
\includegraphics[width=8cm]{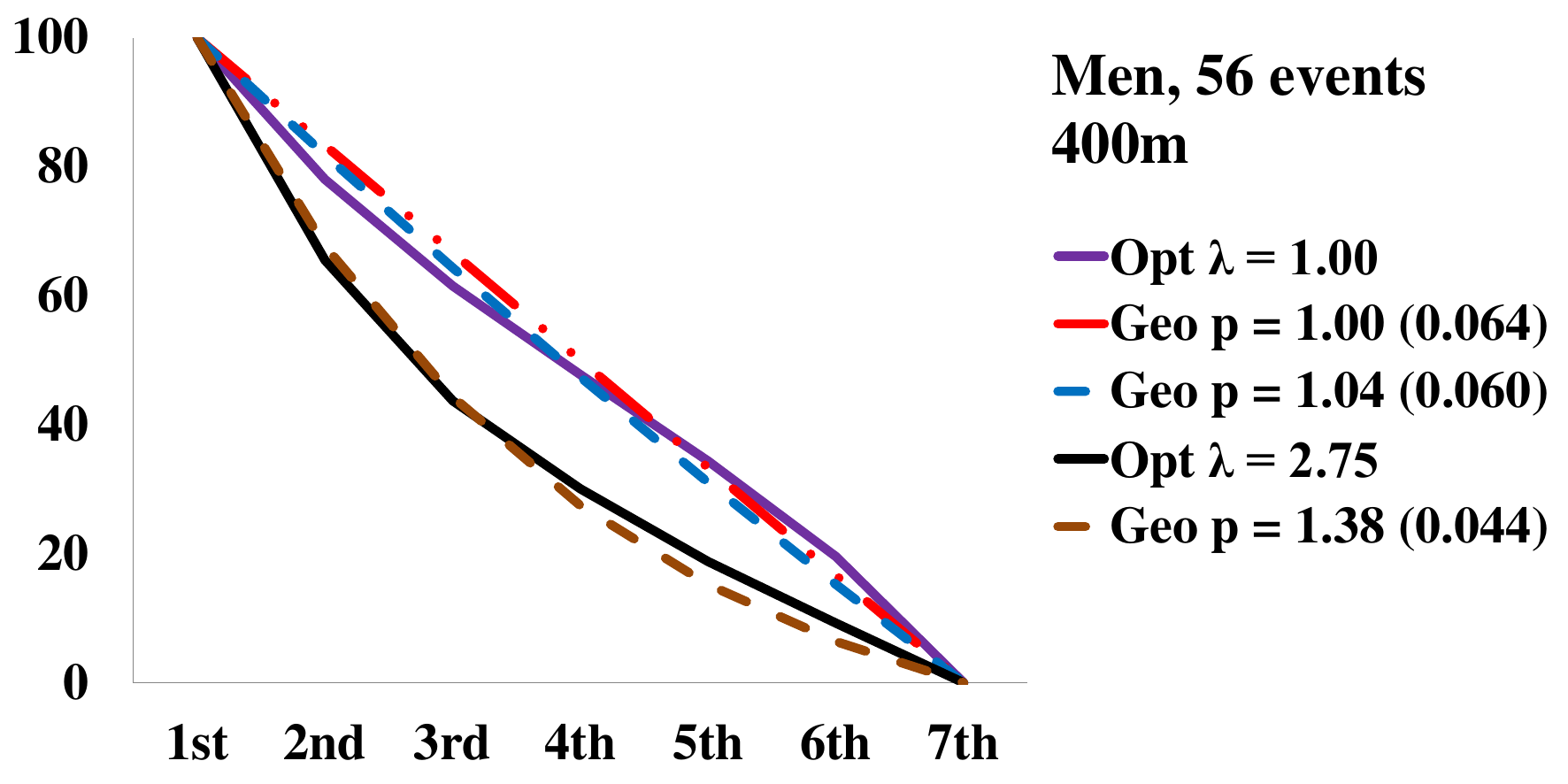}
\includegraphics[width=8cm]{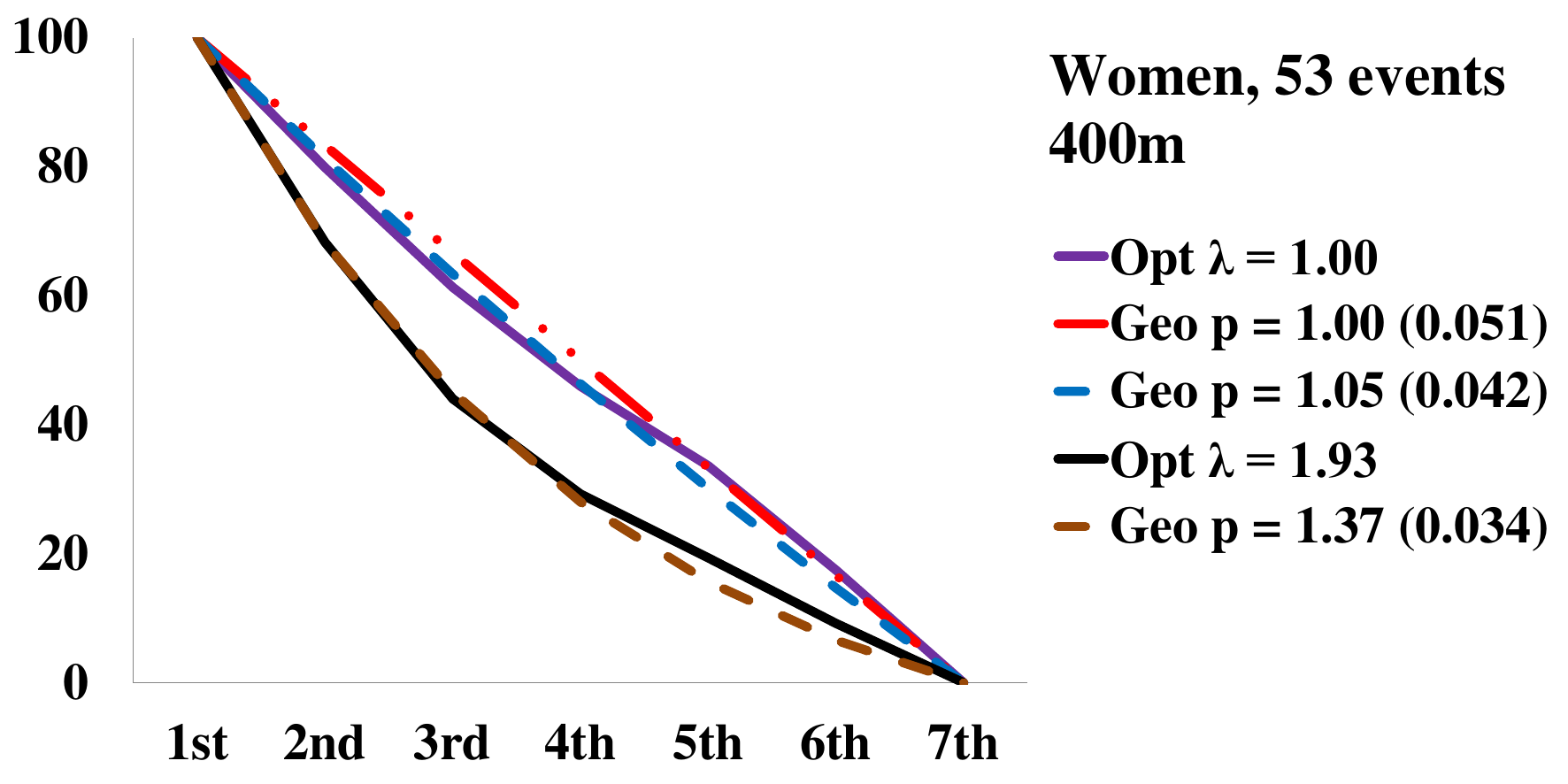}
\includegraphics[width=8cm]{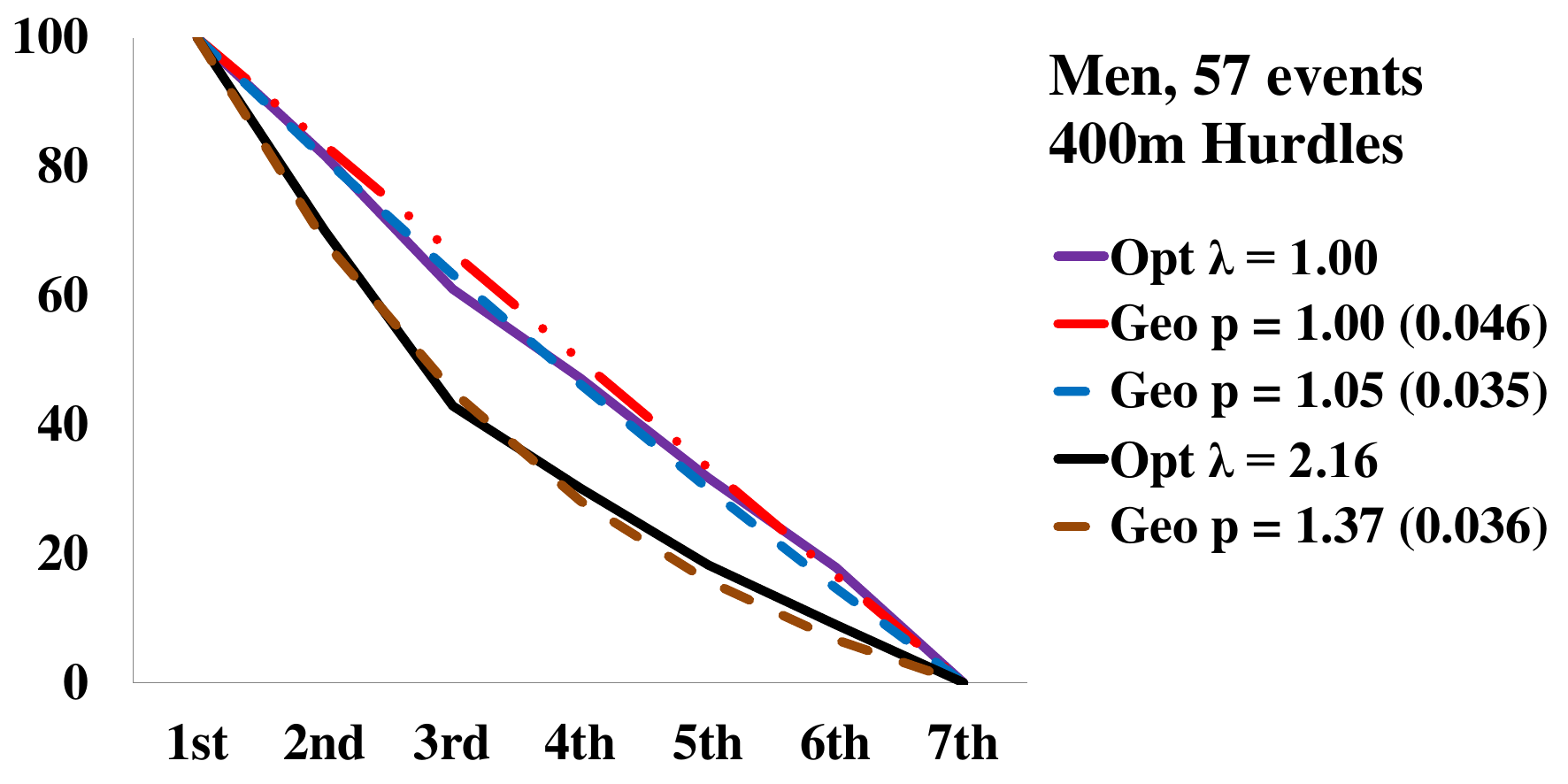}
\includegraphics[width=8cm]{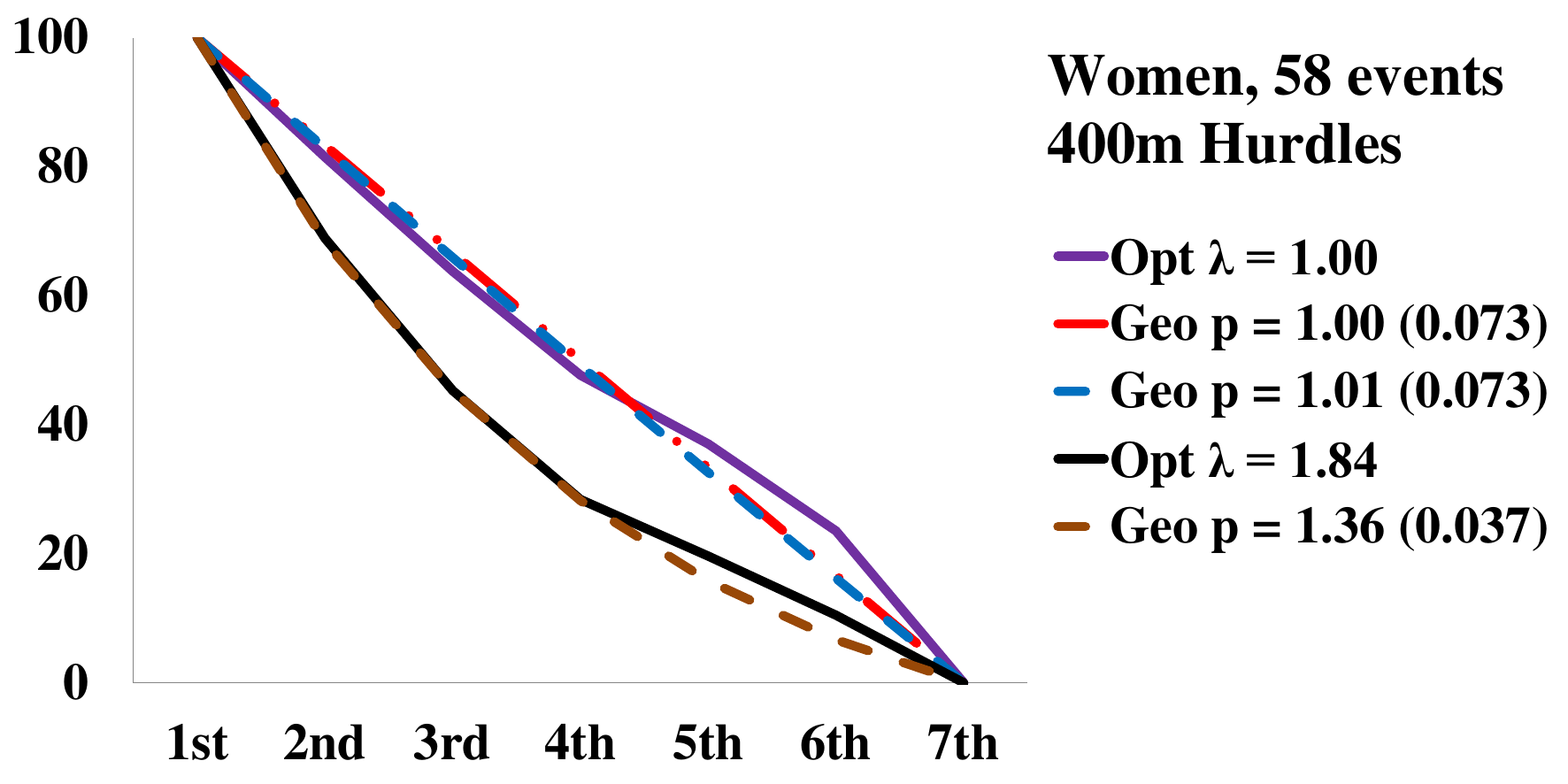}
\end{center}
\caption{Scores in IAAF Diamond League athletics}
\label{OptimalGraphIAAF}
\vspace{0.2cm}
\justify
\footnotesize{\emph{Notes}: The optimal scores  in 2010--2021 seasons approximated by geometric scores. The $x$-axis is the position, the $y$-axis the normalised score. Scores for first position were normalised to 100, for seventh position to~0. The eighth position is excluded to account for the discouragement effect in running. Observe that the actual Borda scores used since 2017 (geometric p~=~1, red long dash two dots) closely approximate most of the optimal scores for $\lambda=1$ (purple solid, higher curve). The curves for $\lambda>1$ (black solid, lower curve, performance measured in seconds) illustrate how closely other geometric scores (blue dash, higher curve, and brown dash, lower curve) can approximate the optimal scores on this data. The approximation distance is in brackets and calculated by formula~(\ref{distance2}), and denotes the distance to the first curve without brackets above the approximation in the legend.}
\end{figure}

\begin{figure}
\begin{center}
\includegraphics[width=8cm]{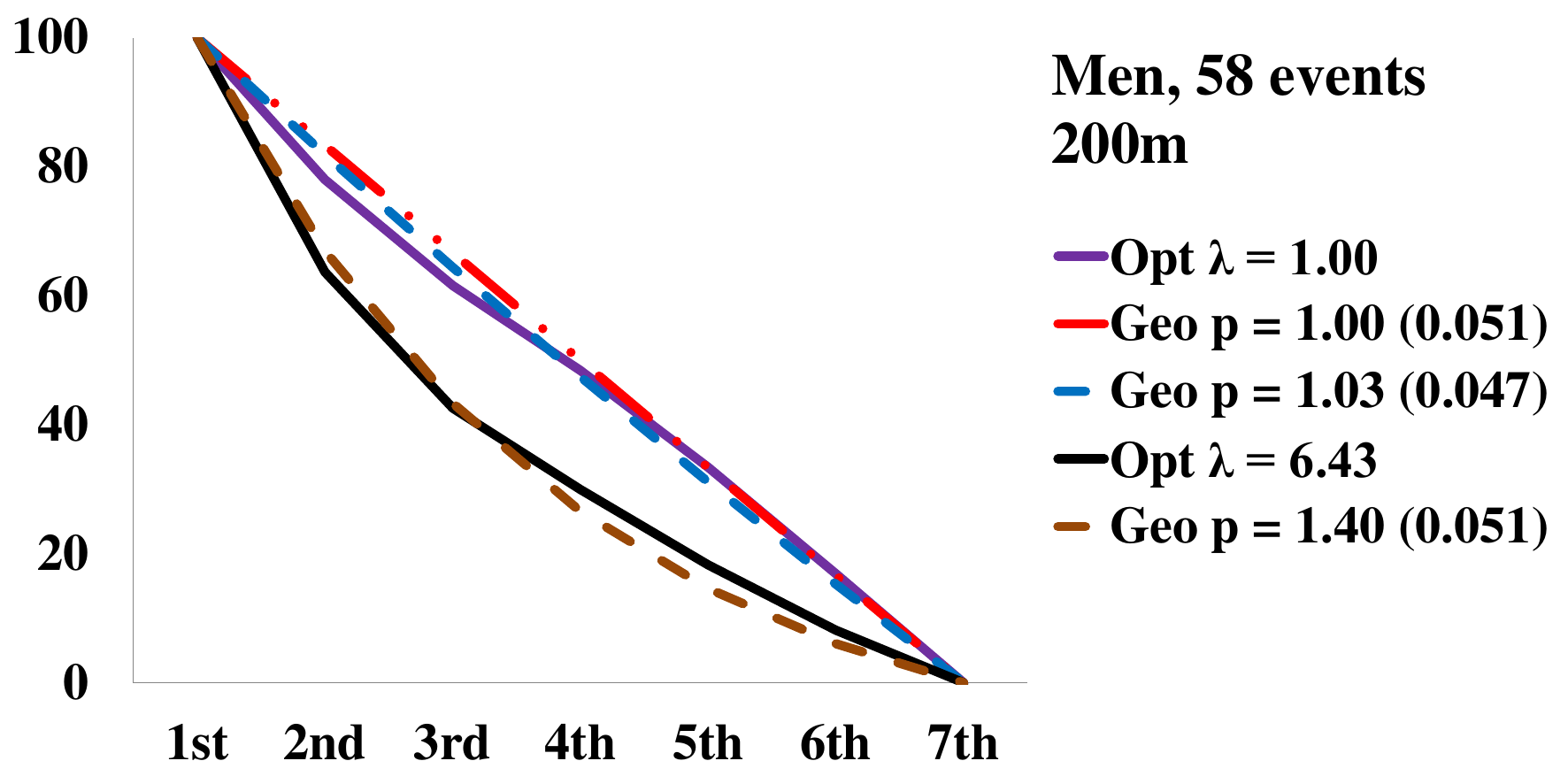}
\includegraphics[width=8cm]{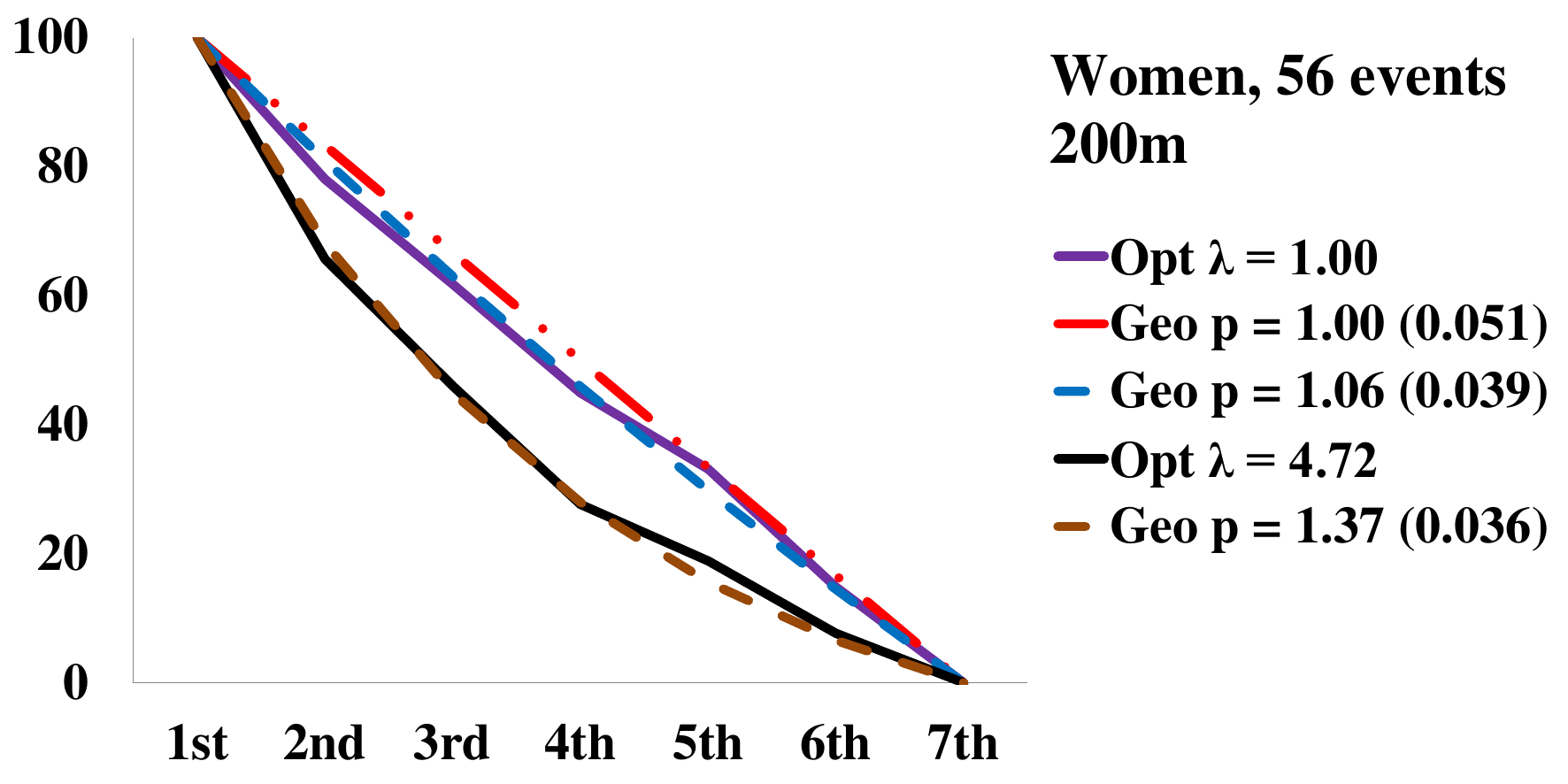}
\includegraphics[width=8cm]{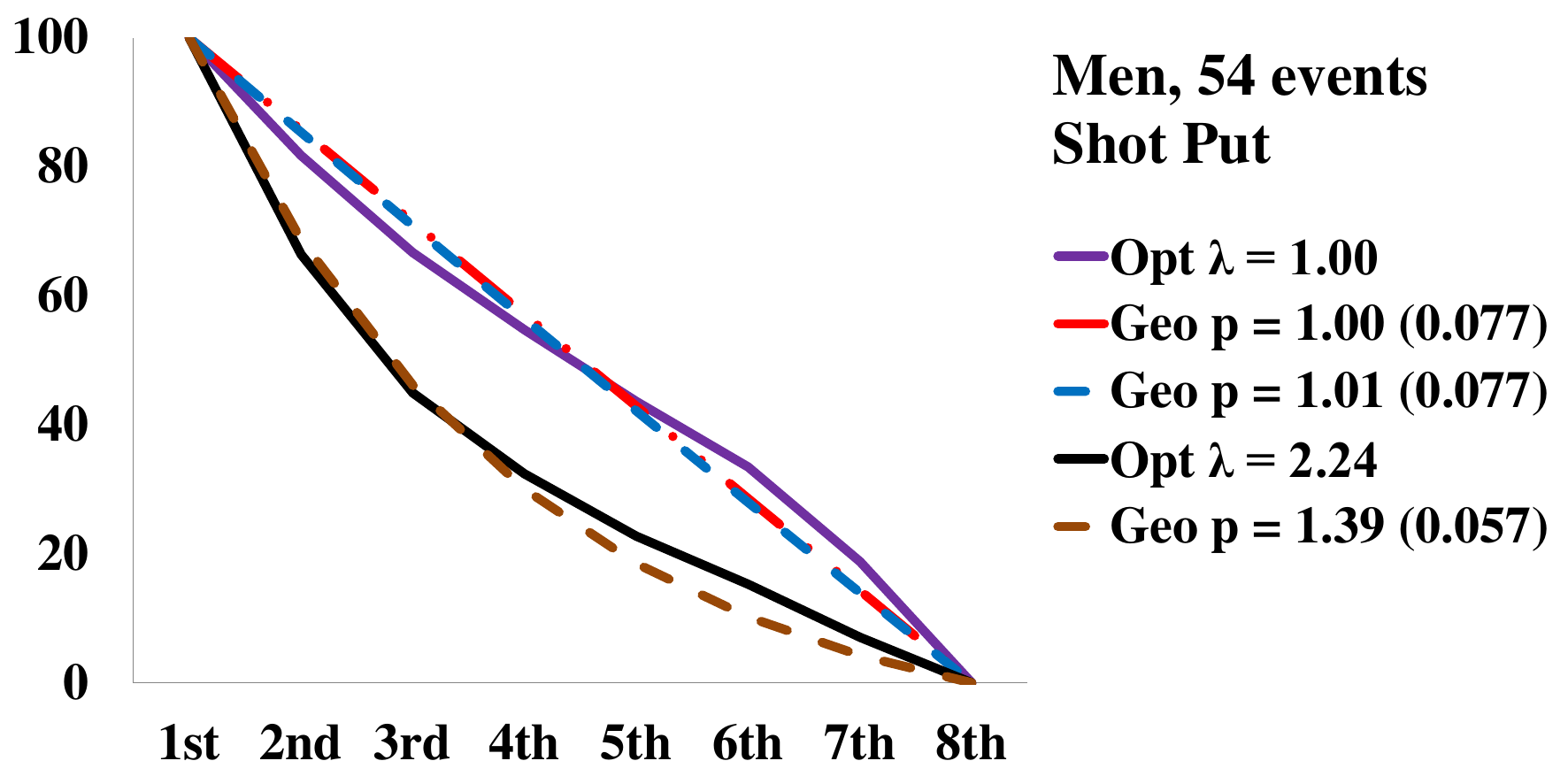}
\includegraphics[width=8cm]{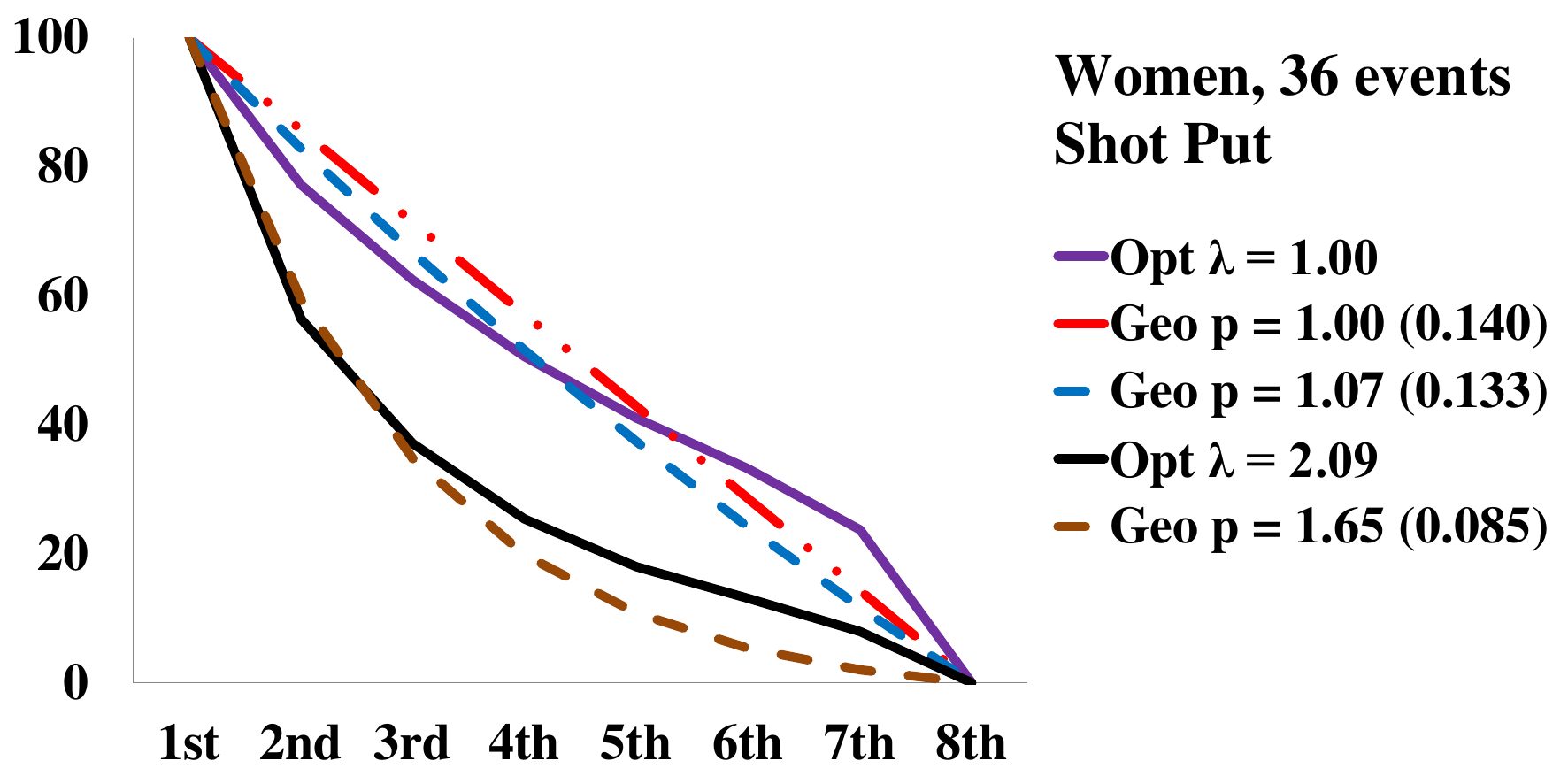}
\includegraphics[width=8cm]{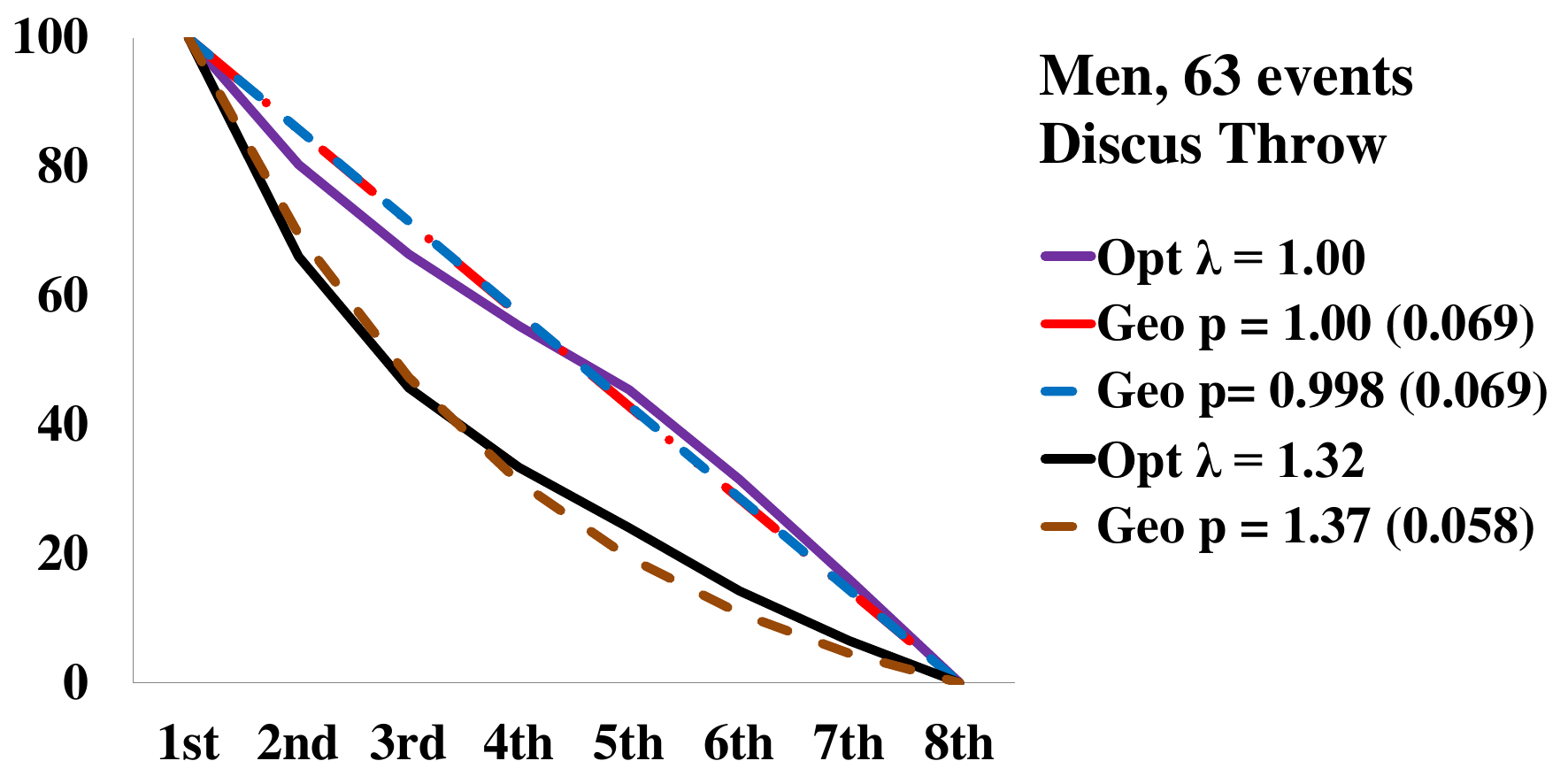}
\includegraphics[width=8cm]{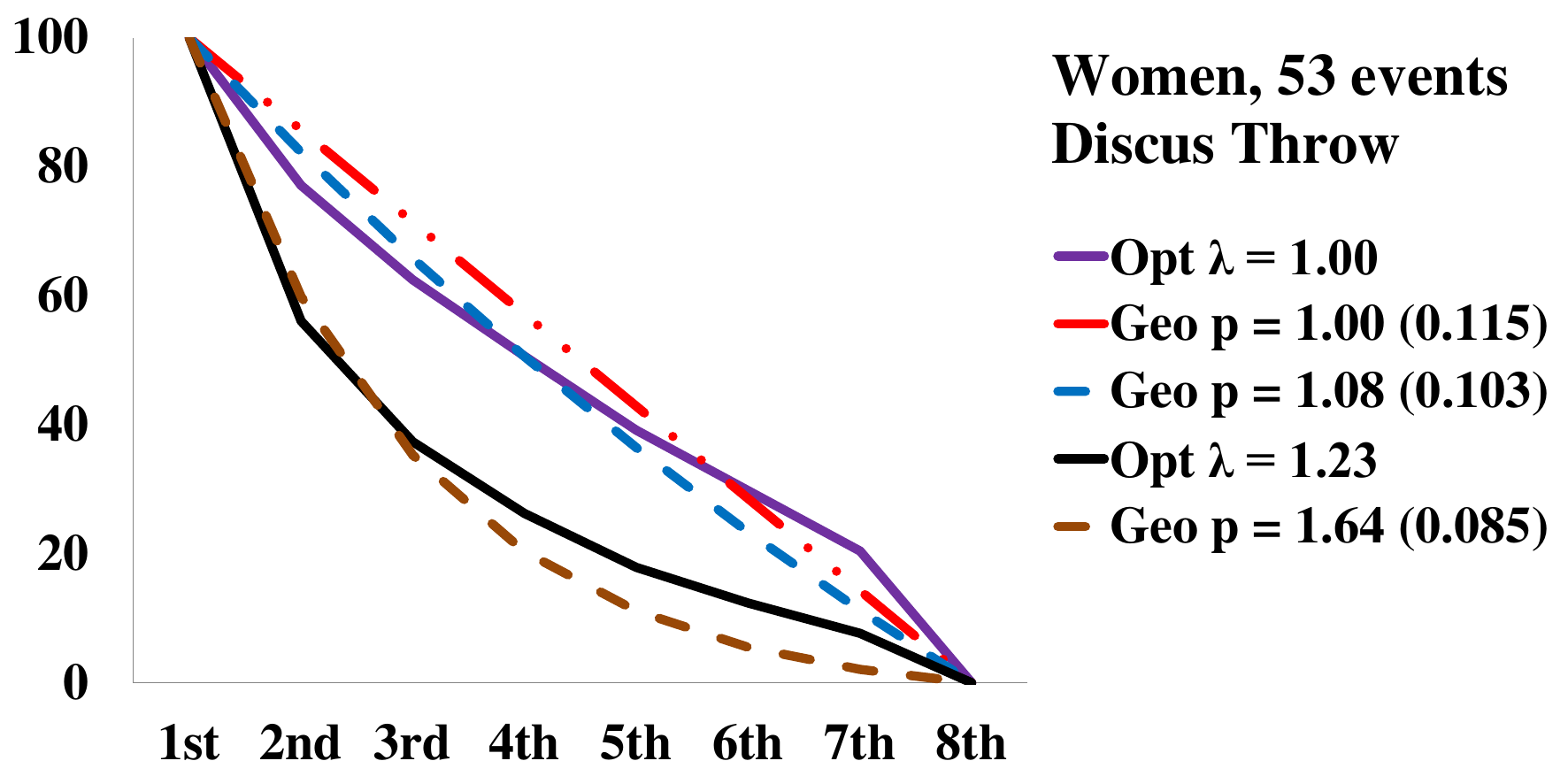}
\includegraphics[width=8cm]{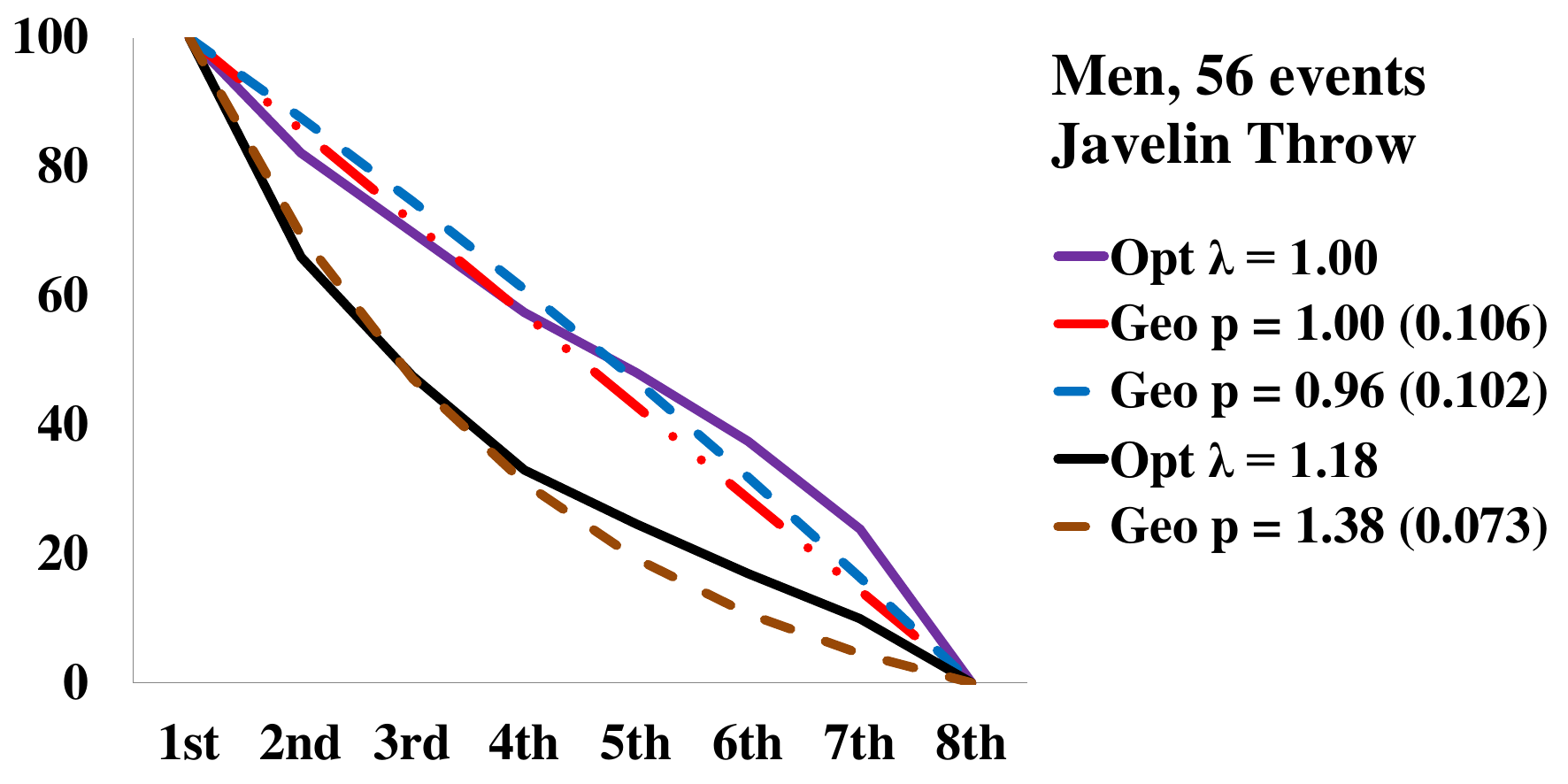}
\includegraphics[width=8cm]{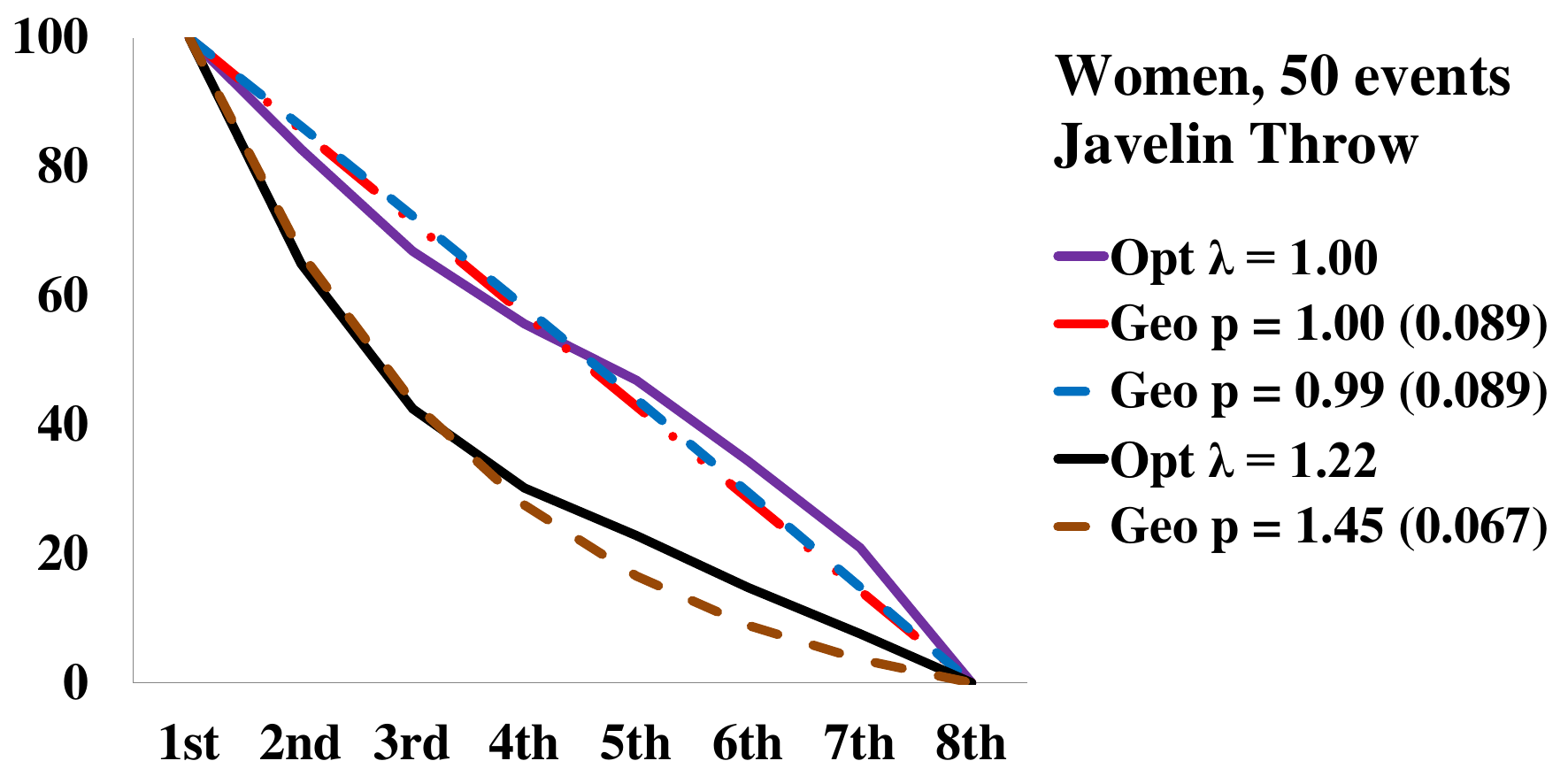}
\end{center}
\caption{Scores in IAAF Diamond League athletics}
\label{OptimalGraphIAAF2}
\vspace{0.2cm}
\justify
\footnotesize{\emph{Notes}: The optimal scores  in 2010--2021 seasons approximated by geometric scores. The $x$-axis is the position, the $y$-axis the normalised score. Scores for first position were normalised to 100, for seventh  (or eighth) position to~0. The eighth position is excluded to account for the discouragement effect in running. Observe that the actual Borda scores used since 2017 (geometric p~=~1, red long dash two dots) closely approximate most of the optimal scores for $\lambda=1$ (purple solid, higher curve). The curves for $\lambda>1$ (black solid, lower curve, performance measured in seconds and metres) illustrate how closely other geometric scores (blue dash, higher curve, and brown dash, lower curve) can approximate the optimal scores on this data. The approximation distance is in brackets and calculated by formula~(\ref{distance2}), and denotes the distance to the first curve without brackets above the approximation in the legend.}
\end{figure}

\begin{figure}
\begin{center}
\includegraphics[width=8cm]{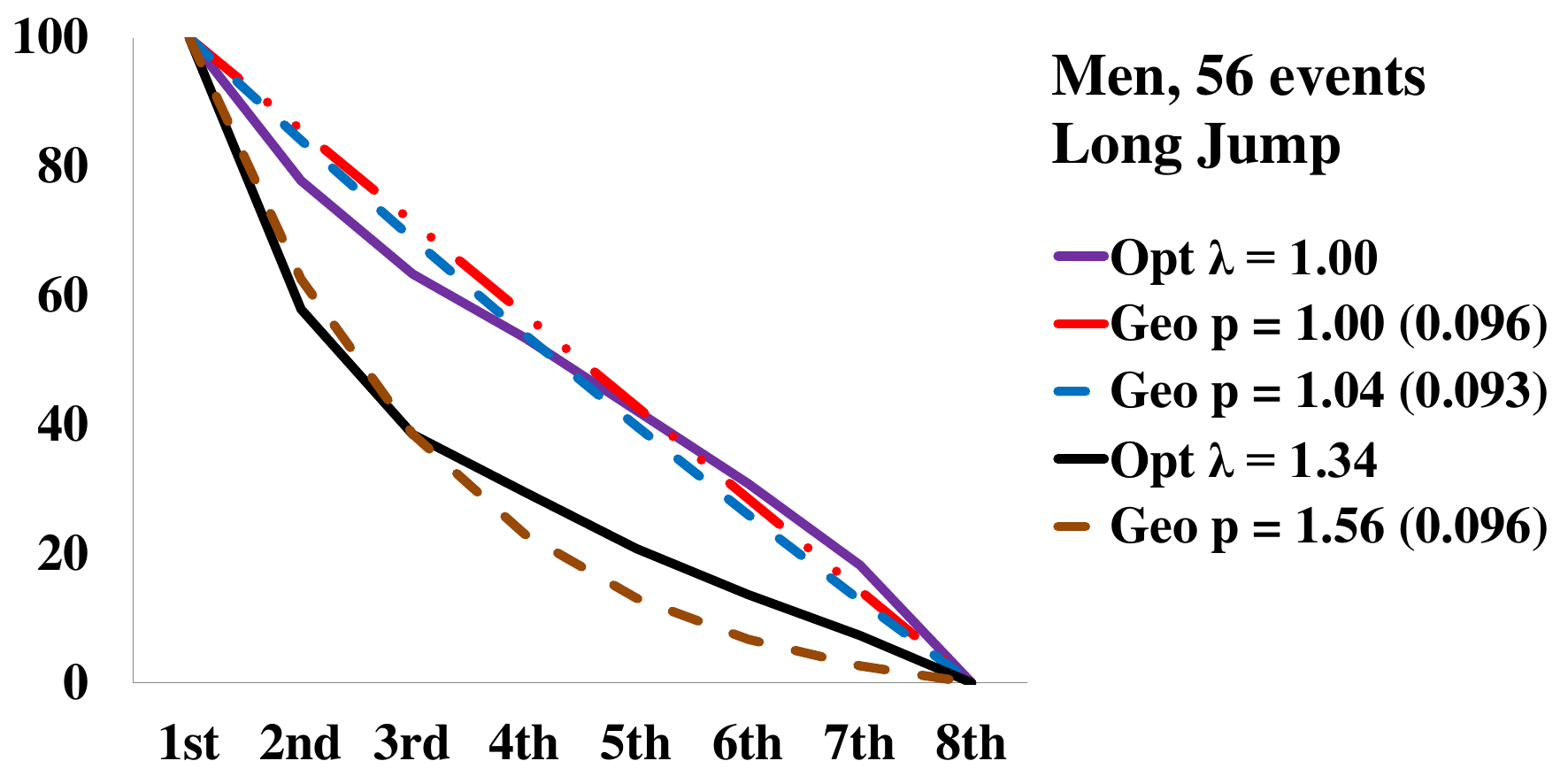}
\includegraphics[width=8cm]{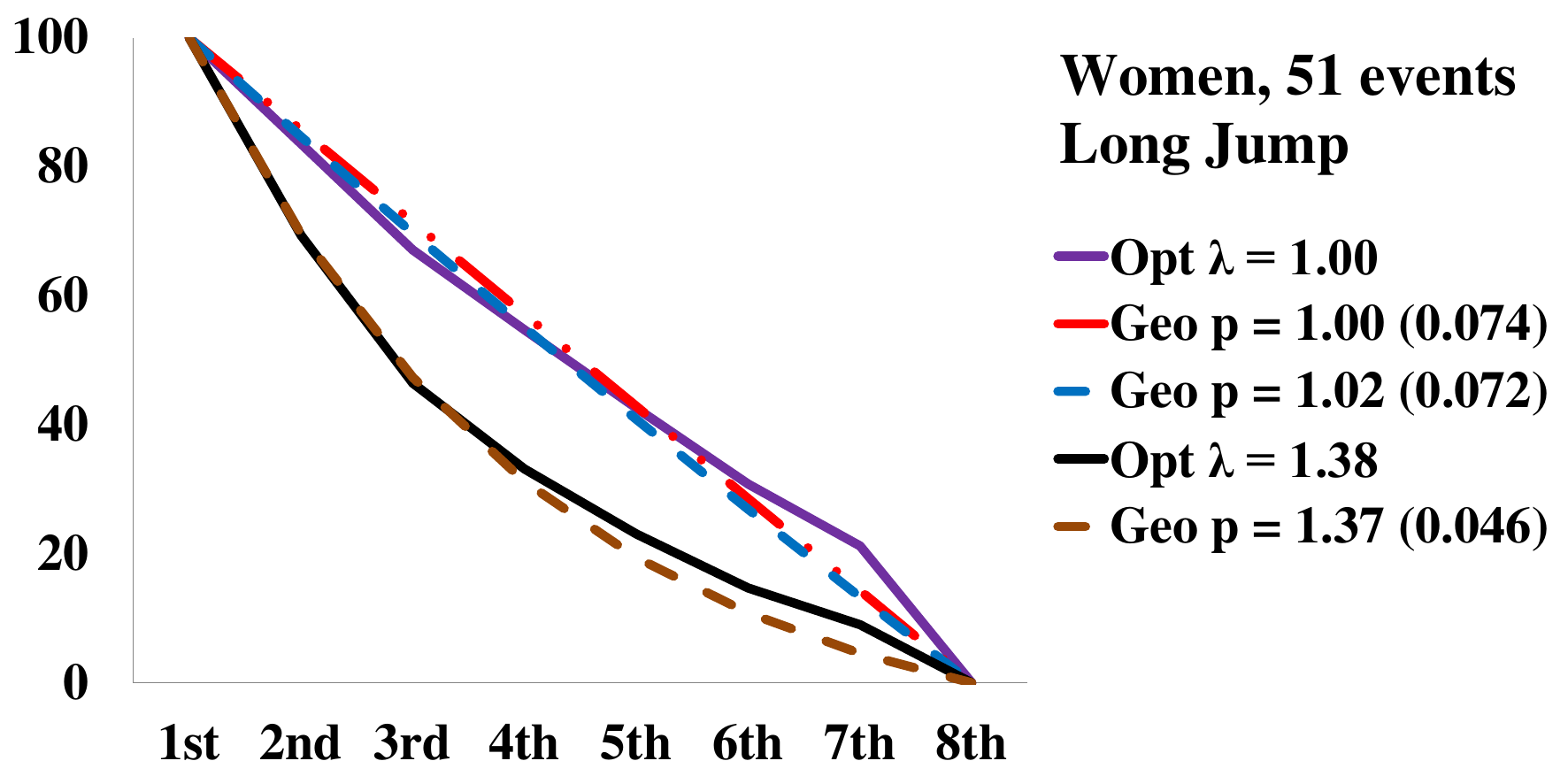}
\includegraphics[width=8cm]{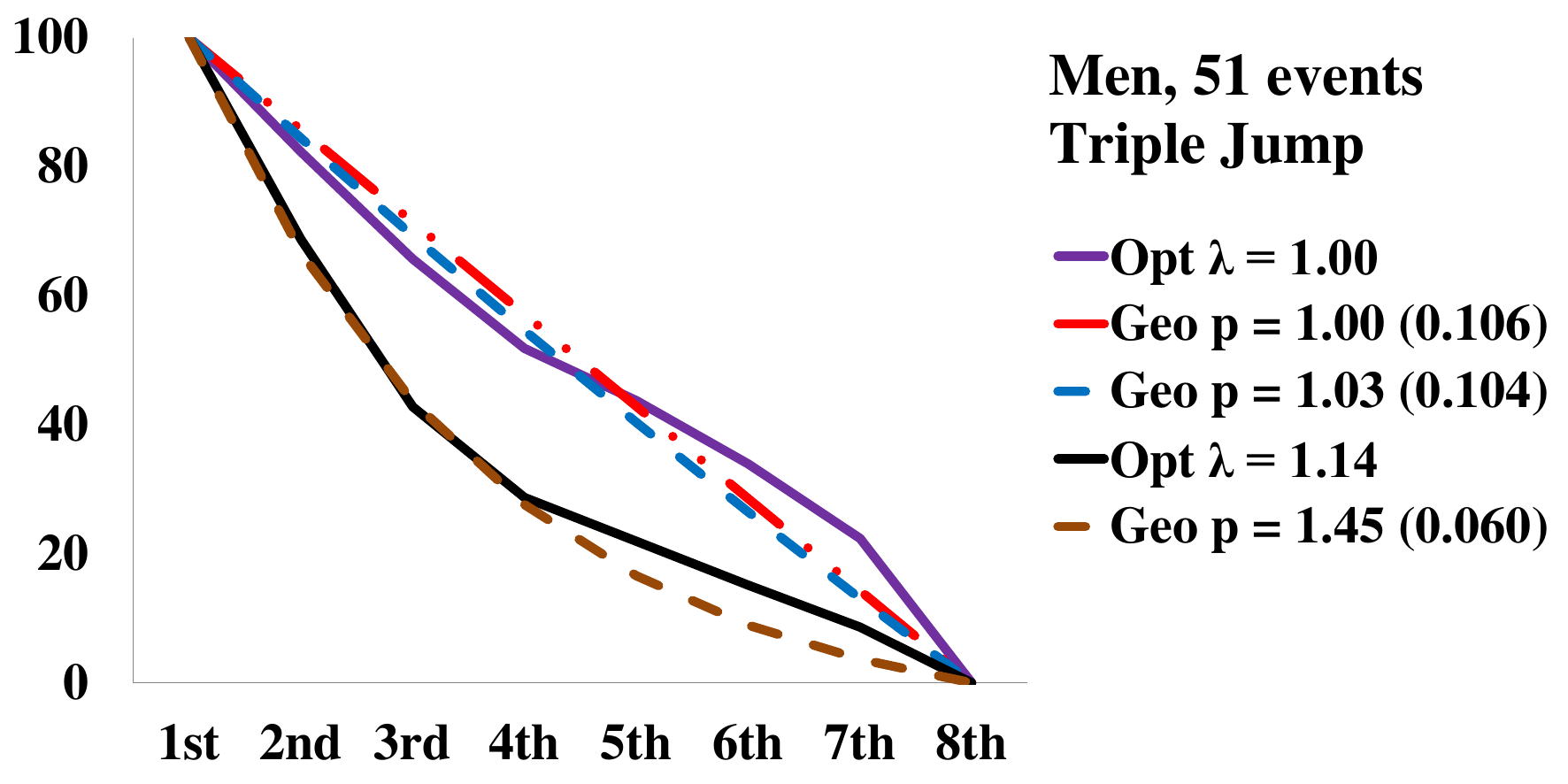}
\includegraphics[width=8cm]{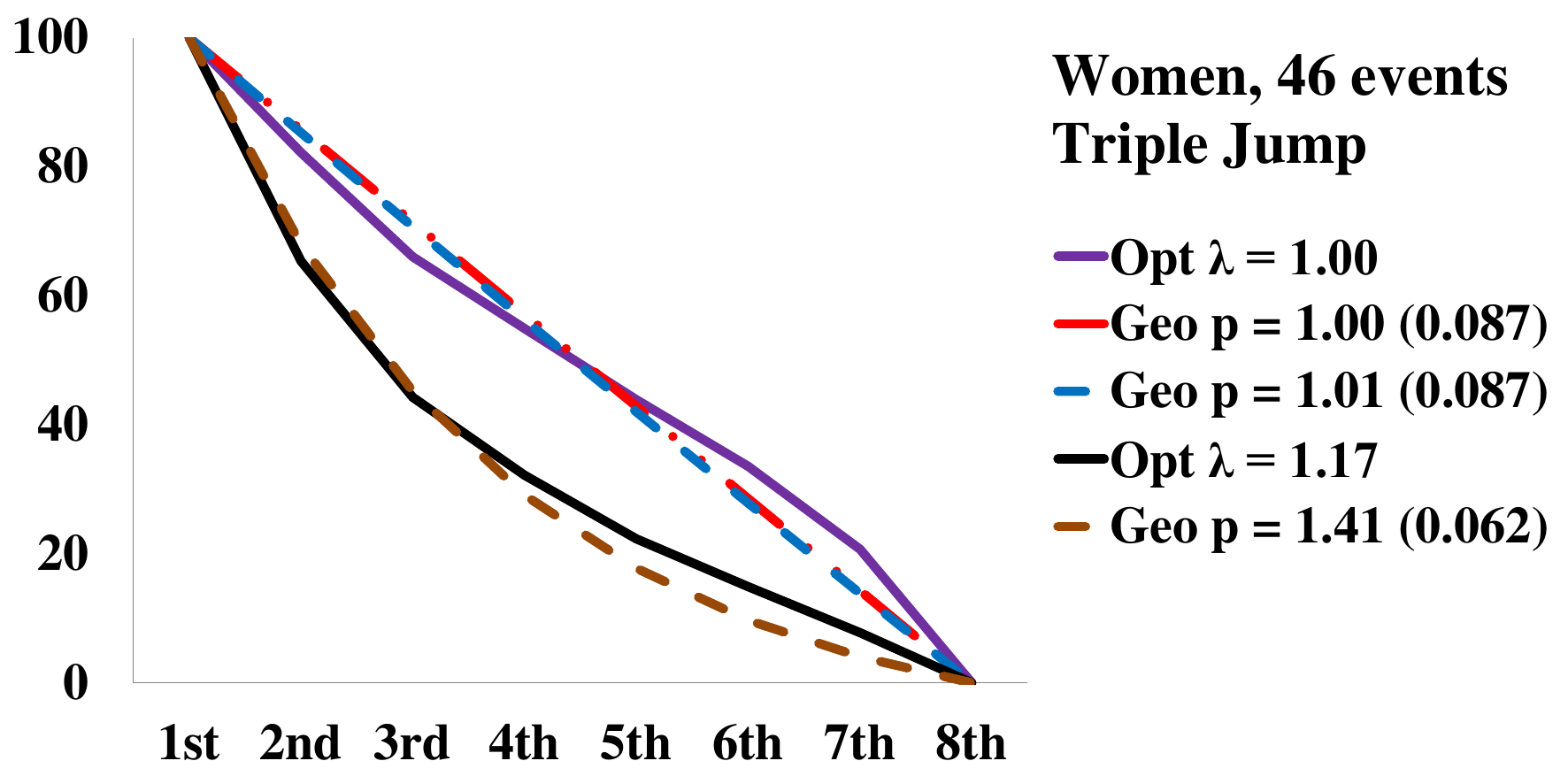}
\includegraphics[width=8cm]{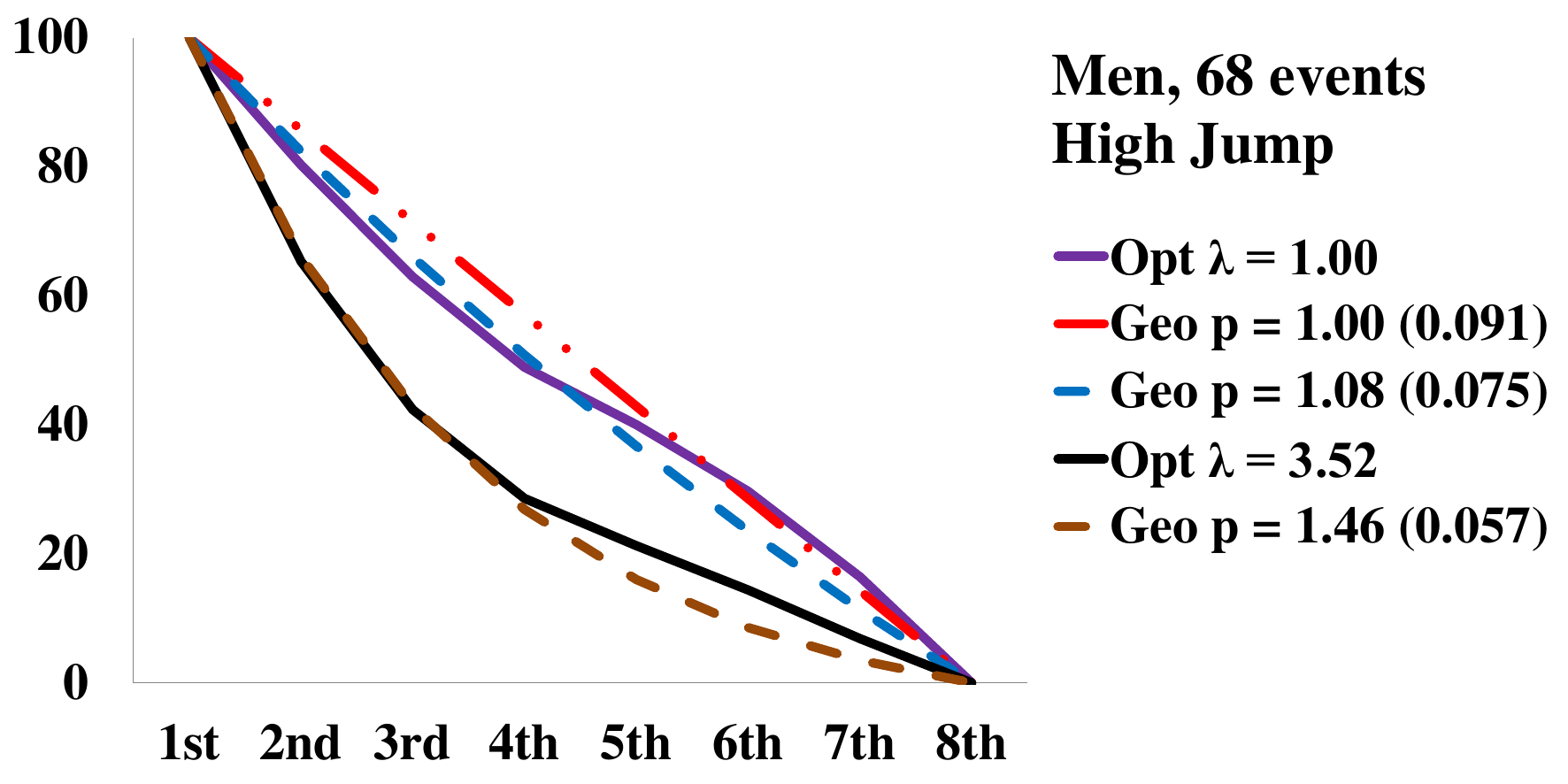}
\includegraphics[width=8cm]{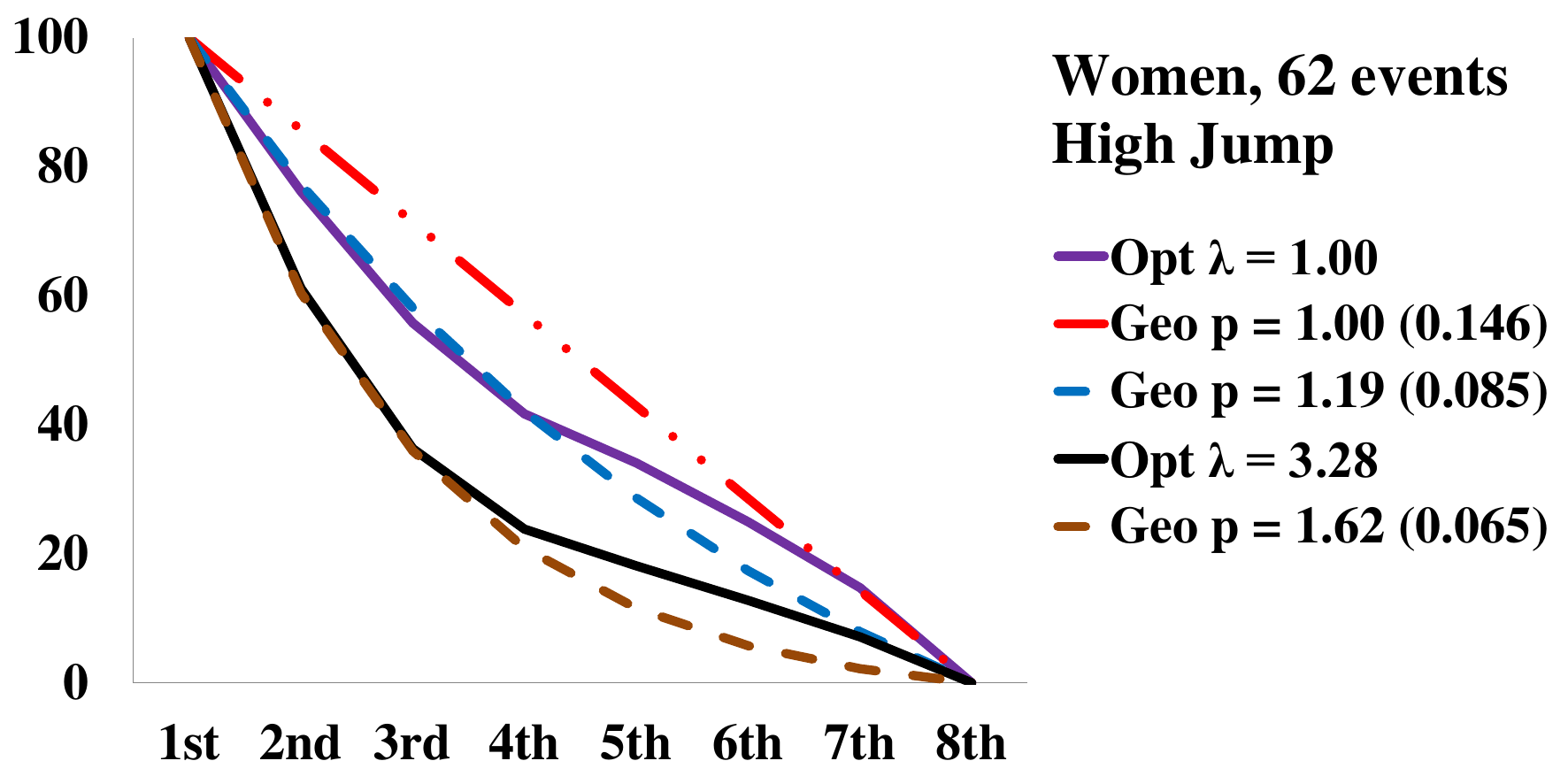}
\includegraphics[width=8cm]{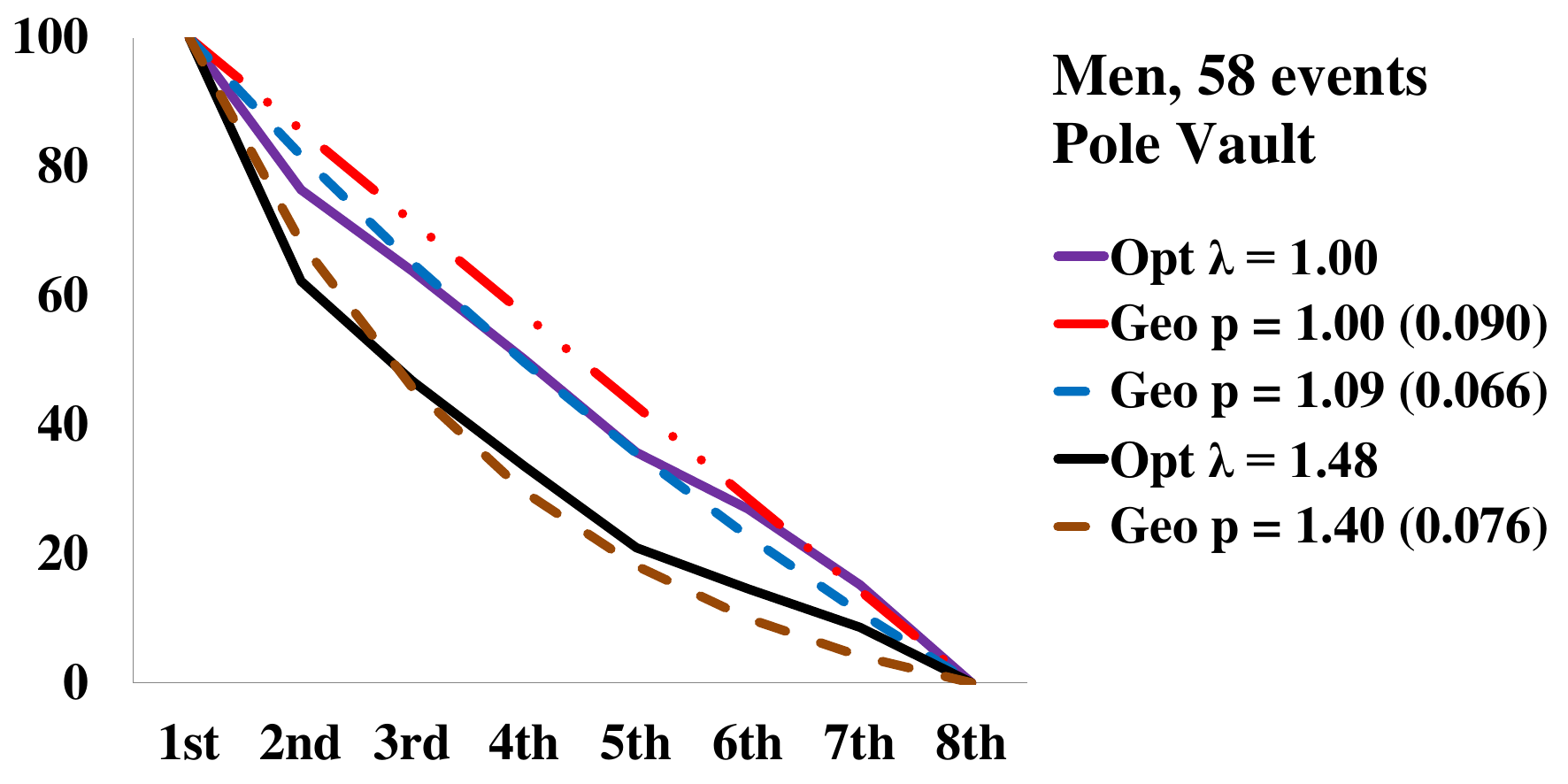}
\includegraphics[width=8cm]{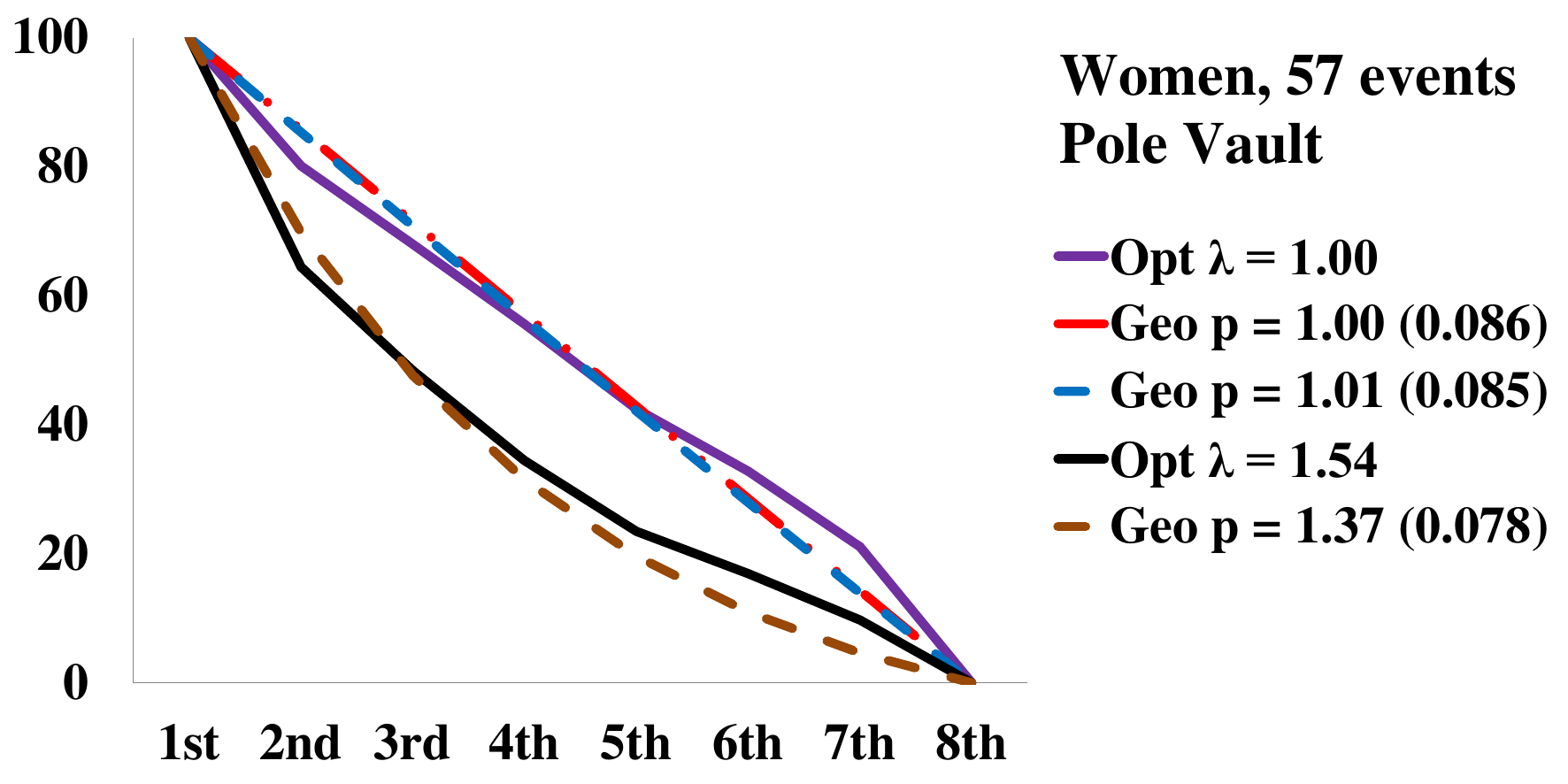}
\end{center}
\caption{Scores in IAAF Diamond League athletics}
\label{OptimalGraphIAAF3}
\vspace{0.2cm}
\justify
\footnotesize{\emph{Notes}: The optimal scores  in 2010--2021 seasons approximated by geometric scores. The $x$-axis is the position, the $y$-axis the normalised score. Scores for first position were normalised to 100, for eighth position to~0. Observe that the actual Borda scores used since 2017 (geometric p = 1, red long dash two dots) closely approximate most of the optimal scores for $\lambda=1$ (purple solid, higher curve). The curves for $\lambda>1$ (black solid, lower curve, performance measured in decimetres) illustrate how closely other geometric scores (blue dash, higher curve, and brown dash, lower curve) can approximate the optimal scores on this data. The approximation distance is in brackets and calculated by formula~(\ref{distance2}), and denotes the distance to the first curve without brackets above the approximation in the legend.}
\end{figure}

\clearpage

\bibliographystyle{apalike}
\bibliography{biathlon}

\end{document}